\documentclass[12pt]{article}
\usepackage{scicite}
\usepackage{times}
\usepackage[T1]{fontenc}
\usepackage{amsmath,amsfonts,amssymb}
\usepackage{amsthm}
\usepackage{tikz}
\usepackage{units}
\usepackage[margin=1in]{geometry}
\usepackage[colorlinks]{hyperref}
\usepackage{hypcap}
\usepackage{overpic}
\usepackage{subfigure}
\usetikzlibrary{decorations.markings}
\usepackage{setspace}

\newcommand{\f}[2]{\frac{#1}{#2}}
\newcommand{\Norm}[1]{\left\|{#1}\right\|}
\newcommand{\norm}[1]{\|{#1}\|}

\DeclareMathOperator{\poly}{poly}

\DeclareMathOperator{\Sym}{Sym}
\DeclareMathOperator{\adj}{adj}

\renewcommand{\natural}{\mathbb{N}}
\newcommand{\integer}{\mathbb{Z}}
\newcommand{\II}{\mathbb{I}}

\newcommand{\ket}[1]{|{#1}\rangle}
\newcommand{\bra}[1]{\langle{#1}|}
\newcommand{\braket}[2]{\langle{#1}|{#2}\rangle}

\newcommand{\CP}{\mathrm{CP}}
\newcommand{\CD}{\mathrm{C}\theta}
\newcommand{\XX}{\mathrm{X}}
\newcommand{\CZ}{\mathrm{CZ}}
\newcommand{\CNOT}{\mathrm{CNOT}}

\newcommand{\Had}{\mathrm{H}}
\DeclareMathOperator{\T}{T}
\DeclareMathOperator{\Bas}{B}

\renewcommand{\H}{\mathcal{H}}
\newcommand{\K}{\mathcal{K}}
\renewcommand{\O}{\mathcal{O}}
\newcommand{\med}{\mathrm{m}}

\theoremstyle{plain}
\newtheorem{theorem}{Theorem}
\newtheorem{lemma}{Lemma}
\newtheorem{prop}{Proposition}

\newcommand{\eq}[1]{\hyperref[eq:#1]{(\ref*{eq:#1})}}
\renewcommand{\sec}[1]{\hyperref[sec:#1]{Appendix~\ref*{sec:#1}}}
\newcommand{\app}[1]{\hyperref[app:#1]{Appendix~\ref*{app:#1}}}
\newcommand{\fig}[1]{\hyperref[fig:#1]{Figure~\ref*{fig:#1}}}
\newcommand{\thm}[1]{\hyperref[thm:#1]{Theorem~\ref*{thm:#1}}}
\newcommand{\lem}[1]{\hyperref[lem:#1]{Lemma~\ref*{lem:#1}}}
\newcommand{\pro}[1]{\hyperref[pro:#1]{Proposition~\ref*{pro:#1}}}

\title{Universal computation by multi-particle quantum walk}

\author
{Andrew M. Childs,$^{1,2}$ David Gosset,$^{1,2}$ Zak Webb$^{2,3}$\\
\\
\normalsize{$^1$ Department of Combinatorics \& Optimization, University of Waterloo}\\
\normalsize{$^2$ Institute for Quantum Computing, University of Waterloo}\\
\normalsize{$^3$ Department of Physics \& Astronomy, University of Waterloo}\\
\\
}
\date{}

\begin{document}
\baselineskip24pt
\maketitle

\begin{abstract} 
A quantum walk is a time-homogeneous quantum-mechanical process on a graph defined by analogy to classical random walk. The quantum walker is a particle that moves from a given vertex to adjacent vertices in quantum superposition. Here we consider a generalization of quantum walk to systems with more than one walker. A continuous-time multi-particle quantum walk  is generated by a time-independent Hamiltonian with a term corresponding to a single-particle quantum walk for each particle, along with an interaction term. Multi-particle quantum walk includes a broad class of interacting many-body systems such as the Bose-Hubbard model and systems of fermions or distinguishable particles with nearest-neighbor interactions. We show that multi-particle quantum walk is capable of universal quantum computation. Since it is also possible to efficiently simulate a multi-particle quantum walk of the type we consider using a universal quantum computer, this model exactly captures the power of quantum computation. In principle our construction could be used as an architecture for building a scalable quantum computer with no need for time-dependent control. 
\end{abstract}
\singlespace
\section*{Introduction}

Quantum walk is a versatile and intuitive framework for developing quantum algorithms.  Applications of continuous- \cite{FG98} and discrete-time \cite{AAKV01,ABNVW01} models of quantum walk  include an example of exponential speedup over classical computation \cite{CCDFGS03} and optimal algorithms for element distinctness \cite{Amb07} and formula evaluation \cite{FGG08}.

Quantum walk can also be viewed as a model of computation. From this perspective it is natural to ask which quantum computations can be performed efficiently by quantum walk. This question was answered in reference \cite{Chi09} where it was shown that continuous-time quantum walk on sparse unweighted graphs is equivalent in power to the quantum circuit model.

While this universality construction shows that quantum walk is a powerful computational model, it does not directly provide an architecture for a quantum computer. Since each vertex of the underlying graph corresponds to a basis state in the Hilbert space, the graph associated with a quantum computation on $n$ qubits is exponentially large as a function of $n$. This means that this quantum walk cannot be efficiently implemented using an architecture where each vertex of the graph occupies a different spatial location. Although the quantum walk on any sufficiently sparse graph can be efficiently simulated by a universal quantum computer \cite{AT03,BACS05,Chi10,CCDFGS03}, it may be no easier to implement this quantum walk than to perform a general quantum computation. 

Nevertheless, many experimental implementations of quantum walk have been carried out.  Typically, such experiments rely on non-scalable encodings of the quantum walk into a physical system such that locality of the graph translates into locality of the Hamiltonian, but with substantial overhead that precludes the possibility of dramatic computational speedup (see for example \cite{DSBE05,KFC09,PLP08}).  While such experiments may serve as useful testbeds of presently available quantum information processors, they cannot directly realize efficient universal quantum computation.

In this paper we consider a natural generalization of quantum walk.  Whereas a standard quantum walk concerns a single walker moving (in superposition) on a graph, we consider the generalization to many interacting walkers. We show that such a multi-particle quantum walk is universal for quantum computation. Specifically, we show that any $n$-qubit circuit with $g$ gates can be simulated by the dynamics of $\O(n)$ particles that interact for a time $\poly(n,g)$ on an unweighted planar graph of maximum degree $4$ with $\poly(n,g)$ vertices. Our construction is based on a discrete version of scattering theory and implements single-qubit gates similarly to the single-particle universality construction \cite{Chi09}.  However, we use a different encoding of quantum data, and we implement two-qubit gates via interactions between pairs of particles.  For indistinguishable particles, almost any interaction gives rise to universality.  We present explicit universal constructions based on the Bose-Hubbard model, fermions with nearest-neighbor interactions, and distinguishable particles with nearest-neighbor interactions.

To prove that our construction works, we develop several tools for analyzing multi-particle scattering on graphs. We prove error bounds for the propagation of one- and two-particle wave packets,  and we prove a truncation lemma formalizing the idea that a particle moving at a fixed speed only ``sees'' part of the graph in a fixed time interval. Using these tools we prove that multi-particle quantum walk can efficiently simulate a standard quantum computer, establishing the computational power of multi-particle quantum walk. This result also provides limitations on the classical simulation of many-body interacting systems.   For example, assuming quantum computers are more powerful than classical ones, our work implies that the dynamics of the Bose-Hubbard model on a sparse, planar graph cannot be efficiently simulated on a classical computer. 

Prior to this work it was known that systems of interacting particles on a lattice can be used to perform universal computation by changing the Hamiltonian as a function of time~\cite{IZ02,mizel07}. Because of the time dependence, such a system is not a quantum walk.  In contrast, in our scheme the Hamiltonian is time-independent and the computation is encoded entirely in the graph on which the particles interact.

 Multi-particle quantum walk has previously been considered as an algorithmic tool for solving graph isomorphism (see for example \cite{gamble2010}). However, this technique has known limitations (see for example \cite{smith2011}). Other previous work on multi-particle quantum walk has focused on two-particle quantum walk \cite{sheridan2006,pathak2007,bromberg2009,PLM10,OBB11,lahini2012,sansoni2012,schreiber2012,ahlbrecht2012} and multi-particle quantum walk without interactions \cite{sheridan2006,pathak2007,bromberg2009,PLM10,OBB11,Rohde11,sansoni2012}. Here we consider multi-particle quantum walks with interactions, which seem to be required in order to achieve efficient computational universality. Although non-interacting bosonic quantum walks may be difficult to simulate classically \cite{AA11}, such systems are probably not capable of performing universal quantum computation.  For fermionic systems the situation is even clearer: non-interacting fermions can be efficiently simulated on a classical computer~\cite{TD02}.

We hope that our work motivates further experimental investigations of  multi-particle quantum walk. Because our graphs are exponentially smaller (as a function of $n$) than those used in the single-particle construction, the multi-particle quantum walks we describe can be efficiently implemented using an architecture where vertices of the graph are represented by devices at different spatial locations (although we have not addressed issues of fault tolerance). The two-particle bosonic quantum walk experiments of references \cite{bromberg2009,PLM10,OBB11,sansoni2012} can be viewed as a first step toward implementing our construction.  However, these experiments only involve non-interacting particles; as we discussed earlier, a nontrivial interaction appears necessary for universality. Conversely, our work shows that almost any interaction can be used to perform efficient universal computation.  A Bose-Hubbard model of the type we consider could naturally be realized in a variety of experimental systems, including traditional nonlinear optics \cite{CY95}, neutral atoms in optical lattices \cite{BCDJ99} (with a lattice pattern implemented using a quantum gas microscope \cite{BGPFG09}), or photons in arrays of superconducting qubits \cite{HSS11}.  

\section*{Multi-particle quantum walk}

In a multi-particle quantum walk, the particles interact in a local manner on a given simple graph $G$ with vertex set $V(G)$ and edge set $E(G)$.  We consider quantum walks with distinguishable or indistinguishable particles, where in the latter case the particles can be either bosons or fermions.

The Hilbert space for $m$ distinguishable particles on $G$ is spanned by the basis 
\begin{equation}
\{\ket{i_1,\ldots,i_m} \colon i_1,\ldots,i_m\in V(G)\}\label{eq:firstq}
\end{equation}
where $i_w$ is the location of the $w$th particle. A continuous-time multi-particle quantum walk of $m$ distinguishable particles on $G$ is generated by a time-independent Hamiltonian 
\begin{equation}
H_G^{(m)}=\sum_{w=1}^m \sum_{(i,j)\in E(G)} \bigg(|i\rangle\langle j|_w+|j\rangle\langle i|_w\bigg)+ \sum_{i,j\in V(G)} \mathcal{U}_{ij}(\hat{n}_i,\hat{n}_j)\label{eq:dist_ham}
\end{equation}
where the subscript $w$ indicates that the operator acts on the location register for the $w$th particle (tensored with the identity on all other particles). Here $\mathcal{U}_{ij} (\hat{n}_i,\hat{n}_j)$ is a function of the number operators $\hat{n}_i$  and $\hat{n}_j$ that count the numbers of  particles located at vertices $i$ and $j$, respectively (explicitly, $\hat{n}_i=\sum_{w=1}^m \ket{i}\bra{i}_w$).

The first term of \eq{dist_ham} moves particles between adjacent sites, while the second term is an interaction between particles. We only consider interaction terms acting between two or more particles, so we assume that $\mathcal{U}_{ij}$ is zero whenever one of its arguments evaluates to zero and that $\mathcal{U}_{ii}$ is zero if there is only one particle at vertex $i$.  We also assume that the interaction $\mathcal{U}_{ij}$ has a constant range $C\in\{0,1,2,\ldots\}$ (i.e., $\mathcal{U}_{ij}=0$ whenever the shortest path between vertices $i$ and $j$ has more than $C$ edges).  In our universality construction we consider graphs $G$ that include long paths. Each vertex in one of these paths has degree $2$ in the full graph $G$, i.e., the path is only connected at either end to other vertices in $G$. We require that for vertices $i$ and $j$ on such a path, $\mathcal{U}_{ij}$ depends only on the distance between $i$ and $j$ (a kind of translation invariance).  Finally, we assume that the norm of each term $\mathcal{U}_{ij}$ is upper bounded by a polynomial in $m$.

Note that for any $\mathcal{U}$ satisfying our assumptions, the Hamiltonian \eq{dist_ham} for a single particle reduces to
\begin{equation}
H_G^{(1)}=\sum_{(i,j)\in E(G)} |i\rangle\langle j|+|j\rangle\langle i|, \label{eq:single_particle_ham}
\end{equation}
the Hamiltonian for a standard continuous-time quantum walk (namely, the adjacency matrix of $G$).

States representing $m$ indistinguishable particles can be represented in the basis \eq{firstq} as states that are either symmetric (if the particles are bosons) or antisymmetric (if the particles are fermions) under the interchange of any two particles.  Since the Hamiltonian \eq{dist_ham} is symmetric, it preserves both the symmetric and antisymmetric subspaces. Restricted to the appropriate subspace, the Hamiltonian \eq{dist_ham} generates a quantum walk of $m$ bosons or $m$ fermions on $G$. 

This framework includes several well-known interacting many-body systems defined on graphs. For example, it includes the Bose-Hubbard model, where the interaction term is $\mathcal{U}_{ij}(\hat{n}_i,\hat{n}_j)=  (U/2) \delta_{i,j} \hat{n}_i(\hat{n}_i-1)$. It also includes systems with nearest-neighbor interactions, such as the model with interaction term $\mathcal{U}_{ij}(\hat{n}_i,\hat{n}_j) = U \delta_{(i,j)\in E(G)} \hat{n}_i \hat{n}_j$. 

\section*{Single-particle and two-particle scattering}\label{sec:momentum_states}

In our scheme for universal quantum computation we design a multi-particle quantum walk where the dynamics can be understood by considering scattering events involving one or two particles. In this section we show how single-particle and two-particle quantum walks on certain graphs can be analyzed using a discrete version of scattering theory. 

\paragraph*{Single-particle scattering}
\begin{figure}
\centering
\capstart
\begin{tikzpicture}[
  label distance=-5.5pt,
  thin,
  vertex/.style={circle,draw=black,fill=black,inner sep=1.25pt,
    minimum size =0mm},
  attach/.style={circle,draw=black,fill=white,inner sep=1.25pt,
    minimum size =0mm},
  dots/.style={circle,fill=black,inner sep=.5pt,
    minimum size= 0pt},
  every text node part/.style={font=\footnotesize}]

  \node at (.15, .63) [rectangle,fill=white] {$(1,1)$};
  \node at (.48, .17) [rectangle,fill=white,inner sep=0pt] {$(1,2)$};
  \node at (.22,-.5) [rectangle,fill=white] {$(1,N)$};

  \draw[thick] (15:1) arc (15:335:1);
  \draw[densely dotted] (15:1)  arc (15:5:1);
  \draw[densely dotted] (-15:1)  arc (-15:-25:1);

  \node at (-.42,.12) [rectangle]{${\widehat{G}}$};

  \foreach \x  in {50, 25,-35}{
    \draw (\x:1cm) -- (\x:3.15cm);
    \draw (\x:2.9cm) -- (\x:3.4cm) [densely dotted];
  }
  
  \node at (50:1)[attach]{};
  \node at (50:1.66)[vertex,label=150:{$(2,1)$}]{};
  \node at (50:2.33)[vertex,label=150:{$(3,1)$}]{};
  \node at (50:3)[vertex,label=150:{$(4,1)$}]{};

  \node at (25:1)[attach]{};
  \node at (25:1.66)[vertex]{};
  \node at (25:2.33)[vertex]{};
  \node at (25:3)[vertex]{};
  
  \node at (12:1.65) {$(2,2)$};
  \node at (15:2.5) {$(3,2)$};
  \node at (17.5:3.35) {$(4,2)$};

  \node at (-35:1)[attach]{};
  \node at (-35:1.66)[vertex,label=60:{$(2,N)$}]{};
  \node at (-35:2.33)[vertex,label=60:{$(3,N)$}]{};
  \node at (-35:3)[vertex,label=60:{$(4,N)$}]{};

  \node at (-2.5:2.9)[dots] {};
  \node at (2.5:2.9)[dots] {};
  \node at (-7.5:2.9)[dots] {};
\end{tikzpicture}
\caption{The graph $G$ obtained by attaching $N$ semi-infinite paths to a graph $\widehat{G}$. We label the vertices on the semi-infinite paths as $(x,j)$, with $j\in\{1,\ldots,N\}$ indexing the path and $x\in\natural = \{1,2,3,\ldots\}$ denoting the distance from $\widehat{G}$.  The vertices labeled $(1,1),$ $(1,2),$ $\ldots,$ $(1,N)$ are vertices of the original graph $\widehat{G}$.}
\label{fig:graph}
\end{figure}
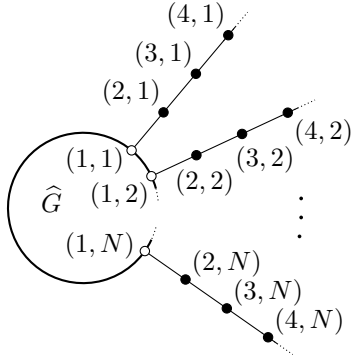

Consider a single-particle quantum walk on an infinite graph $G$ obtained by attaching a semi-infinite path to each of $N$ chosen vertices of an arbitrary $(N+m)$-vertex graph $\widehat{G}$ as depicted in \fig{graph}. Here we discuss the single-particle scattering process where a particle is initially prepared in a state that moves (under Schr\"{o}dinger time evolution) toward the subgraph $\widehat{G}$ along one of the semi-infinite paths.  After scattering through the subgraph, the particle moves away from $\widehat{G}$ in superposition along the semi-infinite paths. To understand this scattering process we discuss the scattering eigenstates of the Hamiltonian $H_{G}^{(1)}$.

Scattering states are states with definite incoming momentum. Given the graph $\widehat{G}$, for each momentum $k\in (-\pi,0)$ and each path $j\in \{1,2,\ldots,N\}$, there is a scattering state $|\mathrm{sc}_{j}(k)\rangle$ with amplitudes
\begin{align}
\langle x,q|\mathrm{sc}_{j}(k)\rangle 
  &= e^{-i k x} \delta_{qj}
   + e^{i k x} S_{qj}(k)
\label{eq:single_particle_states}
\end{align}
on the semi-infinite paths (with the labeling $(x,q)$ of vertices on the paths as in \fig{graph}), where the $N \times N$ matrix $S(k)$ appearing in the above equation is a unitary matrix called the S-matrix. In \sec{sing_part} we show how to obtain the S-matrix and the amplitudes of $|\mathrm{sc}_{j}(k)\rangle$ within the graph $\widehat{G}$ \cite{Childs_Gosset}. The state $|\mathrm{sc}_{j}(k)\rangle$ is an eigenstate of $H^{(1)}_G$ with energy $2\cos k$.  We also show in \sec{sing_part} that the states $|\mathrm{sc}_j(k)\rangle$ are delta-function normalized.

A wave packet is a normalized state with most of its amplitude on scattering states with momentum close to some particular value. The scattering state $|\mathrm{sc}_{j}(k)\rangle$ gives us information about how a  wave packet with momentum near $k$ located on the semi-infinite path $j$ scatters from the graph $\widehat{G}$.  The wave packet initially moves toward the graph with speed $|\frac{dE}{dk}|=|2\sin k|$. After scattering, the wave packet moves away from $\widehat{G}$ along each of the semi-infinite paths (in superposition), and the amplitude associated with finding the wave packet on path $q$ is $S_{qj}(k)$.

In \sec{analysis} we discuss scattering of single-particle wave packets in more detail. In particular, we show that this picture of the scattering process is valid when the finite extent of the wave packets is taken into account, and when the infinite paths are truncated to be long but finite.

\paragraph*{Two-particle scattering}

Now consider two-particle scattering on an infinite path. Translation symmetry makes this system easier to analyze than more general two-particle quantum walks (see for example references \cite{VP08,V10}).

We now discuss the scattering of two indistinguishable particles initially prepared in spatially separated wave packets moving toward each other along a path with momenta $k_1\in(-\pi,0) $ and $k_2\in (0,\pi)$ (we discuss distinguishable particles in \sec{distinguishable}). Due to conservation of energy and momentum, the state of this system after scattering still consists of one particle with momentum $k_1$ and another with momentum $k_2$, but now moving apart. Since the particles are indistinguishable, there is no distinction between transmission and reflection of the particles, so the effect of the interaction is to change the global phase of the wave function after scattering (relative to the case with no interaction). We show in the Appendices how to calculate this phase given the two momenta, the interaction $\mathcal{U}$, and the particle type (bosons or fermions). In \sec{analysis} we show that this picture of the two-particle scattering process holds on a long (but finite) path and with finite wave packets.

\section*{Computation by multi-particle quantum walk}

In this section we describe our scheme for performing quantum computation using a multi-particle quantum walk of indistinguishable particles (we present a refinement of our scheme that uses distinguishable particles in \sec{distinguishable}). We encode $n$ logical qubits into the locations of $n$ particles on a graph of size $\poly(n)$.  In addition to these $n$  particles, we use another particle to encode an ancilla qubit that facilitates two-qubit gates. We call the original $n$ qubits \emph{computational qubits}, and we call the ancilla qubit the \emph{mediator qubit}.  Time evolution of a simple initial state with the Hamiltonian corresponding to a suitably chosen graph $G$ implements a quantum circuit on this encoded data. 

The quantum circuit to be simulated (on the $n$ computational qubits) is given as a product of one- and two-qubit gates from a universal gate set consisting of single-qubit gates along with the controlled phase gate 
\[
\CP=
\begin{pmatrix}
1 & 0 & 0 & 0\\
0 & 1 & 0 & 0\\
0 & 0 & 1 &0\\
0 & 0 & 0 & -i\\
\end{pmatrix}.
\]
Note that a controlled phase between computational qubits $i$ and $j$ can be expanded as the following set of gates also acting on the mediator qubit $\med$ (initialized to $\ket{0}$):
\begin{align*}
  \CP_{ij} \ket{a_i, b_{j},0_{\med}}
     &= \CNOT_{i\med} \CP_{j\med} \CNOT_{i\med} \ket{a_i, b_{j},0_{\med}}\\
     &= \Had_\med \CP_{i\med}^2 \Had_\med \CP_{j\med} \Had_\med \CP_{i\med}^2 \Had_\med
         \ket{a_i, b_{j},0_{\med}}.
\end{align*}
Writing the controlled phase gate in this way, we can transform the given quantum circuit into one with only single-qubit gates on computational qubits, Hadamard gates on the mediator qubit, and controlled phase gates between the mediator qubit and arbitrary computational qubits. In our construction we use this gate set acting on $n+1$ qubits to simulate a quantum computation on $n$ qubits.

Each of the $n+1$ qubits is represented in a dual-rail encoding using two paths that run through the graph, as shown in \fig{wavetrain_encoding}. The encoded state $|0\rangle$ has a particle moving along the top path whereas the encoded state $|1\rangle$ has a particle moving along the bottom path. The particle moves as a wave packet with momentum near $k$. For the $n$ computational qubits we choose $k=-{\pi}/{4}$ and for the mediator qubit we choose $k=-{\pi}/{2}$ (for concreteness).

To implement single-qubit unitaries on each of the encoded qubits, we design the graph so that the particles scatter through a series of small subgraphs while remaining far apart. When the particles are all far from each other, the interaction term in the Hamiltonian is negligible and the $n+1$ wave packets propagate independently through the graph (here we neglect the interaction only for ease of exposition; our detailed analysis in \sec{description} handles the full interacting system). In this case the multi-particle quantum walk can be viewed as a single-particle quantum walk for each of the particles. 

To implement the controlled phase gate between the mediator and a computational qubit, we design a subgraph that routes two particles toward each other and causes them to interact on a long path for a short time.  The two particles then scatter away from one another and the system returns to a state where the particles are all far apart.

\begin{figure}
\centering
\capstart
\subfigure[Encoded $|0\rangle$.] 
{
\begin{tikzpicture}[
  vert/.style={circle,fill=black,inner sep=.7pt,minimum width=0pt},
  dots/.style={circle,fill=black,inner sep=.25pt,minimum width=0pt},
  thick,
  scale=0.4]
  \foreach \x in {0,.4,...,10}{
    \node at (\x ,0) [vert]{};
    \node at (\x ,3) [vert]{};
  }

  \draw (0,0) -- (10,0);
  \draw (0,3) -- (10,3);

  \begin{scope}[yshift=0cm]
  \draw (1.6,3.25) -- (2,3.25) -- (2,4) to node [above] {$k\rightarrow$}
        (8,4) -- (8,3.25) -- (8.4,3.25);
  \end{scope}

  \foreach \xsh in {-1.5cm , 10.5cm}{
  \foreach \ysh in {0cm, 3cm}{
    \begin{scope}[xshift=\xsh,yshift=\ysh]
      \node at (0,0) [dots]{};
      \node at (0.5,0) [dots] {};
      \node at (1,0) [dots]{};
    \end{scope}
  }}
\end{tikzpicture}
}
\hspace{2cm}
\subfigure[Encoded $|1\rangle$.] 
{
\begin{tikzpicture}[
  vert/.style={circle,fill=black,inner sep=.7pt,minimum width=0pt},
  dots/.style={circle,fill=black,inner sep=.25pt,minimum width=0pt},
  thick,
  scale=0.4]
  \foreach \x in {0,.4,...,10}{
    \node at (\x ,0) [vert]{};
    \node at (\x ,3) [vert]{};
  }

  \draw (0,0) -- (10,0);
  \draw (0,3) -- (10,3);

  \begin{scope}[yshift=-3cm]
  \draw (1.6,3.25) -- (2,3.25) -- (2,4) to node [above] {$k\rightarrow$}
        (8,4) -- (8,3.25) -- (8.4,3.25);
  \end{scope}

  \foreach \xsh in {-1.5cm , 10.5cm}{
  \foreach \ysh in {0cm, 3cm}{
    \begin{scope}[xshift=\xsh,yshift=\ysh]
      \node at (0,0) [dots]{};
      \node at (0.5,0) [dots] {};
      \node at (1,0) [dots]{};
    \end{scope}
  }}
\end{tikzpicture}
}
\caption{A qubit is encoded using single-particle wave packets.  For an $n$-qubit computation we use $n$ particles with momentum $k=-{\pi}/{4}$ and one with $k=-{\pi}/{2}$. }
\label{fig:wavetrain_encoding}
\end{figure}
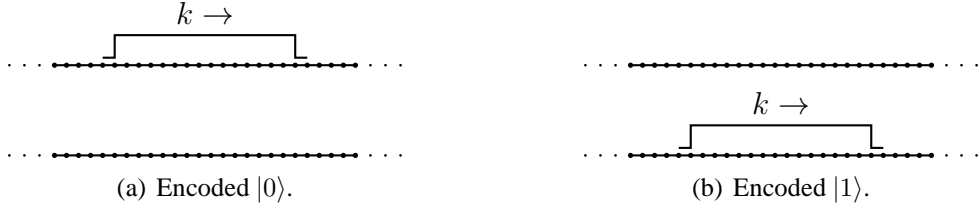

\paragraph*{Single-qubit gates}
To apply a one-qubit unitary $U$ to a computational qubit, we insert an associated graph $\widehat{G}$ into the paths representing the qubit as follows. We attach two long ``input'' paths and two long ``output'' paths to four suitably chosen vertices of $\widehat{G}$ so that the S-matrix at momentum $-{\pi}/{4}$ has the form
\begin{equation}
S=
\begin{pmatrix}
0 & U^{\prime}\\
U & 0\\
\end{pmatrix}
\label{eq:S_matrix_circuit}.
\end{equation}
Each block of this matrix has size $2\times 2$.  A particle incident on the input paths with momentum $-{\pi}/{4}$ transmits perfectly to the output paths, with amplitudes determined by the unitary $U$. The scattering process implements the unitary $U$ on the encoded qubit. 

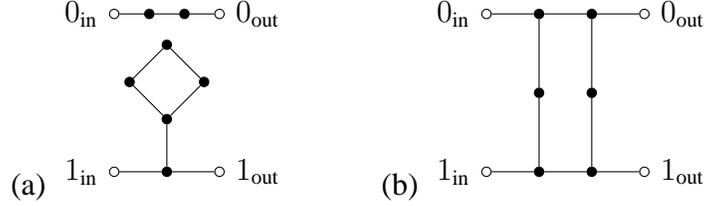
\begin{figure}
\centering
\capstart
\subfigure{(a)}
\begin{tikzpicture}
  [ scale=0.7,inner/.style={circle,draw=black!100,fill=black!100,inner sep = 1.25pt},
    attach/.style={circle,draw=black!100,fill=black!0,thin,inner sep = 1.25pt}]
  \node (0) at (-1,      0)      [attach,label=left:$1_\text{in}$] {};
  \node (1) at ( 0,      0)      [inner]  {};
  \node (2) at ( 1,      0)      [attach,label=right:$1_\text{out}$] {};  
  \node (3) at ( 0,      1)      [inner]  {};
  \node (4) at ( 0,      2.4142) [inner]  {};
  \node (5) at (-0.7071, 1.7071) [inner]  {};
  \node (6) at ( 0.7071, 1.7071) [inner]  {};
  \node (7) at (-1,      3)      [attach,label=left:$0_\text{in}$] {};
  \node (8) at (-0.3333, 3)      [inner]  {};
  \node (9) at ( 0.3333, 3)      [inner]  {};
  \node (10)at ( 1,      3)      [attach,label=right:$0_\text{out}$] {};

  \draw (1) to (0) [thin];
  \draw (1) to (2) [thin];
  \draw (1) to (3) [thin];
  \draw (3) to (5) [thin];
  \draw (3) to (6) [thin];
  \draw (5) to (4) [thin];
  \draw (6) to (4) [thin];  
  \draw (8) to (7) [thin];
  \draw (8) to (9) [thin];
  \draw (9) to (10)[thin];
\end{tikzpicture}
\hspace{1cm}
\subfigure{(b)}
\begin{tikzpicture}
  [ scale=0.7,
    yscale=1.5,
    inner/.style={circle,draw=black!100,fill=black!100,inner sep = 1.25pt},
    attach/.style={circle,draw=black!100,fill=black!0,thin,inner sep = 1.25pt}]

  \node (1) at ( 0, 2) [inner]  {};
  \node (2) at ( 0, 0) [inner]  {};
  \node (3) at ( 1, 2) [inner]  {};
  \node (4) at ( 1, 0) [inner]  {};
  \node (5) at ( 0, 1) [inner]  {};
  \node (6) at ( 1, 1) [inner]  {};
  \node (7) at (-1, 2) [attach,label=left:$0_\text{in}$] {};
  \node (8) at (-1, 0) [attach,label=left:$1_\text{in}$] {};
  \node (9) at ( 2, 2) [attach,label=right:$0_{\text{out}}$] {};
  \node (0) at ( 2, 0) [attach,label=right:$1_{\text{out}}$] {};

  \draw (7) to (1) [thin];
  \draw (8) to (2) [thin];
  \draw (3) to (9) [thin];
  \draw (4) to (0) [thin];
  \draw (1) to (3) [thin];
  \draw (1) to (5) [thin];
  \draw (2) to (4) [thin];
  \draw (2) to (5) [thin];
  \draw (3) to (6) [thin];
  \draw (4) to (6) [thin];
\end{tikzpicture}
\caption{One-qubit gates at $k=-{\pi}/{4}$. (a) Phase gate. (b) Basis-changing gate. }
\label{fig:onepphase}
\end{figure}

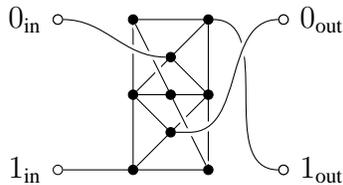
\begin{figure}
\centering
\capstart
\begin{tikzpicture}
  [ scale=1,
    thin,
    attach/.style={circle,fill=white,draw=black,
      inner sep=1.25pt,minimum size=0pt},
    vert/.style={circle,draw=black,fill=black,
      inner sep=1.25pt,minimum size=0pt}]
  
    \node (0inhad) at (0,2) [attach]{};
    \node (1inhad) at (0,0) [attach]{};
    \node (0outhad) at (3,2) [attach]{};
    \node (1outhad) at (3,0) [attach]{};
    
    \node (1had) at (1,0) [vert]{};
    \node (2had) at (2,0) [vert]{};
    \node (3had) at (1.5,.5) [vert]{};
    \node (4had) at (1,1) [vert]{};
    \node (5had) at (1.5,1) [vert] {};
    \node (6had) at (2,1) [vert] {};
    \node (7had) at (1.5,1.5) [vert]{};
    \node (8had) at (1,2) [vert]{};
    \node (9had) at (2,2) [vert]{};

    \draw (8had) -- (1had);
    \draw (9had) -- (2had);
    \draw (4had) -- (9had);
    \draw (1had) -- (6had);
    \draw (4had) -- (6had);
    \draw (3had) -- (4had);
    \draw (7had) -- (6had);
    \draw (1inhad) -- (2had);
    \draw (8had) -- (9had);
    \draw (8had) -- (5had)[line width=3pt,white];
    \draw (5had) -- (2had)[line width=3pt,white];
    \draw (8had) -- (2had);
    
    \draw (9had) to[out=0,in=180] (1outhad) [looseness=.3];
    \draw (0inhad) to[out=0,in=180] (7had)[line width = 3pt,white];
    \draw (0inhad) to[out=0,in=180] (7had);

    \draw (3had) --(1.75,.5) to[out=0,in=180] (0outhad)[line width =3pt,white];
    \draw (3had) --(1.75,.5) to[out=0,in=180] (0outhad);

    \node at (0inhad) [attach]{};
    \node at (1outhad) [attach]{};
    \node at (0outhad) [attach]{};

  \node at (1inhad.west) [left] {$1_{\text{in}}$};
  \node at (0inhad.west) [left] {$0_{\text{in}}$};
  \node at (0outhad.east) [right] {$0_{\text{out}}$};
  \node at (1outhad.east) [right] {$1_{\text{out}}$};
\end{tikzpicture}
\caption{Graph implementing a Hadamard gate at $k = -{\pi}/{2}$.}
\label{fig:halfpihad}
\end{figure}
Graphs $\widehat{G}$ that implement a phase gate and a basis-changing gate at momentum $-{\pi}/{4}$ are shown in \fig{onepphase} \cite{Chi09}.  The input and output paths are attached to the vertices denoted by open circles. The S-matrix at momentum $-{\pi}/{4}$ for each of these graphs is a $4\times 4$ matrix of the form \eq{S_matrix_circuit}, with the lower left $2\times 2$ submatrix given by 
\[
  U_{\text{phase}} = \begin{pmatrix}e^{-i\pi/4} & 0\\ 0 &  1\end{pmatrix} \qquad	
  U_{\text{basis}} = -\frac{i}{\sqrt{2}}\begin{pmatrix} 1 & - i  \\- i & 1\end{pmatrix}.
\]

These two gates allow us to approximate arbitrary single-qubit unitaries on each of the $n$ computational qubits. To obtain a graph implementing two of these unitaries in series, simply concatenate the two graphs by attaching the output vertices of the first graph to the input vertices of the second graph. At momentum $-{\pi}/{4}$, the S-matrix of the resulting graph is the product of the S-matrices of the two graphs being concatenated (this is true whenever the two graphs have perfect transmission from input vertices to output vertices \cite{Chi09}).  

The Hadamard gate is the only nontrivial single-qubit gate we apply to the mediator qubit. To do so we use the graph depicted \fig{halfpihad} \cite{2011PhRvA..84f2302B}.  This graph has S-matrix at momentum $k={-\pi}/{2}$ of the form \eq{S_matrix_circuit}, with lower left submatrix
\[
U_H=-\frac{e^{i\frac{\pi}{4}}}{\sqrt{2}}\begin{pmatrix}1 & 1\\ 1 &  -1\end{pmatrix},
\]
which is the Hadamard gate up to an irrelevant global phase.
\paragraph*{Two-qubit gates}
To implement the controlled phase gate between the mediator qubit and a computational qubit we use some facts about two-particle scattering on a long path. Recall  that two indistinguishable particles of momentum $k_1$ and $k_2$ initially traveling toward each other will, after scattering, continue to travel as if no interaction occurred, except that the phase of the wave function is modified by the interaction. In general this phase depends on $k_1$ and $k_2$ (as well as the interaction $\mathcal{U}$ and the particle statistics).  For us, $k_1=-{\pi}/{2}$ and $k_2={\pi}/{4}$ (moving in opposite directions).  We write $e^{i\theta}$ for the phase acquired at these momenta.

In our scheme we design a subgraph that routes a computational particle and a mediator particle toward each other along a long path only when the two associated qubits are in state $\ket{11}$. This allows us to implement the two-qubit gate
\[
\CD=
\begin{pmatrix}
1 & 0 & 0 & 0\\
0 & 1 & 0 & 0\\
0 & 0 & 1 &0\\
0 & 0 & 0 & e^{i\theta}\\
\end{pmatrix}.
\]
For some models $\CD =\CP$. We show in \sec{twopart_scat} that this holds in the Bose-Hubbard model (where the interaction term is  $\mathcal{U}_{ij}(\hat{n}_i,\hat{n}_j)=  (U/2) \delta_{i,j} \hat{n}_i(\hat{n}_i-1)$) when the interaction strength is chosen to be $U=2+\sqrt{2}$, since in this case $e^{i\theta}=-i$. For nearest-neighbor interactions with fermions, with $\mathcal{U}_{ij}(\hat{n}_i,\hat{n}_j) = U \delta_{(i,j) \in E(G)} \hat{n}_i \hat{n}_j$, the choice $U=-2-\sqrt{2}$ gives $e^{i\theta}=i$, so $\CP=(\CD)^3$. While tuning the interaction strength makes the $\CP$ gate easier to implement, almost any interaction between indistinguishable particles allows for universal computation. We can approximate the required CP gate by repeating the $\CD$  gate $a$ times, where $e^{i a \theta} \approx -i$ (which is possible for most values of $\theta$, assuming $\theta$ is known \cite{note2}). 

Our strategy requires routing the particles onto a long path.  This is done via a subgraph we call the \emph{momentum switch}, as depicted in \fig{onepsplit}(a). The S-matrices for this graph at momenta $-{\pi}/{4}$ and $-{\pi}/{2}$ are
\begin{equation}
  S_{\text{switch}}\left(-\pi/4\right) = \begin{pmatrix} 0 & 0 & e^{-i\pi/4}\\
    0 & -1 & 0\\
    e^{-i\pi/4} & 0 & 0\end{pmatrix}\qquad
  S_{\text{switch}}\left(-\pi/2\right) = \begin{pmatrix}1 & 0 &0\\
    0 & 0 & -1\\
    0 & -1 & 0\end{pmatrix}.
\label{eq:switch_S}
\end{equation}
 The momentum switch has perfect transmission between vertices 1 and 3 at momentum $-{\pi}/{4}$ and perfect transmission between vertices 2 and 3 at momentum $-{\pi}/{2}$. In other words, in the schematic shown in \fig{onepsplit}(a), the path a particle follows through the switch depends on its momentum. A particle with momentum $-{\pi}/{2}$ follows the double line, while a particle with momentum $-{\pi}/{4}$ follows the single line.

The graph used to implement the $\CD$ gate has the form shown in \fig{onepsplit}(b).  We specify the number of vertices on each of the paths in \sec{description}. To see why this graph implements a $\CD$ gate, consider the movement of two particles as they pass through the graph. If either particle begins in the state $\ket{0_{\text{in}}}$, then it travels along a path to the output without interacting with the second particle. When the computational particle (qubit $c$ in the figure) begins in the state $\ket{1_{\text{in}}}^c$, it is routed downward as it passes through the top momentum switch (following the single line). It travels down the vertical path and then is routed to the right (along the single line) as it passes through the bottom switch.  Similarly, when the mediator particle begins in the state $\ket{1_{\text{in}}}^\med$, it is routed upward (along the double line) through the vertical path at the bottom switch and then to the right (along the double line) at the top switch. If both particles begin in the state $\ket{1_{\text{in}}}$, then they interact on the vertical path. In this case, as the two particles move past each other, the wave function acquires a phase $e^{i\theta}$ arising from this interaction. 

\begin{figure}
\centering
\capstart
\subfigure{(a)}
\begin{tikzpicture}
  [ scale = 0.6,
    thin,
    inner/.style={circle,draw=black!100,fill=black!100,inner sep=1.25pt},
    attach/.style={circle,draw=black!100,fill=black!0,thin,inner sep=1.25pt},
    sin/.style={line width=.7pt},
    doub/.style={line width=2.1pt},
    trip/.style={draw=white,line width=.7pt}]
    \node at (0,0){};
\begin{scope}[yshift=1.8cm]
  \node (2)  at (0,0) [attach,label=left:3] {};
  \node (1)  at (0,1) [attach,label=left:1] {};
  \node (3)  at (3,1) [attach,label=right:2] {};
  \node (4)  at (1,0) [inner]  {};
  \node (5)  at (2,0) [inner]  {};
  \node (6)  at (1,1) [inner]  {};
  \node (7)  at (2,1) [inner]  {};
  \node (8)  at (2,2) [inner]  {};
  \node (9)  at (2.866,2.5)  [inner] {};
  \node (10) at (1.134,2.5)  [inner] {};
  \node (11) at (2,-1)       [inner] {};
  \node (12) at (2.866,-1.5) [inner] {};
  \node (13) at (1.134,-1.5) [inner] {};

  \draw (2) to (4);
  \draw (4) to (5);
  \draw (3) to (7);
  \draw (1) to (6);
  \draw (6) to (4);
  \draw (6) to (7);
  \draw (7) to (5);
  \draw (7) to (8);
  \draw (8) to (9);
  \draw (8) to (10);
  \draw (11) to (5);
  \draw (11) to (12);
  \draw (11) to (13);

  \node (split) at (-3.6,.8) [draw=black,circle,inner sep=3mm,
         label=left:1,label=right:2,label=below:3] {};
 
  \node at (-1.6,.5) {=};

  \draw[sin]  (split.west) to[out=0,in=90]   (split.south);
  \draw[doub] (split.east) to[out=180,in=90] (split.south);
  \draw[trip] (split.east) to[out=180,in=90] (split.south);
  
  \node at (split.east) [attach] {};
  \node at (split.west) [attach] {};
  \node at (split.south) [attach]{};
\end{scope}
\end{tikzpicture}
\hspace{0.5cm}
\subfigure{(b)}
\begin{tikzpicture}
  [ scale = 1,
    yscale = .8,
    attach/.style={circle,draw=black!100,fill=white,thick,
    minimum size = 6mm},
    cross/.style={line width=4pt, draw=white},
    drawn/.style={draw=black},
    vert/.style = {circle,fill=black,inner sep=.6pt, minimum size=0},
    nofill/.style = {circle,draw=black,fill=white,inner sep = 1.25pt,minimum size=0},
    decoration={markings,
		mark=between positions 0 and 10 step .1cm
 		with { \node at (0,0) [vert]{}; }}]

  \node (bottom) at ( 1, 0) [attach] {};
  \node (top)    at ( 1, 2.5) [attach] {};
  
  \foreach \x in {0,.1,...,3.5}{
  \foreach \y in {.75, 3.25}{
    \node at (\x,\y) [vert] {};
  }}

  \draw[postaction={decorate}] (0,0) node[left] {$1_{\med,\text{in}}$} 
    -- (bottom.west);
  \draw (0,.75) node [left] {$0_{\med,\text{in}}$} 
    -- (3.5,.75) node [right] {$0_{\med,\text{out}}$};
  \draw (0,3.25) node [left] {$0_{c,\text{in}}$}
    -- (3.5,3.25) node [right] {$0_{c,\text{out}}$};
  \draw[postaction={decorate}] (0,2.5) node [left] {$1_{c,\text{in}}$} 
    -- (top.west) ;
  \draw (top.south) -- (bottom.north) [cross];

  \draw (top.south) -- (bottom.north) [drawn,postaction={decorate}];
  \draw (top.east) -- ( 2, 2.5)  .. controls (2.5,2.5) and (2.5, 0) 
                   .. (3,0) -- (3.5, 0)  [cross];
  \draw (top.east) -- ( 2, 2.5)  .. controls (2.5,2.5) and (2.5, 0) 
                    .. (3,0) -- (3.5, 0) 
                    node [right] {$1_{\med,\text{out}}$} [drawn,postaction={decorate}];
  \draw (bottom.east) -- ( 2, 0) .. controls (2.5,0) and (2.5,2.5)  
                      .. (3,2.5) -- (3.5, 2.5)  [cross];
  \draw[drawn,postaction={decorate}] (bottom.east) -- ( 2, 0) .. controls (2.5,0) and (2.5,2.5)  
                      .. (3,2.5) -- (3.5, 2.5) 
                      node [right] {$1_{c,\text{out}}$} ;

  \draw (top.west) to[out=0,in=90] (top.south) [line width = .7pt];
  \draw (bottom.east) to[out=-180,in=-90] (bottom.north) [line width=.7pt];
  \draw (top.east) to[out=-180,in=90] (top.south) [line width=2.1pt];
  \draw (top.east) to[out=-180,in=90] (top.south) [line width=.7pt,draw=white];
  \draw (bottom.west) to[out=0,in=-90] (bottom.north) [line width=2.1pt];
  \draw (bottom.west) to[out=0,in=-90] (bottom.north) [line width=.7pt,draw=white];

  \foreach \x in {0, 3.5}{
  \foreach \y in {0,.75, 2.5, 3.25}{
    \node at (\x,\y) [nofill] {};
  }}
 
 \node at (top.west) [nofill]{};
 \node at (top.east) [nofill]{};
 \node at (top.south) [nofill]{};
 \node at (bottom.west) [nofill]{};
 \node at (bottom.east) [nofill]{};
 \node at (bottom.north) [nofill]{};
  
\end{tikzpicture}

\caption{(a) Momentum switch. (b) $\CD$ gate.}
\label{fig:onepsplit}
\end{figure}
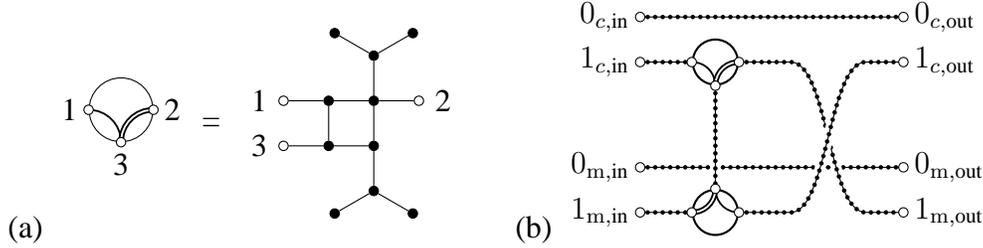

Note that timing is important: the wave packets of the two particles must be on the vertical path at the same time. We achieve this by choosing the number of vertices on each of the segments in the graph appropriately, taking into account the different propagation speeds of the two wave packets (see \sec{description} for details). 

\paragraph*{Discussion}
Our scheme for simulating a quantum circuit is summarized as follows. We encode $n$ computational qubits as single-particle wave packets with momentum $-{\pi}/{4}$ traveling along two paths, along with a single mediator qubit similarly encoded but with momentum $-{\pi}/{2}$.  We perform single-qubit gates as shown in \fig{onepphase} and \fig{halfpihad}, and we implement two-qubit gates using the mediator qubit and the graph shown in \fig{onepsplit}. The subgraphs representing circuit elements are connected by paths. To illustrate how these ingredients are put together, a graph corresponding to a simple 2-qubit computation is depicted in \fig{example}. 

In \sec{description} we present all the details of our scheme and we prove that it performs the desired quantum computation up to an error term that can be made arbitrarily small. To prove this error bound we analyze the full $(n+1)$-particle interacting many-body system. Our analysis goes beyond the scattering theory discussion presented in the previous section; we take into account the fact that the wave packets are finite (and therefore have support on a range of momenta), as well as the fact that the graphs involved in our scheme are finite. Specifically, we prove that by choosing the size of the wave packets, the number of vertices in the graph, and the total evolution time to be polynomial functions of both $n$ and $g$, the error in simulating an $n$-qubit, $g$-gate quantum circuit is bounded above by an arbitrarily small constant (\sec{analysis}).  The bounds we prove are sufficient to establish universality with only polynomial overhead, but they are almost certainly not optimal. For example, for the Bose-Hubbard model and for models with nearest-neighbor interactions, we prove that the error can be made arbitrarily small by choosing the size of the wave packets to be $\O(n^{12}g^4)$, the total number of vertices in the graph to be  $\O(n^{13}g^5)$, and the total evolution time to be $\O(n^{12}g^5)$.

We also describe two refinements of the scheme presented here. In \sec{planar} we show how the scheme presented above can be adapted so that the graph is planar and has maximum degree four, making it more amenable to implementation.  In \sec{distinguishable} we show how the scheme can be adapted to use distinguishable particles with nearest-neighbor interactions.

\begin{figure}
\centering
\begin{tikzpicture}
  [ scale = .45,
    xscale = 0.8,
    inner/.style={circle,draw=black!100,fill=black!100,inner sep=1pt},
    attach/.style={circle,draw=black!100,fill=black!0,thin,inner sep=.7pt},
    switch/.style={circle,draw=black!100,fill=black!50,thin,inner sep=1.25 pt},
    cross/.style={draw=white,double=black,ultra thick},
    splitter/.style={circle,fill=white,draw=black!100,inner sep=.8mm},
    had/.style={rectangle,draw=black!100,fill=white,rounded corners,thick,minimum size = 1},
    verts/.style={circle,fill=black,inner sep =.6pt,minimum size = 0pt},
    vert/.style={circle,fill=black,inner sep=0.5pt,minimum size=0pt},
    decoration={markings, 
      mark= between positions 0 and 1 step 1.2 mm with{
        \node[vert] at (0,0){};
      }}
    ]
  \node (mao) at (20,0) [attach,label=right:$0_{\med,\text{out}}$] {};

  \node (mbo) at (20,-1) [attach,label=right:$1_{\med,\text{out}}$] {};
 
  \node (l2bo) at (20,1.5) [attach,label=right:$1_{2,\text{out}}$] {};

  \node (l2ao) at (20,2.5)   [attach,label=right:$0_{2,\text{out}}$] {};

  \node (l1bo) at (20,4) [attach,label=right:$1_{1,\text{out}}$] {};

  \node (l1ao) at (20,5) [attach,label=right:$0_{1,\text{out}}$] {};
 
  \node (l1aa) at (-20,5) [attach,label=left:$0_{1,\text{in}}$] {};
  \node (l1ba) at (-20,4) [attach,label=left:$1_{1,\text{in}}$] {};
  \node (l2aa) at (-20,2.5) [attach,label=left:$0_{2,\text{in}}$] {};
  \node (l2ba) at (-20,1.5) [attach,label=left:$1_{2,\text{in}}$] {};
  \node (maa)  at (-20,0) [attach,label=left:$0_{\med,\text{in}}$] {};
  \node (mba)  at (-20,-1) [attach,label=left:$1_{\med,\text{in}}$] {};

  \node (bl1) at (-11,-1) [splitter] {};
  \node (bl2) at (0, -1) [splitter] {};
  \node (bl3) at (9, -1) [splitter] {};

  \node (br1) at (-9,-1) [splitter] {};
  \node (br3) at (11, -1) [splitter] {};

  \node (tl1) at (-11, 4) [splitter] {};
  \node (tl2) at (0 ,1.5) [splitter] {};
  \node (tl3) at (9,  4) [splitter] {};

  \node (tr1) at (-9, 4) [splitter] {};
  \node (tr3) at (11,  4) [splitter] {};

  \begin{scope}[thin]
  \draw[postaction={decorate}] (l1aa) -- (l1ao);
  \draw[postaction={decorate}] (l2aa) -- (14.25,2.5);
  \draw[postaction={decorate}] (15.75,2.5) -- (l2ao);
  \draw[postaction={decorate}] (l1ba) -- (tl1.west);
  \draw[postaction={decorate}] (tl1.east) -- (tr1.west);
  \draw[postaction={decorate}] (tr1.east) -- (tl3.west);
  \draw[postaction={decorate}] (tl3.east) -- (tr3.west);
  \draw[postaction={decorate}] (tr3.east) -- (l1bo);
  \draw[postaction={decorate}] (l2ba) -- (tl2.west);
  \draw[white,line width=5pt] (0.5,1.5) to[out=0,in=180] (2.5,-1);
  \draw[postaction={decorate}] (tl2.east) -- (0.5,1.5) 
      to[out=0,in=180] (2.5,-1) -- (4.25,-1);
  \draw[postaction={decorate}] (5.75,-1) -- (bl3.west);
  \draw[postaction={decorate}] (15.75,1.5) -- (l2bo);
  \draw[postaction={decorate}] (maa)  -- (-15.75,0);
  \draw[postaction={decorate}] (-14.25,0) -- (-5.75,0);
  \draw[postaction={decorate}] (-4.25,0) -- (4.25,0);
  \draw[postaction={decorate}] (5.75,0) -- (14.25,0);
  \draw[postaction={decorate}] (15.75,0) -- (mao);
  \draw[postaction={decorate}] (mba)  -- (-15.75,-1);
  \draw[postaction={decorate}] (-14.25,-1) -- (bl1.west);
  \draw[postaction={decorate}] (bl1.east) -- (br1.west);
  \draw[postaction={decorate}] (br1.east) -- (-5.75,-1);
  \draw[postaction={decorate}] (-4.25,-1) -- (bl2.west);
  \draw[white,line width=5pt] (0.5,-1) to[out=0,in=180] (2.5,1.5);
  \draw[postaction={decorate}] (bl2.east) -- (0.5,-1) 
      to[out=0,in=180] (2.5,1.5) -- (14.25,1.5);
  \draw[postaction={decorate}] (bl3.east) -- (br3.west);
  \draw[postaction={decorate}] (br3.east) -- (14.25,-1);
  \draw[postaction={decorate}] (15.75,-1) -- (mbo);
  \end{scope}

  \begin{scope}[draw=white, line width = 5 pt]
  \draw (bl1.north) -- (tl1.south);
  \draw (bl1.north) -- (tl1.south);
  \draw (bl2.north) -- (tl2.south);
  \draw (bl3.north) -- (tl3.south);
  
  \draw (br1.north) -- (tr1.south);
  \draw (br3.north) -- (tr3.south);
  \end{scope}

  \begin{scope}[thin]
  \draw[postaction={decorate}] (bl1.north) -- (tl1.south);
  \draw[postaction={decorate}] (bl2.north) -- (tl2.south);
  \draw[postaction={decorate}] (bl3.north) -- (tl3.south);
  
  \draw[postaction={decorate}] (br1.north) -- (tr1.south);
  \draw[postaction={decorate}] (br3.north) -- (tr3.south);
  \end{scope}
 
  \draw (br1.west) to[out=0,in=-90] (br1.north);
  \draw (br3.west) to[out=0,in=-90] (br3.north);

  \draw (tr1.south) to[out=90,in=180] (tr1.east);
  \draw (tr3.south) to[out=90,in=180] (tr3.east);

  \draw (tl1.west) to[out=0,in=90] (tl1.south);
  \draw (tl2.west) to[out=0,in=90] (tl2.south);
  \draw (tl3.west) to[out=0,in=90] (tl3.south);

  \draw (bl1.north) to[out=-90,in=180] (bl1.east);
  \draw (bl2.north) to[out=-90,in=180] (bl2.east);
  \draw (bl3.north) to[out=-90,in=180] (bl3.east);

  \draw[line width=.9pt] (bl1.west) to[out=0,in=-90] (bl1.north);
  \draw[line width=.9pt] (bl2.west) to[out=0,in=-90] (bl2.north);
  \draw[line width=.9pt] (bl3.west) to[out=0,in=-90] (bl3.north);

  \draw[line width=.9pt] (br1.east) to[out=180,in=-90] (br1.north);
  \draw[line width=.9pt] (br3.east) to[out=180,in=-90] (br3.north);

  \draw[line width=.9pt] (tr1.south) to[out=90,in=0] (tr1.west);
  \draw[line width=.9pt] (tr3.south) to[out=90,in=0] (tr3.west);

  \draw[line width=.9pt] (tl1.east) to[out=180,in=90] (tl1.south);
  \draw[line width=.9pt] (tl2.east) to[out=180,in=90] (tl2.south);
  \draw[line width=.9pt] (tl3.east) to[out=180,in=90] (tl3.south);

  \draw[line width=.3pt,draw=white] (bl1.west) to[out=0,in=-90] (bl1.north);
  \draw[line width=.3pt,draw=white] (bl2.west) to[out=0,in=-90] (bl2.north);
  \draw[line width=.3pt,draw=white] (bl3.west) to[out=0,in=-90] (bl3.north);

  \draw[line width=.3pt,draw=white] (br1.east) to[out=180,in=-90] (br1.north);
  \draw[line width=.3pt,draw=white] (br3.east) to[out=180,in=-90] (br3.north);

  \draw[line width=.3pt,draw=white] (tr1.south) to[out=90,in=0] (tr1.west);
  \draw[line width=.3pt,draw=white] (tr3.south) to[out=90,in=0] (tr3.west);

  \draw[line width=.3pt,draw=white] (tl1.east) to[out=180,in=90] (tl1.south);
  \draw[line width=.3pt,draw=white] (tl2.east) to[out=180,in=90] (tl2.south);
  \draw[line width=.3pt,draw=white] (tl3.east) to[out=180,in=90] (tl3.south);

  \foreach \n in {tl1, tl2, tl3, tr1, tr3}{
    \node at (\n.east) [attach]{};
    \node at (\n.west) [attach]{};
    \node at (\n.south) [attach]{};
  }
  
  \foreach \n in {bl1, bl2, bl3, br1, br3}{
    \node at (\n.east) [attach]{};
    \node at (\n.west) [attach]{};
    \node at (\n.north) [attach]{};
  }


\begin{scope}[xshift=14.25 cm, yshift= 1.5cm]
  \node [had] at (0.8,0.5) {$B$};
  \foreach \y in {0,1}{
  \foreach \z in {-.15,1.75}{
  \node [attach] at (\z cm, \y cm) {};}}
\end{scope}

\begin{scope}[xshift=-15.75 cm]
\foreach \x in {0,10,20,30}{
\begin{scope}[xshift= \x cm, yshift= -1cm]
  \node [had] at (0.8,0.5) {$H$};
  \foreach \y in {0,1}{
  \foreach \z in {-.2,1.8}{
  \node [attach] at (\z cm, \y cm) {};}}
\end{scope}}
\end{scope}


\begin{scope}[xshift=-15.75 cm, yshift = 4 cm]
  \node [had] at (0.8,0.5) {$T$};
  \foreach \y in {0,1}{
  \foreach \z in {-.08,1.68}{
  \node [attach] at (\z cm, \y cm) {};}}
\end{scope}
\end{tikzpicture}
\caption{ Schematic depiction of a graph simulating $\Bas_2\CP_{1,2}\T_1$ on two qubits, for the Bose-Hubbard model with interaction strength $U=2+\sqrt{2}$. The dotted lines represent paths and the single-qubit unitary gates represent their corresponding subgraphs ($B$ is the basis-changing gate, $T$ is the phase gate, and $H$ is the Hadamard gate).}
\label{fig:example}
\end{figure}
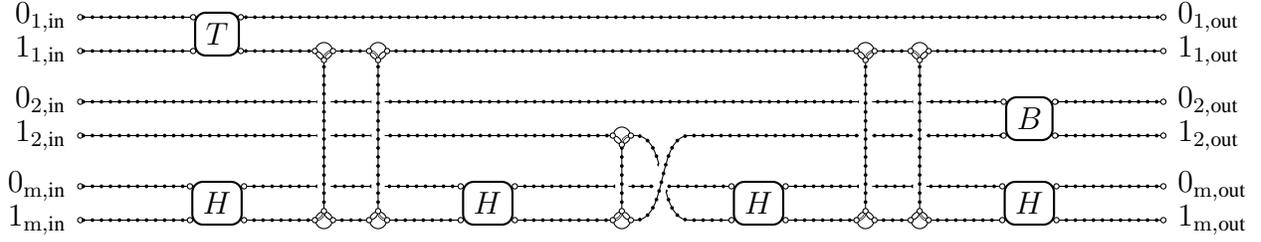

We have shown how multi-particle quantum walk can be used to perform efficient universal quantum computation. Our results provide an impetus for further experimental investigation of multi-particle quantum walk. We also hope that some of the tools we have developed in this work can be used to design and analyze new quantum algorithms based on multi-particle quantum walk.
\bibliographystyle{unsrt}
\bibliography{refs}

\begin{thebibliography}{10}
\bibitem{note2}
We can find such an $a$ provided $\theta = 2\pi\gamma$ for some rational number
  $\gamma$ whose denominator in lowest terms is a multiple of $4$. If $\gamma$
  is irrational, then we can choose $a$ so that $e^{i a \theta}$ is arbitrarily
  close to any desired phase, say $-i$. If $\gamma$ is rational but its
  denominator is not a multiple of $4$ (e.g., if $\gamma = 1/3$), then we
  cannot approximate the phase $-i$ arbitrarily closely, but a slight
  perturbation will allow for the correct phase to be approximated (possibly
  with a large but constant overhead).
\end{thebibliography}

\section*{Acknowledgments}
 This work was supported in part by MITACS; NSERC; the Ontario Ministry of Research and Innovation; the Ontario Ministry of Training, Colleges, and Universities; and the US ARO.

\newpage
\appendix

\section{Single-particle scattering states}
\label{sec:sing_part}

In this section we establish some basic facts about the scattering states of the single-particle Hamiltonian $H^{(1)}_G$. Write
\[
\widehat{H}=
\begin{pmatrix}
A & B^{\dagger} \\
B & D\\
\end{pmatrix}
\]
for the adjacency matrix of $\widehat{G}$, where $A$ is $N\times N$, $D$ is $m\times m$, and $B$ is $m \times N$.  The first $N$ rows and columns of this matrix correspond to the vertices attached to the semi-infinite paths. 

The results of this section apply to a broader class of Hamiltonians $H_G^{(1)}$ than those considered so far. In particular, the results are valid when the finite graph $\widehat{G}$ is a weighted, directed graph with edges in opposite directions having complex conjugate weights, in which case the Hamiltonian $\widehat{H}$ is an arbitrary Hermitian matrix. In this situation the full Hamiltonian $H_G^{(1)}$  has a term corresponding to $\widehat{H}$ and a term representing the (unweighted, undirected) adjacency matrix of the semi-infinite paths.

The S-matrix and the amplitudes of $|\mathrm{sc}_{j}(k)\rangle$ within the graph $\widehat{G}$ can be obtained using the following procedure \cite{Childs_Gosset}. Let $\vec{\psi}_j (k)$ be the $m$-component vector containing the amplitudes of $|\mathrm{sc}_{j}(k)\rangle$ at the $m$ vertices of $\widehat{G}$ not connected to the infinite paths, and let
\[
\Psi(k)= 
\begin{pmatrix} 
 \vec{\psi}_1(k) & \vec{\psi}_2(k) & \cdots & \vec{\psi}_N(k)\\
\end{pmatrix}.
\]
 Let $z=e^{ik}$ and define
\[
\gamma(z)=
\begin{pmatrix}
zA-1 & zB^{\dagger}\\
zB & zD-z^2-1 \\
\end{pmatrix}.
\]
Then the S-matrix and the amplitudes $\Psi$ are given by 
\begin{equation}
\begin{pmatrix}
S(k) & 0\\
\frac{1}{z}\Psi(k) & -\frac{1}{z^2} \\
\end{pmatrix}
= -\gamma(z)^{-1}\gamma(z^{-1}).
\label{eq:smatrixgamma}
\end{equation}
An equivalent expression for the S-matrix is 
\begin{align}
S(k) & =-Q(z)^{-1}Q(z^{-1})\label{eq:S_matrix}
\end{align}
where
\begin{align*}
Q(z) = 1-z\left(A+B^{\dagger}\frac{1}{\frac{1}{z}+z-D}B\right).
\end{align*}
Using this expression, one can verify that the graphs described in this paper implement the claimed unitaries. However, it is not clear how to construct a graph that implements a desired unitary. In some cases one can find a suitable graph by performing a numerical search over graphs with a small number of vertices. Graphs with known properties at a given momentum can sometimes be combined (for example, the graph in \fig{halfpiplanhad} was obtained in this manner).

Using the expression 
\[
\gamma(z)^{-1}=\frac{1}{\det\gamma(z)} \adj\gamma(z),
\]
where $\adj\gamma(z)$ is the adjugate matrix of $\gamma(z)$, we see from \eq{smatrixgamma} that each matrix element of $S(k)$ and each matrix element of $\Psi(k)$ is a rational function of $z$. Since $S(k)$ is unitary, its matrix elements (as functions of $z$) have no poles on the unit circle. The following lemma shows that the matrix elements of $\Psi(k)$ also have no poles when $z$ is on the unit circle.

\begin{lemma}\label{lem:defn_scatteringstates}
Given $\widehat{G}$, there exists a constant $\lambda\in\mathbb{R}$ such that $|\langle v|sc_{j}(k)\rangle|<\lambda$ for all $k\in[-\pi,\pi)$, $j\in\{1,\ldots,N\}$, and $v\in\widehat{G}$.
\end{lemma}

\begin{proof}
Note that 
\[
\gamma\left(\frac{1}{z}\right)=\frac{1}{z^{2}}\gamma(z)+\left(\frac{1}{z^{2}}-1\right)\widehat{P}\]
where 
\[
\widehat{P}=\begin{pmatrix}
1 & 0\\
0 & 0\end{pmatrix}
\]
projects onto the $N$ vertices of $\widehat{G}$ attached to semi-infinite paths. Hence 
\[
-\gamma(z)^{-1}\gamma\left(\frac{1}{z}\right)=-\frac{1}{z^{2}}+\left(1-\frac{1}{z^{2}}\right)\gamma(z)^{-1}\widehat{P}.
\]

Let $\{|\psi_{c}\rangle \colon c\in\{1,\ldots,n_{c}\}\}$ be eigenstates of $\widehat{H}$ satisfying $\widehat{P}|\psi_c\rangle=0$, and let this set be an orthonormal basis for the span of all such states. Then
 \[
\left(1-\frac{1}{z^{2}}\right)\gamma(z)^{-1}\widehat{P}=\left(1-\frac{1}{z^{2}}\right)\left(1-\sum_{j=1}^{n_{c}}|\psi_{c}\rangle\langle\psi_{c}|\right)\gamma(z)^{-1}\left(1-\sum_{j=1}^{n_{c}}|\psi_{c}\rangle\langle\psi_{c}|\right)\widehat{P}\]
 since each $|\psi_c\rangle$ is an eigenvector of $\gamma(z)$ and
$\widehat{P}|\psi_{c}\rangle=0$. Reference \cite{Childs_Gosset} shows (in Part 2 of the proof of Theorem 1) that
\[
\det\left(\frac{1}{1-z^{2}}\right)M(z)\neq0\text{ for }|z|=1,
\]
where $M(z)$ is the $(N+m-n_{c})\times(N+m-n_{c})$ matrix of $\gamma(z)$ in the subspace of states orthogonal to the span of $\{|\psi_{c}\rangle \colon c\in\{1,\ldots,n_{c}\}\}$. Therefore 
\begin{align*}
\frac{1}{z}\langle v|\mathrm{sc}_{j}(k)\rangle & =  -\langle v|\gamma(z)^{-1}\gamma\left(\frac{1}{z}\right)|j\rangle\\
 & = \langle v|\left(1-\frac{1}{z^{2}}\right)\left(1-\sum_{j=1}^{n_{c}}|\psi_{c}\rangle\langle\psi_{c}|\right)\gamma(z)^{-1}\left(1-\sum_{j=1}^{n_{c}}|\psi_{c}\rangle\langle\psi_{c}|\right)|j\rangle
\end{align*}
has no poles on the unit circle, and the result follows.\end{proof}

We now establish the delta-function normalization of the scattering states. Let 
\begin{align*}
\Pi_{1} & = \sum_{x=1}^{\infty}\sum_{q=1}^{N}|x,q\rangle\langle x,q|\\
\Pi_{2} & = \mathbb{I}-\sum_{x=2}^{\infty}\sum_{q=1}^{N}|x,q\rangle\langle x,q|\\
\Pi_{3} & = \sum_{q=1}^{N}|1,q\rangle\langle1,q|.\end{align*}
We show that, for $k\in(-\pi,0)$, $p\in(-\pi,0)$, and $i,j\in\{1,\ldots,N\}$,
\begin{equation}
\langle\mathrm{sc}_{i}(p)|\mathrm{sc}_{j}(k)\rangle=\langle\mathrm{sc}_{i}(p)|\Pi_{1}+\Pi_{2}-\Pi_{3}|\mathrm{sc}_{j}(k)\rangle=2\pi\delta_{ij}\delta(k-p).\label{eq:delta}\end{equation}
First write 
\begin{align*}
\langle\mathrm{sc}_{i}(p)|\Pi_{1}|\mathrm{sc}_{j}(k)\rangle & = \sum_{x=1}^{\infty}\sum_{q=1}^{N}(\delta_{iq}e^{ipx}+S_{qi}^{\ast}(p)e^{-ipx})(\delta_{jq}e^{-ikx}+S_{qj}(k)e^{ikx})\\
 & = \frac{1}{2}\left(\delta_{ij}+\sum_{q=1}^{N}S_{qi}^{\ast}(p)S_{qj}(k)\right)\left(\sum_{x=1}^{\infty}e^{i(p-k)x}+\sum_{x=1}^{\infty}e^{-i(p-k)x}\right)\\
 & \quad +\frac{1}{2}\left(\delta_{ij}-\sum_{q=1}^{N}S_{qi}^{\ast}(p)S_{qj}(k)\right)\left(\sum_{x=1}^{\infty}e^{i(p-k)x}-\sum_{x=1}^{\infty}e^{-i(p-k)x}\right)\\
 & \quad +\frac{1}{2}(S_{ji}^{\ast}(p)+S_{ij}(k))\left(\sum_{x=1}^{\infty}e^{-i(p+k)x}+\sum_{x=1}^{\infty}e^{i(p+k)x}\right)\\
 & \quad +\frac{1}{2}(S_{ji}^{\ast}(p)-S_{ij}(k))\left(\sum_{x=1}^{\infty}e^{-i(p+k)x}-\sum_{x=1}^{\infty}e^{i(p+k)x}\right).\end{align*}
We use the following identities for $p,k\in(-\pi,0)$: 
\begin{align*}
\sum_{x=1}^{\infty}e^{i(p-k)x}+\sum_{x=1}^{\infty}e^{-i(p-k)x} & = 2\pi\delta(p-k)-1 \\
\sum_{x=1}^{\infty}e^{i(p+k)x}+\sum_{x=1}^{\infty}e^{-i(p+k)x} &=-1 \\
\sum_{x=1}^{\infty}e^{i(p-k)x}-\sum_{x=1}^{\infty}e^{-i(p-k)x} & = i\cot\left(\frac{p-k}{2}\right) \\
\sum_{x=1}^{\infty}e^{i(p+k)x}-\sum_{x=1}^{\infty}e^{-i(p+k)x} &= i\cot\left(\frac{p+k}{2}\right).
\end{align*}
These identities hold when both sides are integrated against a smooth function of $p$ and $k$. Substituting, we get
\begin{align}
\langle\mathrm{sc}_{i}(p)|\Pi_{1}|\mathrm{sc}_{j}(k)\rangle & =  2\pi\delta_{ij}\delta(p-k)+\delta_{ij}\left(\frac{i}{2}\cot\left(\frac{p-k}{2}\right)-\frac{1}{2}\right)\nonumber\\
&\quad+\sum_{q=1}^{N}S_{qi}^{\ast}(p)S_{qj}(k)\left(-\frac{i}{2}\cot\left(\frac{p-k}{2}\right)-\frac{1}{2}\right)\nonumber \\
 &\quad+S_{ji}^{\ast}(p)\left(-\frac{1}{2}-\frac{i}{2}\cot\left(\frac{p+k}{2}\right)\right)\nonumber\\
&\quad+S_{ij}(k)\left(-\frac{1}{2}+\frac{i}{2}\cot\left(\frac{p+k}{2}\right)\right)\label{eq:pi1}
\end{align}
where we used unitarity of the $S$-matrix to simplify the first term.
Now turning to $\Pi_{2}$ we have \[
\langle\mathrm{sc}_{i}(p)|H\Pi_{2}|\mathrm{sc}_{j}(k)\rangle=2\cos(p)\langle\mathrm{sc}_{i}(p)|\Pi_{2}|\mathrm{sc}_{j}(k)\rangle\]
 and 
\begin{align*}
\langle\mathrm{sc}_{i}(p)|H\Pi_{2}|\mathrm{sc}_{j}(k)\rangle & = \langle\mathrm{sc}_{i}(p)|\bigg(2\cos(k)\Pi_{2}|\mathrm{sc}_{j}(k)\rangle+\sum_{q=1}^{N}(e^{-ik}\delta_{qj}+S_{qj}(k)e^{ik})|2,q\rangle\\
&\quad -\sum_{q=1}^{N}(e^{-2ik}\delta_{qj}+S_{qj}(k)e^{2ik})|1,q\rangle\bigg).
\end{align*}
 Using these two equations we get 
\begin{align*}
(2\cos(p)-2\cos(k))\langle\mathrm{sc}_{i}(p)|\Pi_{2}|\mathrm{sc}_{j}(k)\rangle & =  \delta_{ij} (e^{2ip-ik}-e^{-2ik+ip})+S_{ji}^{\ast}(p)(e^{-2ip-ik}-e^{-2ik-ip})\\
 & \quad +S_{ij}(k)(e^{2ip+ik}-e^{2ik+ip})\\
& \quad +\sum_{q=1}^{N}S_{qi}^{\ast}(p)S_{qj}(k)(e^{-2ip+ik}-e^{2ik-ip}).
\end{align*}
 Noting that
\[
\langle\mathrm{sc}_{i}(p)|\Pi_{3}|\mathrm{sc}_{j}(k)\rangle=\sum_{q=1}^{N}(\delta_{iq}e^{ip}+S_{qi}^{\ast}(p)e^{-ip})(\delta_{jq}e^{-ik}+S_{qj}(k)e^{ik}),
\]
we have 
\begin{align}
\langle\mathrm{sc}_{i}(p)|\Pi_{2}-\Pi_{3}|\mathrm{sc}_{j}(k)\rangle & = \delta_{ij}\left(\frac{e^{2ip-ik}-e^{-2ik+ip}}{2\cos(p)-2\cos(k)}-e^{ip-ik}\right)\nonumber\\
&\quad+S_{ji}^{\ast}(p)\left(\frac{e^{-2ip-ik}-e^{-2ik-ip}}{2\cos(p)-2\cos(k)}-e^{-ip-ik}\right)\nonumber \\
 & \quad +S_{ij}(k)\left(\frac{e^{2ip+ik}-e^{2ik+ip}}{2\cos(p)-2\cos(k)}-e^{ip+ik}\right)\nonumber\\
&\quad+\sum_{q=1}^{N}S_{qi}^{\ast}(p)S_{qj}(k)\left(\frac{e^{-2ip+ik}-e^{2ik-ip}}{2\cos(p)-2\cos(k)}-e^{-ip+ik}\right)\nonumber \\
  & = \delta_{ij}\left(\frac{1}{2}-\frac{i}{2}\cot\left(\frac{p-k}{2}\right)\right)+S_{ji}^{\ast}(p)\left(\frac{1}{2}+\frac{i}{2}\cot\left(\frac{p+k}{2}\right)\right)\nonumber\\
&\quad+S_{ij}(k)\left(\frac{1}{2}-\frac{i}{2}\cot\left(\frac{p+k}{2}\right)\right)\nonumber\\
&\quad+\sum_{q=1}^{N}S_{qi}^{\ast}(p)S_{qj}(k)\left(\frac{1}{2}+\frac{i}{2}\cot\left(\frac{p-k}{2}\right)\right).
\label{eq:pi2_pi3}
\end{align}
 Adding equation \eq{pi1} to equation \eq{pi2_pi3} gives
equation \eq{delta}.

\section{Two-particle scattering states}\label{sec:twopart_scat}

Here we derive scattering states of the two-particle quantum walk on an infinite path. We write the Hamiltonian in the basis $\ket{x,y}$, where $x$ denotes the location of the first particle and $y$ denotes the location of the second particle, with the understanding that bosonic states are symmetrized and fermionic states are antisymmetrized. The Hamiltonian \eq{dist_ham} can be written as
\begin{equation}
  H^{(2)} = H^{(1)}_x \otimes \II_y + \II_x \otimes H^{(1)}_y + \sum_{x,y\in\integer} 
     \mathcal{V}(|x-y|) \, \ket{x,y}\bra{x,y}\label{eq:twoham}
\end{equation}
where $\mathcal{V}$ corresponds to the interaction term $\mathcal{U}$ and (with a slight abuse of notation) the subscript indicates which variable is acted on. Here
\[
  H^{(1)} = \sum_{x\in\integer} \ket{x+1}\bra{x} + \ket{x}\bra{x+1}
\]
is the adjacency matrix of an infinite path. Our assumption that $\mathcal{U}$ has finite range $C$ means that $\mathcal{V}(r)= 0$ for $r>C$.  

The scattering states we are interested in provide information about the dynamics of two particles initially prepared in spatially separated wave packets moving toward each other along the path with momenta $k_1\in(-\pi,0) $ and $k_2\in (0,\pi)$.

We derive scattering eigenstates of this Hamiltonian by transforming to the new variables $s = x+y$ and $r = x-y$ and exploiting translation symmetry.  Here the allowed values $(s,r)$ range over the pairs of integers where either both are even or both are odd.  Writing states in this basis as $|s;r\rangle$, the Hamiltonian takes the form
\begin{equation}
  H^{(1)}_s\otimes H^{(1)}_r+ \II_s\otimes \sum_{r\in \integer} \mathcal{V}(|r|) \, \ket{r}\bra{r}.
\label{eq:twopart_ham}
\end{equation}
For each $p_1\in (-\pi,\pi)$ and $p_2\in(0,\pi)$ there is a scattering eigenstate $|\mathrm{sc}(p_1;p_2)\rangle$ of the form
\[
\langle s;r |\mathrm{sc}(p_1;p_2)\rangle=e^{-ip_1 s/2} \langle r|\psi(p_1;p_2)\rangle,
\]
where the state $|\psi(p_1;p_2)\rangle$ can be viewed as an effective single-particle scattering state of the Hamiltonian
\begin{equation}
 2\cos\left(\frac{p_1}{2}\right) H^{(1)}_r + \sum_{r\in \integer} \mathcal{V}(|r|) \, \ket{r}\bra{r}\label{eq:vr_eqn}
\end{equation}
with eigenvalue $4 \cos( p_1/2) \cos(p_2)$.  For a given $\mathcal{V}$, the state $|\psi(p_1;p_2)\rangle$ can be obtained explicitly by solving a set of linear equations (see for example \cite{Childs_Gosset}). It has the form
\begin{equation}
\langle r|\psi (p_1;p_2)\rangle= \begin{cases}  e^{-i p_2 r} + R(p_1,p_2) e^{i p_2 r} &  \text{if } r \leq -C\\
  	f(p_1,p_2,r) &  \text{if } |r| < C\\
  	T(p_1,p_2) e^{- i p_2 r}  & \text{if } r \geq C\end{cases}
\label{eq:psip1p2}
\end{equation}
for $p_2\in (0,\pi)$. Here the reflection and transmission coefficients $R$ and $T$ and the amplitudes of the scattering state for $|r|<C$ (described by the function $f$) depend on both momenta as well as the interaction $\mathcal{V}$.  With $R$, $T$, and $f$ chosen appropriately, the state $|\mathrm{sc}(p_1;p_2)\rangle$ is an eigenstate of $H^{(2)}$ with eigenvalue $4\cos(p_1/2)\cos(p_2)$.

Since $\mathcal{V}(|r|)$ is an even function of $r$, we can also define scattering states for $p_2\in (-\pi,0)$ by
\[
\langle s;r|\mathrm{sc}(p_1;p_2)\rangle=\langle s;-r|\mathrm{sc}(p_1;-p_2)\rangle.
\]
These other states are obtained by swapping $x$ and $y$, corresponding to interchanging the two particles.

The states $\{|\mathrm{sc}(p_1;p_2)\rangle \colon p_1\in (-\pi,\pi),\,p_2\in(-\pi,0)\cup(0,\pi)\}$ are (delta-function) orthonormal:
\begin{align*}
\langle  \mathrm{sc}(p_1';p_2')|\mathrm{sc}(p_1;p_2)\rangle &= \langle \mathrm{sc}(p_1'; p_2')|\left(\sum_{\text{$r,s$ even}}|r\rangle\langle r| \otimes |s\rangle \langle s| \right)|\mathrm{sc}(p_1;p_2)\rangle\\
&\quad + \langle \mathrm{sc}(p_1'; p_2')|\left(\sum_{\text{$r,s$ odd}}|r\rangle \langle r|\otimes  |s\rangle \langle s|\right)|\mathrm{sc}(p_1;p_2)\rangle\\
&= \sum_{\text{$s$ even}} e^{-i(p_1-p_1') {s}/{2}}\sum_{\text{$r$ even}}\langle \psi(p_1';p_2')|r\rangle\langle r|\psi(p_1;p_2)\rangle \\
& \quad + \sum_{\text{$s$ odd}} e^{-i(p_1-p_1') {s}/{2}}\sum_{\text{$r$ odd}}\langle  \psi(p_1';p_2')|r\rangle\langle r|\psi(p_1;p_2)\rangle\\
&= 2\pi \delta(p_1-p_1') \sum_{r=-\infty}^{\infty}\langle \psi(p_1;p_2')|r\rangle\langle r|\psi(p_1;p_2)\rangle\\
&= 4\pi^2 \delta(p_1-p_1')\delta(p_2-p_2')
\end{align*}
where in the last step we used the fact that $\langle\psi(p_1;p_2')|\psi(p_1;p_2)\rangle=2\pi\delta(p_2-p_2')$.
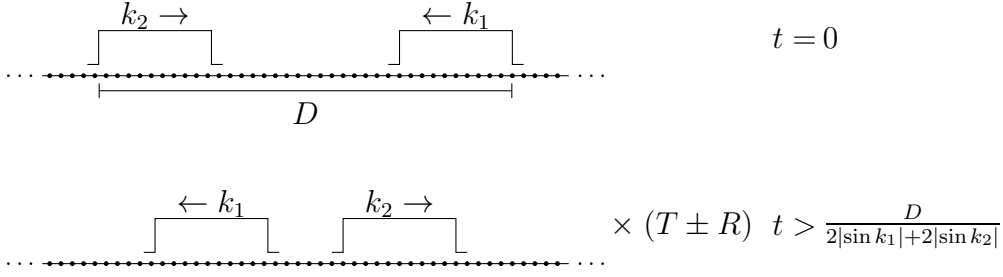
\begin{figure}
\centering
\capstart
\begin{tikzpicture}[label distance= -6pt,
    verts/.style={circle,draw=black,fill=black,inner sep=.5pt,minimum size=0pt},
    dots/.style={circle,fill=black,inner sep=.25pt,minimum width=0pt}]
  \draw (0,0) -- (7,0);

  \draw (3.15,.15) -- (3,.15) -- (3,.6) -- (1.5,.6) 
     -- (1.5,.15) -- (1.35,.15);
  
  \draw (3.85,.15) -- (4,.15) -- (4,.6) -- (5.5,.6) 
     -- (5.5,.15) -- (5.65,.15);
     
  \node at (2.25,.8) {$\leftarrow k_1$};
  \node at (4.75,.8) {$k_2\rightarrow$};

  \node at (8.5,.5) {$\times \; (T\pm R)$};

  \node at (10,.5)[label=right:$\frac{D}{2|{\sin k_1}|+ 2|{\sin k_2}|}$]{$t>$};

  \foreach \x in {.1,.25,...,6.9}
  \node at (\x ,0) [verts] {};
  
  \begin{scope}[yshift=2.5cm]
    \draw (0,0) -- (7,0);

    \draw[xshift=-.75cm] (3.15,.15) -- (3,.15) -- (3,.6) -- (1.5,.6) 
       -- (1.5,.15) -- (1.35,.15);
  
    \draw[xshift=.75cm] (3.85,.15) -- (4,.15) -- (4,.6) -- (5.5,.6) 
       -- (5.5,.15) -- (5.65,.15);
 
    \draw [|-|] (.75,-.2) to node[below] {$D$} (6.25,-.2);
    
    \node at (1.5,.8) {$k_2\rightarrow$};
    \node at (5.5,.8) {$\leftarrow k_1$};

    \node at (10,.5)[label=right:$0$] {$t=$};

    \foreach \x in {.1,.25,...,6.9}
    \node at (\x ,0) [verts] {};

  \end{scope}

  \foreach \xsh in {-0.45cm, 7.15cm}{
  \foreach \ysh in {0cm, 2.5cm}{
    \begin{scope}[xshift=\xsh,yshift=\ysh]
      \node at (0,0) [dots]{};
      \node at (0.15,0) [dots] {};
      \node at (0.3,0) [dots]{};
    \end{scope}
  }}

\end{tikzpicture}
\caption{Scattering of two particles on an infinite path.}
\label{fig:wte}
\end{figure}
	
To construct bosonic or fermionic scattering states, we symmetrize or antisymmetrize as follows. For $p_1\in (-\pi,\pi)$ and $p_2\in (0,\pi)$, we define
\[
  \ket{\mathrm{sc}(p_1;p_2)}_\pm = \frac{1}{\sqrt2}(\ket{\mathrm{sc}(p_1;p_2)} \pm \ket{\mathrm{sc}(p_1;-p_2)}).
\]
Then
\begin{align}
    \braket{s;r}{\mathrm{sc}(p_1;p_2)}_\pm
      &= \f{1}{\sqrt{2}}e^{-i p_1 s/2} \begin{cases}  e^{-i p_2 r} \pm e^{i\theta_{\pm}(p_1,p_2)} e^{i p_2 r} &  \text{if } r \leq -C\\
  	f(p_1,p_2,r) \pm f(p_1,p_2,-r) & \text{if }  |r| < C\\
  	e^{i\theta_{\pm}(p_1,p_2)}e^{- i p_2 r} \pm e^{i p_2 r}  & \text{if } r \geq C\end{cases}
\label{eq:symscatter}
\end{align}
where $\theta_{\pm}(p_1,p_2)$ is a real function defined through
\begin{equation}
e^{i\theta_{\pm}(p_1,p_2)}= T(p_1,p_2)\pm R(p_1,p_2). \label{eq:delta_pm}
\end{equation}
Note that $|T\pm R| = 1$; this follows from the potential $\mathcal{V}(|r|)$ being even in $r$ and from unitarity of the S-matrix.  These eigenstates allow us to understand what happens when two particles with momenta $k_1\in(-\pi,0)$ and $k_2\in(0,\pi)$ move toward each other. Here $p_1=-k_1-k_2$ and $p_2=(k_2-k_1)/2$.  Recall (from the main text of the paper) that we defined $e^{i\theta}$ to be the phase acquired by the two-particle wavefunction when $k_1=-{\pi}/{2}$ and $k_2={\pi}/{4}$ ($\theta$ depends implicitly on the interaction $\mathcal{V}$ and the particle type), so $\theta=\theta_\pm ({\pi}/{4},{3\pi}/{8})$ .

For $|r|\geq C$ the scattering state is a sum of two terms, one corresponding to the two particles moving toward each other and one corresponding to the two particles moving apart after scattering. The outgoing term has a  phase of $T\pm R$ relative to the incoming term (as depicted in \fig{wte}). This phase arises from the interaction between the two particles.

For example, consider the Bose-Hubbard model, where $\mathcal{V}(|r|) = U\delta_{r,0}$. Here $C=0$ and $T=1+R$.  In this case the scattering state $|\mathrm{sc}(p_1;p_2)\rangle_+$ is
\[
\langle x,y|\mathrm{sc}(p_1;p_2)\rangle_+=\frac{1}{\sqrt{2}}e^{-ip_1 \left(\frac{x+y}{2}\right)}\left(e^{ip_2 |x-y|}+e^{i\theta_+(p_1,p_2)}e^{-ip_2 |x-y|}\right).
\]
The first term describes the two particles moving toward each other and the second term describes them moving away from each other. To solve for the applied phase $e^{i\theta_+(p_1,p_2)}$ we look at the eigenvalue equation for $|\psi(p_1;p_2)\rangle$ at $r=0$. This gives
\[
  R(p_1,p_2) =- \f{U}{U - 4i\cos({p_1}/{2})\sin(p_2)}.
\]
So for the Bose-Hubbard model,
\[
  e^{i \theta_{+} (p_1,p_2)} = T(p_1,p_2) + R(p_1,p_2) = - \frac{ U + 4 i \cos({p_1}/{2}) \sin(p_2)}{U - 4 i \cos({p_1}/{2}) \sin(p_2)} =  \frac{2 \left(\sin(k_2) - \sin(k_1)\right) - i U}{2 \left(\sin(k_2) - \sin(k_1)\right) + i U}.
\]
For example, if $U = 2+\sqrt{2}$ then two particles with momenta $k_1 =-{ \pi}/{2}$ and $k_2={\pi}/{4}$ acquire a phase of $e^{-i\pi/2}= -i$ after scattering.

For a multi-particle quantum walk with nearest-neighbor interactions, $\mathcal{V}(|r|)=U\delta_{|r|,1}$ and $C=1$.  In this case the eigenvalue equations for $|\psi(p_1;p_2)\rangle$ at $r=-1$, $r=1$, and $r=0$ are
\begin{align*}
 4 \cos\left(\frac{p_1}{2}\right)  \cos(p_2) ( e^{i p_2} + R(p_1,p_2) e^{-i p_2} ) &= U ( e^{i p_2} + R(p_1,p_2) e^{-i p_2}) \\
& \quad + 2\cos\left(\frac{p_1}{2}\right) \left( e^{2i p_2} + R(p_1,p_2) e^{-2i p_2}+f(p_1,p_2,0)\right) \\
 4 \cos\left(\frac{p_1}{2}\right)  \cos(p_2) T(p_1,p_2) e^{-ip_2} & =UT(p_1,p_2)e^{-ip_2}\\
& \quad +2\cos \left(\frac{p_1}{2}\right)\left(f(p_1,p_2,0)+T(p_1,p_2)e^{-2ip_2}\right)\\
2 \cos(p_2) f(p_1,p_2,0) &=T(p_1,p_2)e^{-ip_2}+e^{ip_2}+R(p_1,p_2)e^{-ip_2},
\end{align*}
respectively.

Solving these equations for $R$, $T$, and $f(p_1,p_2,0)$, we can construct the corresponding scattering states for bosons, fermions, or distinguishable particles (for more on the last case, see \sec{distinguishable}). Unlike the case of the Bose-Hubbard model, we may not have $1+R=T$. For example, when $U=-2-\sqrt{2}$, $p_1={\pi}/{4}$, and $p_2={3\pi}/{8}$, we get $R=0$ and $T=i$ (see \sec{distinguishable}).

\section{Refinements of the universality construction}

In this section we present two refinements of our scheme. We first show how our scheme can be modified to use planar graphs of maximum degree four.  We then give a universality construction using \textit{distinguishable} particles with nearest-neighbor interactions.

\subsection{Making the graph planar}\label{sec:planar}

The example in \fig{example} shows that the graphs in the scheme described by the main text of our paper may not be planar: the mediator qubit can interact with any of the computational qubits, so the vertical paths for the $\CD$ gate cross other paths in the graph.  Furthermore, both the graph used to implement the Hadamard gate on the mediator qubit (\fig{halfpihad}) and the graph used to implement the $\CD$ gate (\fig{onepsplit}(b)) can lead to a nonplanar overall graph when input and output paths are attached in the prescribed manner. In this section we describe how to modify the scheme so that the resulting graph is planar and has maximum degree 4.

The first simple modification is to replace the graph from \fig{halfpihad} with a planar graph (with input and output vertices on the same face in the correct relative positions) that also implements a Hadamard gate on the mediator qubit.  The graph in \fig{halfpiplanhad} does the trick: its S-matrix at momentum $-{\pi}/{2}$ has the form \eq{S_matrix_circuit} with lower left submatrix
\[
U_{\Had'}= \frac{1}{\sqrt{2}}e^{-3\pi i/4}\begin{pmatrix} 1 & 1\\
     1 & -1 \end{pmatrix}.
\]
The smaller maximum degree (4 instead of 5) and planarity of this graph come at the expense of increasing the number of vertices (as compared to the graph in \fig{halfpihad}).

\begin{figure}
\centering
\capstart
\begin{tikzpicture}
  [ scale = .4,
  	thin,
    inner/.style={circle,draw=black!100,fill=black!100,inner sep = 1.25pt},
    attach/.style={circle,draw=black!100,fill=black!0,thin,inner sep=1.25pt},
    cross/.style={draw=white,double=black,thin}]

 \node at (0,0)[inner] {};
 \node at (1,0)[inner] {};
 \node at (2,0)[inner] {};
 \node at (0,1)[inner] {};
 \node at (2,1)[inner] {};
 \node at (0,2)[inner] {};
 \node at (1,2)[inner] {};
 \node at (2,2)[inner] {};
 \node at (0,3)[inner] {};
 \node at (1,3)[inner] {};
 \node at (2,3)[inner] {};
 \node at (3,1)[inner]{};
 \node at (4,1)[inner]{};
 \node at (3,2)[inner]{};
 \node at (4,2)[inner]{};
 \node at (3,0)[inner]{};
 \node (a) at (3.5,2.714)[inner]{};
 \node (b) at (4.714,1.5)[inner]{};
 \node at (0,-1) [inner]{};
 \node at (1,-1) [inner]{};
 \node at (2,-1) [inner]{};
 \node at (1,-2) [inner]{};
 \node at (2,-2) [inner]{};
 \node (c) at (2.714,-1.5) [inner]{};
 \node (d) at (1.5,-2.714) [inner]{};
 \node at (0,4) [inner]{};

 \draw (-1,0) -- (4,0);
 \draw (-1,4) -- (0,4);
 \draw (0,4) -- (0,-1);
 \draw (0,2) -- (2,2);
 \draw (0,3) -- (2,3);
 \draw (0,1) -- (4,1);
 \draw (1,3) -- (1,2);
 \draw (3,0) -- (2,1) -- (2,3);
 \draw (0,1) -- (1,0);
 \draw (2,3) -- (4,4);
 
 \draw (0,-1) -- (2,-1) -- (c) -- (2,-2) -- (2,-1);
 \draw (1,-1) -- (1,-2) -- (2,-2) -- (d) -- (1,-2);
 
 \draw (4,1) -- (b) -- (4,2) -- (4,1);
 \draw (3,1) -- (3,2) -- (a) -- (4,2) -- (3,2);

 \node at (-1,0) [attach,label=left:$1_{\text{in}}$]{};
 \node at (-1,4) [attach,label=left:$0_{\text{in}}$]{};
 \node at (4,4) [attach,label=right:$0_{\text{out}}$]{};
 \node at (4,0) [attach,label=right:$1_{\text{out}}$]{};
 
\end{tikzpicture}
\caption{A planar graph that implements a Hadamard gate at momentum $-{\pi}/{2}$.}
\label{fig:halfpiplanhad}
\end{figure}
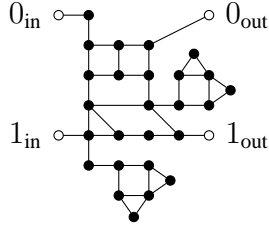

As a second modification, we introduce additional mediator qubits.  Throughout the graph, we arrange the input and output paths for the computational qubits vertically from $1,\ldots, n$ with the path corresponding to logical $0$ always above the path corresponding to logical $1$. For each $i\in\{1,\ldots,n-1\}$ we place a mediator qubit labeled $\med(i)$ between computational qubits $i$ and $i+1$.  We only perform two-qubit gates between adjacent qubits throughout the computation (i.e., mediator qubit $\med(i)$ only interacts with logical qubits $i$ and $i+1$).

To implement two-qubit gates in a planar manner, we use the graph shown in \fig{PlanCPGate}. This graph is obtained by concatenating two $\CD$ graphs and uncrossing paths to make the drawing planar. We only use this gate between adjacent encoded qubits, one of which is a mediator qubit and one of which is a computational qubit.  Note that this graph involves two adjacent paths (path 1 of the top encoded qubit and path 0 of the bottom encoded qubit) as opposed to the two 1 paths in the $\CD$ gate in \fig{onepsplit}. The resulting logical gate is $(\CD)^2$ conjugated by an X gate on the bottom qubit.  More explicitly, if we interact computational qubit $i$ and mediator $\med(i)$, we implement $\XX_{\med(i)} (\CD)^2_{i,\med(i)}\XX_{\med(i)}$, whereas if we interact mediator $\med(i)$ and computational qubit $i+1$, we implement $\XX_{i+1} (\CD)^2_{\med(i),i+1} \XX_{i+1}$.  Applying this gate $a$ times, where $e^{i\theta a} \approx \pm i$, we can approximate the gates $\XX_{\med(i)} \CZ_{i,\med(i)}\XX_{\med(i)}$ and $\XX_{i+1} \CZ_{\med(i),i+1} \XX_{i+1}$.

Using these gates, we can perform a controlled-Z gate between computational qubits $i$ and $i+1$ by the following sequence:
\begin{align*}
  \CZ_{i,i+1} \ket{a_i, b_{i+1},0_{\med(i)}}
     &= \CNOT_{i+1,\med(i)}\CZ_{i,\med(i)}\CNOT_{i+1,\med(i)}\ket{a_i, b_{i+1},0_{\med(i)}}\\
     &=\XX_{\med(i)}\XX_{i+1}\CNOT_{i+1,\med(i)}\CZ_{i,\med(i)}\CNOT_{i+1,\med(i)} 
        \XX_{\med(i)}\XX_{i+1}\ket{a_i, b_{i+1},0_{\med(i)}}\\
     &= \Had_{\med(i)}\left( \XX_{i+1} \CZ_{\med(i),i+1} \XX_{i+1}\right) \Had_{\med(i)} 
        \left( \XX_{\med(i)} \CZ_{i,\med(i)}\XX_{\med(i)}\right) \Had_{\med(i)}\\
     &\qquad \left( \XX_{i+1} \CZ_{\med(i),i+1} \XX_{i+1}\right) \Had_{\med(i)}\ket{a_i, b_{i+1},0_{\med(i)}}.
\end{align*}

To implement a CZ gate between arbitrary encoded qubits, we use these $\CZ_{i,i+1}$ and one-qubit gates to implement a $\mathrm{SWAP}_{i,i+1}$ gate, facilitating movement of encoded qubits. To implement a $\CZ_{i,j}$ gate, we use SWAP gates to move the information encoded in qubit $i$ to qubit $j-1$ or $j+1$, perform the required CZ gate, and finally use SWAP gates to return qubit $i$ to its original position.

Note that every graph implementing a gate preserves the relative position of each path for both the input and the output (e.g., qubit 1 remains above qubit $\med(1)$ and each 0 path remains above each 1 path), so concatenation of the graphs preserves planarity.  As each individual graph is planar, and as concatenating two graphs preserves planarity, the overall graph simulating the quantum circuit is planar.

\begin{figure}
\centering
\capstart
\begin{tikzpicture}[
  scale=2,
  split/.style={circle,draw=black,fill=white,
    inner sep=0pt, minimum size= 4mm},
  attach/.style={circle,draw=black,fill=white,
    inner sep=1pt, minimum size=0},
  vert/.style={circle,fill=black,inner sep=.7pt,minimum size=0}]

  \foreach \x in {0,.06,...,4}{
  \foreach \y in {.25, .5, 1.5, 1.75}{
    \node at (\x,\y) [vert] {};
  }}
  
  \foreach \x in {1.5,2.5}{
  \foreach \y in {.5,.56,...,1.5}{
    \node at (\x,\y) [vert]{};
  }}

  \draw (0,.25) node[label=left:$1_{i+1,\text{in}}$] {} 
     -- (4,.25) node[label=right:$1_{i+1,\text{out}}$] {};

  \node (s1) at (1.5,.5) [split] {};
  \node (s2) at (2.5,.5) [split] {};
  \node (s3) at (1.5,1.5) [split] {};
  \node (s4) at (2.5,1.5) [split] {};

  \draw (0,1.75) node[label=left:$0_{\med(i),\text{in}}$] {}
     -- (4,1.75) node[label=right:$0_{\med(i),\text{out}}$] {};

  \draw (0,.5) node[label=left:$0_{i+1,\text{in}}$] {}
     -- (s1.west);
  \draw (0,1.5) node[label=left:$1_{\med(i),\text{in}}$] {}
     -- (s3.west);

  \draw (s4.east) -- (4,1.5) node[label=right:$1_{\med(i),\text{out}}$]{};
  \draw (s2.east) -- (4,0.5) node[label=right:$0_{i+1,\text{out}}$]{};

  \draw (s1.east) -- (s2.west);
  \draw (s3.east) -- (s4.west);
  \draw (s1.north) -- (s3.south);
  \draw (s2.north) -- (s4.south);

  \draw (s1.west) to[out=0,in=-90] (s1.north)[line width=.5pt];
  \draw (s1.east) to[out=-180,in=-90] (s1.north) [line width =1.5pt];
  \draw (s1.east) to[out=-180,in=-90] (s1.north) [white,line width=.5pt];

  \draw (s2.east) to[out=-180,in=-90] (s2.north)[line width=.5pt];
  \draw (s2.west) to[out=0,in=-90] (s2.north)[line width=1.5pt];  
  \draw (s2.west) to[out=0,in=-90] (s2.north)[white,line width=.5pt];

  \draw (s3.east) to[out=-180,in=90] (s3.south)[line width=.5pt];
  \draw (s3.west) to[out=0,in=90] (s3.south)[line width=1.5pt];
  \draw (s3.west) to[out=0,in=90] (s3.south)[white,line width=.5pt];

  \draw (s4.west) to[out=0,in=90] (s4.south)[line width=.5pt];
  \draw (s4.east) to[out=-180,in=90] (s4.south)[line width=1.5pt];
  \draw (s4.east) to[out=-180,in=90] (s4.south)[white,line width=.5pt];

  \foreach \n in {s1, s2, s3, s4}{
    \node at (\n.west) [attach]{};
    \node at (\n.east) [attach]{};
  }
  
  \node at (s1.north) [attach]{};
  \node at (s2.north) [attach]{};
  \node at (s3.south) [attach]{};
  \node at (s4.south) [attach]{};
  
  \foreach \x in {0,4}{
  \foreach \y in {.25, .5, 1.5, 1.75}{
    \node at (\x,\y) [attach]{};
  }}
\end{tikzpicture}
\caption{The planar entangling gate between adjacent encoded qubits $\med(i)$ and $i+1$. This graph implements the unitary $\XX_{i+1} (\CD)^2_{\med(i),i+1} \XX_{i+1}$. A similar graph implements the unitary $\XX_{\med(i)} (\CD)^2_{i,\med(i)}\XX_{\med(i)}$ between adjacent encoded qubits $i$ and $\med(i)$.}
\label{fig:PlanCPGate}
\end{figure}
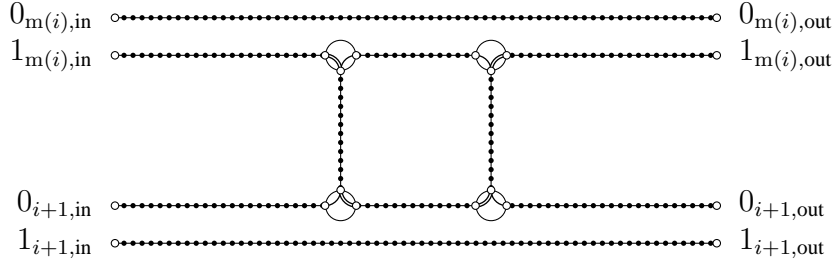
\subsection{Distinguishable particles} 
\label{sec:distinguishable}

So far we have focused on the case of indistinguishable particles.  However, we can also perform universal quantum computation with distinguishable particles, provided the interaction has an appropriate form.

For distinguishable particles we use the same encoding of qubits as before (computational qubits have momentum $-{\pi}/{4}$ and mediator(s) have momentum $-{\pi}/{2}$), except that now each qubit is associated with a specific particle (e.g., computational qubit 1 is associated with particle 1). Since the graphs implementing single-qubit gates only make use of single-particle scattering, which is unaffected by particle statistics, we can use the same graphs for distinguishable particles.  As such, we need only examine the implementation of the $\CD$ gate to see how our construction must be modified. In this section we show that with a simple nearest-neighbor interaction we can make our scheme work for distinguishable particles by carefully choosing the strength of the interaction term in the Hamiltonian, but with no other modifications.

When two indistinguishable particles of momenta $k_1$ and $k_2$ scatter on an infinite path, there is no distinction between the final state where the particles reflect off of each other (exchanging momenta) and where the particles transmit through one another. Thus, after scattering, the global phase of the wave function is multiplied by a factor $T\pm R$, the sum of the amplitude to transmit and the amplitude to reflect (or the difference if the particles are fermions). For any interaction potential, $|T\pm R|=1$, and in most cases the applied phase is nontrivial and can be used for universal computation within our scheme. 

In contrast, when two \textit{distinguishable} particles of momenta $k_1$ and $k_2$ scatter on an infinite path, there are two distinct outgoing states:  one corresponding to the case where the two particles reflect and one where the particles transmit.  We circumvent this potential problem by choosing the interaction strength so that the transmission probability for two-particle scattering at momenta ${\pi}/{4}$ and $-{\pi}/{2}$ is $1$ (forcing $R=0$), yet $T\neq 1$ (so that $T$ is a nontrivial phase). With such a choice, the graph implementing the $\CD$ gates preserves our encoding of qubits. In other words, if encoded qubit 1 is associated with particle 1 before applying the gate, then it is still associated with particle 1 after applying the gate (and similarly for the second qubit involved in the gate). We can then use the same graph as before to implement the controlled phase gate between encoded qubits.

Consider the nearest-neighbor Hamiltonian with $\mathcal{U}_{ij}(\hat{n}_i,\hat{n}_j)=U\delta_{i,j\in E(G)} \hat{n}_i \hat{n}_j$. For two particles on an infinite path, this is \eq{twoham} with $\mathcal{V}(|r|)=U\delta_{|r|,1}$. The reflection coefficient $R(p_1,p_2)$  for $p_1={\pi}/{4}$ and $p_2 = {3\pi}/{8}$ is
\[
  R\left(\frac{\pi}{4},\frac{3\pi}{8}\right) =\frac{-2U\left(\sqrt{2}+(\sqrt{2}-1)U\right)}{(\sqrt{2} - 2)(1+i) U^2 - 4 U + 2 i (\sqrt{2} + 2)}.
\]

Our goal is to choose $U$ so that $R=0$ and $T$ is a nontrivial phase. The values of $U$ that set $R=0$ are  $U = 0$ or $U = -2 - \sqrt{2}$.  The solution $U = 0$ corresponds to no interaction and the trivial phase $T = 1$ which is not sufficient for universal computation within our scheme. Choosing $U=-2-\sqrt{2}$ sets $T= i$ which allows us to perform a CP gate. 

We expect that other types of multi-particle quantum walk with distinguishable particles can also be used for universal computation. However, unlike in the case of indistinguishable particles, the interaction term may have to be tuned to satisfy the conditions $R=0$ and $T\neq1$  as was the case here.  For some interactions, it may not be possible to satisfy these two requirements.  For example, for a model where interactions only occur when particles occupy the same site (such as in the Bose-Hubbard model) we have $1+R = T$ for all momenta, so the transmission amplitude is trivial whenever $R=0$.

Note that it may be possible to implement an entangling two-qubit gate in other ways.  For example, some interactions may allow two-particle scattering with $T=0$ and $R=i$, in which case the graph shown in \fig{PlanCPGate} preserves the encoding of qubits and implements such a gate.

\section{Detailed description of the scheme}\label{sec:description}

In this section we fill in all of the details of the scheme outlined in the main text of this paper. We specify the initial state, the graph used to perform an $n$-qubit quantum computation, and the evolution time. These are all specified as a function of a single parameter $L\in\natural$. We show that by taking $L=\poly(n,g)$ we can achieve an arbitrarily small error in our simulation of a given $g$-gate quantum circuit. The resulting graph and evolution time are both polynomially large in $n$ and $g$. For simplicity, we discuss the case of indistinguishable particles and we concentrate on the nonplanar scheme using one mediator qubit, but similar results clearly apply to the planar scheme and for the case of distinguishable particles as discussed in \sec{planar} and \sec{distinguishable}.

\subsection{Building the graph}
\label{sec:build}
The graph corresponding to a given circuit is built by piecing together subgraphs (which we call \emph{blocks}) that implement gates. These blocks are of the forms shown in \fig{squint} (which we call a \emph{block of type I}) and \fig{CPintall} (which we call a \emph{block of type II}).  

A block of type I applies single-qubit gates $U_1, \ldots,U_n$ to the computational qubits and a single-qubit gate $V_\med$ to the mediator, all acting in parallel. The unitaries $U_1, \ldots,U_n$ may be phase gates, basis-changing gates, or identity gates, up to a global phase. The unitary $V_\med$ is either the identity gate or the Hadamard gate, up to a global phase. The circles in \fig{squint} are replaced by corresponding subgraphs implementing the single-qubit gates.  To simplify our analysis in \sec{block_by_block}, we replace the circles with subgraphs with the property that the shortest path between the two output vertices (or input vertices) of the subgraph is greater than $C$, the interaction range of the Hamiltonian. To implement the identity gate in this manner we use a graph that connects each input vertex to the corresponding output vertex by an edge. Similarly, this condition is automatically satisfied by the subgraph given in \fig{onepphase}(a) for the phase gate (for any $C$), and is satisfied by the single-qubit subgraphs for the basis-changing gate (\fig{onepphase}(b)) and the Hadamard gate on the mediator (\fig{halfpihad}) when $C\in\{0,1,2,3\}$. If $C\geq4$ this condition can be achieved by simply concatenating the graph implementing the desired unitary with two paths of length $C+1$ on the input as well as the output. This has the same effect as simply adding extra vertices on the input and output paths of the block, but in \fig{squint} we include these extra vertices inside the circled subgraph.

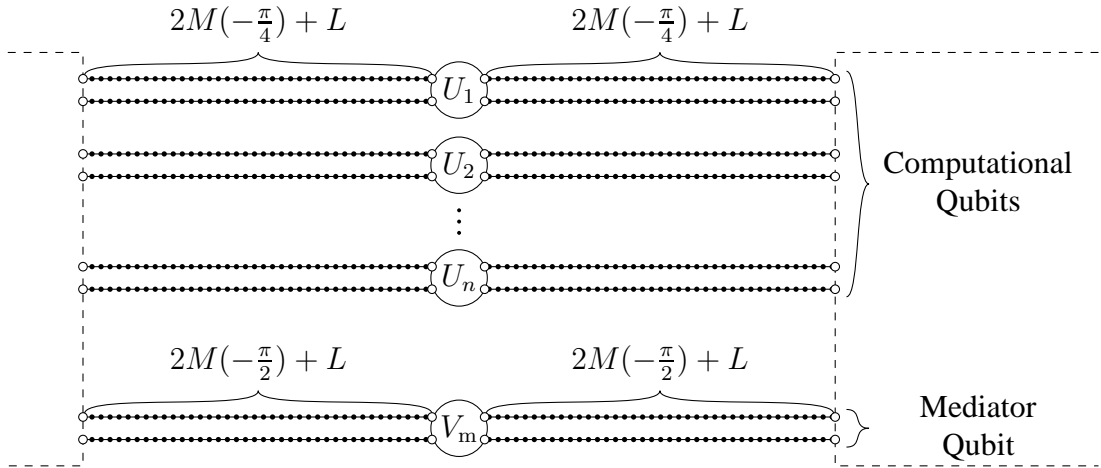
\begin{figure}
\centering
\capstart
\begin{tikzpicture}[
  unitary/.style={circle,draw=black,fill=white,
    inner sep=0pt,minimum size=7.5mm},
  dots/.style={circle,fill=black,
    inner sep=0pt,minimum size=1.5pt},
  vert/.style={circle,fill=black,
    inner sep=.7pt,minimum size=0pt},
  attach/.style={circle,draw=black,fill=white,inner sep=1.15pt,minimum size=0}]

  \foreach \n /\y in {U_1/2.5,U_2/1.5,U_n/0,{V_{\med}}/-2}{
    \begin{scope}[yshift = \y cm]
      \foreach \x in {0,.12,...,10} {
        \node at (\x, -.15) [vert] {};
        \node at (\x, .15) [vert] {};
      }
      \draw (0, -.15) to (10, -.15);
      \draw (0,  .15) to (10,  .15);
      \node at (5,0) [unitary]{$ \n $};

    \end{scope}
  }
  
  \foreach \x in {0,5.34}{
  \foreach \y /\p in {2.5/4,-2/2}{
  \begin{scope}[xshift=\x cm,yshift=\y cm]
    \node (\x\p) at (2.33,.5)[above] {$2M(-\frac{\pi}{\p}) + L$};
  
    \draw (0.02,0.2) to[out=80,in=-90,looseness=0.3] (2.33,.5)
                     to[out=-90,in=100,looseness=0.3] (4.64,.2);
  \end{scope}}} 
                  
  \draw (10.15,2.75) to[out=0,in=-180,looseness=0.3] (10.45,1.25)
                    to[out=-180,in=0,looseness=0.3] (10.15,-.25);
  \node at (11.9,1.25)
      {\begin{tabular}{c}
          Computational\\ 
          Qubits\end{tabular}};

  \draw (-1,3) to (0, 3) to (0,-2.5) to (-1,-2.5) [dashed];
  \draw (13.5,3) to (10,3) to (10,-2.5) to (13.5,-2.5) [dashed];

  \node at (5,.9) [dots] {};
  \node at (5,.75)    [dots] {};
  \node at (5,.6) [dots] {};

  \draw (10.15,-1.75) to[out=0,in=-180,looseness=1.5] (10.45,-2)
                    to[out=-180,in=0,looseness=1.5] (10.15,-2.25);
  \node at (11.9,-2) 
      {\begin{tabular}{c}
          Mediator\\ 
          Qubit\end{tabular}};
          
   \foreach \x in {0, 10, 4.64, 5.34}{
   \foreach \y in {2.35,2.65,1.65,1.35,.15,-.15,-1.85,-2.15}{
     \node at (\x, \y) [attach] {};
   }}
\end{tikzpicture}
\caption{A block of type I. Here the single-qubit gates $U_1,\ldots,U_n$ on the encoded computational qubits are either the phase gate, the basis-changing gate, or the identity (up to a global phase). The single-qubit gate $V_\med$ on the mediator is either the identity or the Hadamard gate (again, up to a global phase).}
\label{fig:squint}
\end{figure}

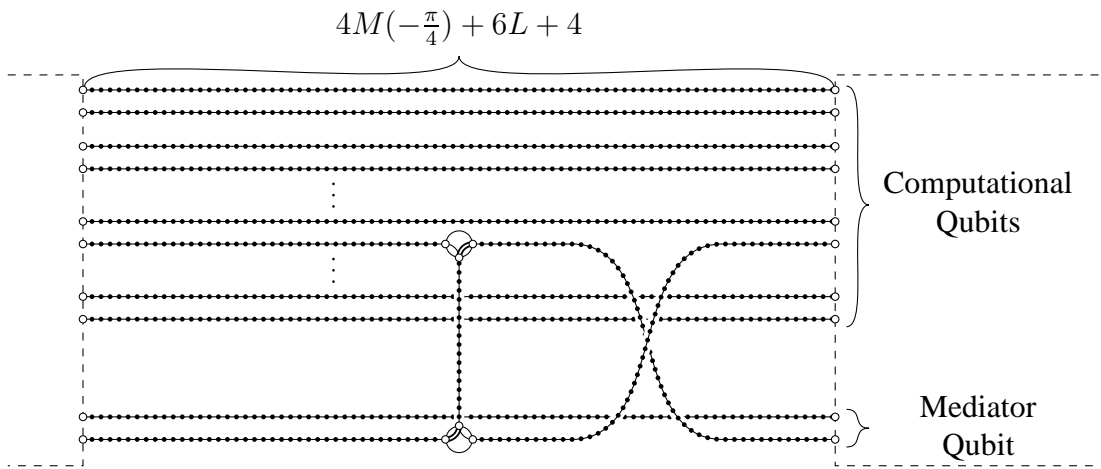
\begin{figure}
\centering
\capstart
\begin{tikzpicture}[scale=1,
  dots/.style={circle,fill=black,inner sep=0pt,
    minimum size=1pt},
  splitter/.style={circle,draw=black,fill=white,
    inner sep=0pt,minimum size=3.5mm},
  vert/.style={circle,fill=black,
    inner sep=.65pt,minimum size=0pt},
  attach/.style={circle,draw=black,fill=white,inner sep=1pt,minimum size=0},
  decoration={markings,
    mark= between positions 0 and 1 step .12 cm with{
      \node [vert] at (0,0){};
    }}]

  \foreach \y in {2.75,2,0} {
  \begin{scope}[yshift=\y cm]
    \draw[postaction={decorate}] (0,.15) -- (10, .15);
    \draw[postaction={decorate}] (0,-.15) --  (10,-.15);
  \end{scope}}

  \draw[postaction={decorate}] (0,1.15) -- (10,1.15);
  \draw[postaction={decorate}] (0,-1.45) -- (10,-1.45);

  \node (splitter1) at (5,.85) [splitter] {};
  \node (splitter2) at (5,-1.75)  [splitter] {};

  \draw (splitter1.west) to[out=0,in=90] (splitter1.south);
  \draw[line width=2.1pt] (splitter1.south) to[out=90,in=180] 
         (splitter1.east);
  \draw[line width=.7pt,white] (splitter1.south) to[out=90,in=180] 
         (splitter1.east);

  \draw[line width=2.1pt] (splitter2.west) to[out=0,in=-90] 
         (splitter2.north);
  \draw[line width=.7pt,white] (splitter2.west) to[out=0,in=-90] 
         (splitter2.north);
  \draw (splitter2.north) to[out=-90,in=180] (splitter2.east);

  \draw[line width=4pt, white] (splitter2.north) to (splitter1.south);
  \draw[postaction={decorate}] (splitter2.north) to (splitter1.south);

  \draw[postaction={decorate}] (0,.85) -- (splitter1.west);
  \draw[postaction={decorate}] (0,-1.75) -- (splitter2.west);
  
  \draw[line width = 4pt, white] (splitter1.east) -- (6.5,.85) 
    to[out=0,in=180,looseness=1] (8.5,-1.75) -- (10,-1.75);
  \draw[postaction={decorate}] (splitter1.east) -- (6.5,.85) 
    to[out=0,in=180,looseness=1] (8.5,-1.75) -- (10,-1.75);
    
  \draw[line width=4pt,white] (splitter2.east) -- (6.5,-1.75) 
    to[out=0,in=180,looseness=1] (8.5,.85) -- (10,.85);
  \draw[postaction={decorate}] (splitter2.east) -- (6.5,-1.75) 
    to[out=0,in=180,looseness=1] (8.5,.85) -- (10,.85);

  \foreach \x in {1.65,1.5,1.35,.65,.5,.35}
    \node at (3.33,\x) [dots] {};

  \draw (-1,3.1) to (0, 3.1) to (0,-2.1) to (-1,-2.1) [dashed];
  \draw (13.5,3.1) to (10,3.1) to (10,-2.1) to (13.5,-2.1) [dashed];

  \draw (10.15,2.95) to[out=0,in=-180,looseness=0.5] (10.45,1.37)
                    to[out=-180,in=0,looseness=0.5] (10.15,-.25);
  \node at (11.9,1.37) 
      {\begin{tabular}{c}
          Computational\\ 
          Qubits\end{tabular}};

  \draw (10.15,-1.35) to[out=0,in=-180,looseness=1.5] (10.45,-1.6)
                    to[out=-180,in=0,looseness=1.5] (10.15,-1.85);
  \node at (11.9,-1.6)
      {\begin{tabular}{c}
          Mediator\\ 
          Qubit\end{tabular}};

  \node at (5,3.35)[above] {$4M(-\frac{\pi}{4})+6L+4$};
  
  \draw (0.02,2.95) to[out=80,in=-90,looseness=0.3] (5,3.35)
                   to[out=-90,in=100,looseness=0.3] (9.98,2.95);

  \foreach \x in {0,10}{
  \foreach \y in {-.15,.15, .85,1.15,1.85,2.15,2.6,2.9,-1.45,-1.75}{
    \node at (\x,\y) [attach]{};
  }}

  \foreach \n in {splitter1, splitter2}{
    \node at (\n.east) [attach]{};
    \node at (\n.west) [attach]{};
  }  
  
  \node at (splitter2.north) [attach]{};
  \node at (splitter1.south) [attach]{};
  
\end{tikzpicture}
\caption{A block of type II.  Here a $\CD$ gate is implemented between a computational qubit and the mediator qubit.}
\label{fig:CPintall}
\end{figure}

For a block of type I, the number of vertices on each input path for the mediator qubit is equal to the number of vertices on each of its output paths and is chosen to be
\[
2M(-{\pi}/{2})+L
\] 
where $M(-{\pi}/{2})=L$. Similarly, the number of vertices on the input and output paths is the same for each computational qubit and is chosen to be
\[
2M(-{\pi}/{4})+L
\]
where 
\[
M(-{\pi}/{4})=\left\lceil\left(\frac{3\sqrt{2} - 2}{4}\right)L\right\rceil.
\]
These lengths are chosen to be different to compensate for the fact that the computational and mediator particles move at different speeds. They are designed so that a wave packet of length $L$ and momentum $k$ on an input path initially a distance $M(k)$ from the subgraph implementing the unitary will be found a distance $M(k)$ from the graph on the output paths after time $t_{\mathrm{I}} = {3L}/{2}$. This is illustrated in \fig{1cartoon}.

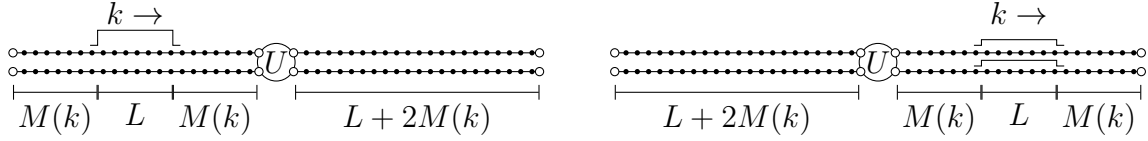
\begin{figure}
\centering
\capstart
\begin{tikzpicture}[
  scale = 0.5,
  unitary/.style={circle,draw=black,fill=white,
    inner sep=0pt,minimum size=5mm},
  verts/.style={circle,fill=black,inner sep=.7pt,minimum size=0},
  attach/.style={circle,fill=white,draw=black,inner sep=1.1pt,minimum size=0pt}]
  \draw (-2,0.25) to (12,0.25);
  \draw (-2,-0.25) to (12,-0.25);

  \foreach \x in {-2,-1.7,...,12}{
  \foreach \y in {-.25,.25}{
    \node at (\x,\y) [verts] {};
  }}

  \node (1) at (5,0) [unitary] {$U$};

  \draw[xshift=-1cm] (1.05,.45) to (1.25,.45) to (1.25,.85) to (3.25,.85) 
           to (3.25,.45) to (3.45,.45);

  \node (momentum) at (1.25,1.3) [] {$k\rightarrow$};

  \draw[|-|] (-2,-.8)   to node[below] {$M(k)$} (0.25,-.8);
  \draw[|-|] (0.25,-.8) to node[below] {$L$}   (2.25,-.8);
  \draw[|-|] (2.25,-.8) to node[below] {$M(k)$} (4.5,-.8); 

  \draw[|-|] (5.5,-.8) to node[below] {$L + 2M(k)$} (12,-.8);

  \foreach \x in {-2, 12, 4.54, 5.46}{
  \foreach \y in {-.25, .25}{  
    \node at (\x,\y) [attach]{};
  }}

\begin{scope}[xshift=4cm]
  \draw (10,0.25) to (24,0.25);
  \draw (10,-0.25) to (24,-0.25);
  
  \foreach \x in {10,10.3,...,24}{
  \foreach \y in {-.25,.25}{
    \node at (\x,\y) [verts] {};
  }}
  
  \node (2) at (17,0) [unitary] {$U$};
 
  \draw[xshift=1cm] (18.55,.45) to (18.75,.45) to (18.75,.6)
     to (20.75,.6) to (20.75,.45) to (20.95,.45);

  \draw[xshift=1cm] (18.55,-.1) to (18.75,-.1) to (18.75,.05)
     to (20.75,.05) to (20.75,-.1) to (20.95,-.1);

  \node (momentum2) at (20.75,1.3) [] {$k\rightarrow$};

  \draw[|-|] (10,-.8) to node[below] {$L+2M(k)$} (16.5,-.8);
  \draw[|-|] (17.5,-.8) to node[below] {$M(k)$} (19.75,-.8);
  \draw[|-|] (19.75,-.8) to node[below] {$L$} (21.75,-.8);
  \draw[|-|] (21.75,-.8) to node[below] {$M(k)$} (24,-.8);
  
  \foreach \x in {10, 24, 16.54, 17.46}{
  \foreach \y in {-.25, .25}{  
    \node at (\x,\y) [attach]{};
  }}
\end{scope}
\end{tikzpicture}
\caption{A single-qubit gate $U$ acts on an encoded qubit. The wave packet starts on the paths on the left-hand side of the figure, a distance $M(k)$ from the ends of the paths. After time $t_{\mathrm{I}}={3L}/{2}$ the logical gate has been applied and the wave packet has traveled a distance $2M(k)+L$ (up to error terms that are bounded as $\O(L^{-{1}/{4}}))$.}
\label{fig:1cartoon}
\end{figure}

Blocks of type II are used to perform $\CD$ gates between a computational qubit and the mediator.  The two-qubit gate involved in such a block, including the lengths of all the relevant paths, is shown in detail in \fig{Graph-used-to-1}.  To implement an identity operation on the other computational qubits, we simply attach the input to the output by a path with $4M(-{\pi}/{4})+6L+4$ vertices.  A wave packet of momentum $-{\pi}/{4}$ and length $L$ starting on an input path a distance $M(-{\pi}/{4})$ from the leftmost vertex will be found a distance approximately $M(-{\pi}/{4})$ from the rightmost vertex on an output path after a time $t_{\mathrm{II}}=(5L+2M(-{\pi}/{4}))/\sqrt{2}$.

To compose blocks $B$ and $B'$, where $B'$ is a block of type I, replace each of the vertices of the input paths for block $B'$ with the rightmost $2M(k)+L$ vertices of the corresponding output paths of block $B$. If $B'$ is a block of type II, replace only the leftmost $2M(k)+L$ vertices of the input paths of $B'$ with the rightmost $2M(k)+L$ vertices of the output paths of $B$. A simple example is shown in \fig{CombiningBlocks}.

This method of composing blocks simplifies our analysis, which proceeds by computing the evolution of the particles inside each block and then using a ``truncation lemma'' to bound the errors arising from the connection to other blocks on either side. In the example shown in \fig{CombiningBlocks}, the initial state of each particle is a wave packet of length $L$ prepared on the input paths of block $B$ a distance $M(k)$ from the single-qubit gate subgraphs. At time $t_{\mathrm{I}}$ the single-qubit gates have been applied and the wave packets encoding the logical state have propagated a distance approximately $M(k)$ on the output paths of the first block. This output state for the first block coincides (approximately, and up to an irrelevant global phase) with an input logical state for the second block, as it is a distance approximately $M(k)$ from the second set of single-qubit gates. Finally, after time $2t_{\mathrm{I}}$ the logical state is on the output paths of the second block.

More generally, for a circuit with $g_{\mathrm{I}}$  blocks of type I and $g_{\mathrm{II}}$ blocks of type II, the total evolution time for the computation is 
\begin{equation}
T=g_{\mathrm{I}}t_{\mathrm{I}}+g_{\mathrm{II}} t_{\mathrm{II}} = \Theta(gL),
\label{eq:total_time}
\end{equation}
where $g=g_{\mathrm{I}}+g_{\mathrm{II}}$.

\begin{figure}
\centering
\capstart
\begin{tikzpicture}[scale=.5,
  attach/.style={circle,fill=white,draw=black,
    inner sep=1pt,minimum size=0pt},
  vert/.style={circle,draw=black,fill=black,
    inner sep=1pt,minimum size=0pt},
  dots/.style={circle,fill=black,inner sep=.4pt,
    minimum size=0pt}
  ]
  \foreach \xsh in {0,9,18}{
  \foreach \ysh in {0,2,5.5,7.5}{
  \begin{scope}[xshift=\xsh cm,yshift=\ysh cm]
    \draw (0,0) -- (2,0);
    \draw (4,0) -- (6,0);
    \draw[densely dotted] (2,0) -- (2.33,0);
    \draw[densely dotted] (3.66,0) -- (4,0);

    \node at (0,0) [attach] {};
    \node at (1,0) [vert]{};
    \node at (2,0) [vert]{};
    \node at (2.66,0) [dots]{};
    \node at (3,0) [dots]{};
    \node at (3.33,0) [dots]{};
    \node at (4,0) [vert]{};
    \node at (5,0) [vert]{};
    \node at (6,0) [attach]{};
  \end{scope}}}

  \begin{scope}[xshift=6cm]
    \node (0inhad) at (0,2) [attach]{};
    \node (1inhad) at (0,0) [attach]{};
    \node (0outhad) at (3,2) [attach]{};
    \node (1outhad) at (3,0) [attach]{};
    
    \node (1had) at (1,0) [vert]{};
    \node (2had) at (2,0) [vert]{};
    \node (3had) at (1.5,.5) [vert]{};
    \node (4had) at (1,1) [vert]{};
    \node (5had) at (1.5,1) [vert] {};
    \node (6had) at (2,1) [vert] {};
    \node (7had) at (1.5,1.5) [vert]{};
    \node (8had) at (1,2) [vert]{};
    \node (9had) at (2,2) [vert]{};

    \draw (8had) -- (1had);
    \draw (9had) -- (2had);
    \draw (4had) -- (9had);
    \draw (1had) -- (6had);
    \draw (4had) -- (6had);
    \draw (3had) -- (4had);
    \draw (7had) -- (6had);
    \draw (1inhad) -- (2had);
    \draw (8had) -- (9had);
    \draw (8had) -- (5had)[ultra thick,white];
    \draw (5had) -- (2had)[ultra thick,white];
    \draw (8had) -- (2had);
    
    \draw (9had) -- (1outhad);
    \draw (0inhad) -- (7had)[ultra thick,white];
    \draw (0inhad) -- (7had);

    \draw (3had) to[out=0,in=240] (0outhad)[ultra thick,white];
    \draw (3had) to[out=0,in=240] (0outhad);

    \node at (0inhad) [attach]{};
    \node at (1outhad) [attach]{};
    \node at (0outhad) [attach]{};

  \end{scope}

  \draw (15,0) -- (18,0);
  \draw (15,2) -- (18,2);
  
  \node at (15,0) [attach] {};
  \node at (18,0) [attach] {};
  \node at (15,2) [attach] {};
  \node at (18,2) [attach] {};

  \begin{scope}[xshift=6cm,yshift = 5.5cm]
    \draw (0,0) -- (3,0);
    \draw (0,2) -- (3,2);
    \draw[xshift=.5cm] (1,0) -- (1,.59) -- (.58,1) --  (1,1.42) 
       -- (1.42,1) -- (1,.59);

    \node at (0,0) [attach]{};
    \node at (1.5,0) [vert] {};
    \node at (3,0) [attach]{};
    \node at (1.5,.59) [vert]{};
    \node at (1.08,1) [vert]{};
    \node at (1.92,1) [vert]{};
    \node at (1.5,1.42) [vert]{};
    \node at (1,2) [vert]{};
    \node at (2,2) [vert]{};
    \node at (0,2) [attach]{};
    \node at (3,2) [attach]{};
  \end{scope}

  \begin{scope}[xshift=15cm,yshift=5.5cm]
    \draw (0,0) -- (3,0);
    \draw (0,2) -- (3,2);
    \draw (1,0) -- (1,2);
    \draw (2,0) -- (2,2);

    \foreach \x in {1,2}{
    \foreach \y in {0,1,2}{
      \node at (\x,\y) [vert]{};
    }}

    \foreach \x in {0,3}{
    \foreach \y in {0,2}{
      \node at (\x,\y) [attach]{};
    }}

  \end{scope}

  \foreach \x in {0,9,18}{
  \draw[xshift=\x cm,yshift=2.4cm,|-|] (-.25,0) 
      to node[above]{$L+2M(-\pi/2)$} (6.25,0);
  
  \draw[xshift=\x cm,yshift=5.1cm,|-|] (-.25,0)
      to node[below]{$L+2M(-\pi/4)$} (6.25,0);
  }    

  \draw [dashed] (-.5,-.5) rectangle (15.5,8.5);

  \draw [dashed] (8.5,-1) rectangle (24.5,8);

  \node at (4,9.1) {Block $B$};

  \node at (20.5,8.6) {Block $B'$};

\end{tikzpicture}
\caption{An example of combining blocks to build a circuit. Here the two blocks are both of type I.}
\label{fig:CombiningBlocks}
\end{figure}
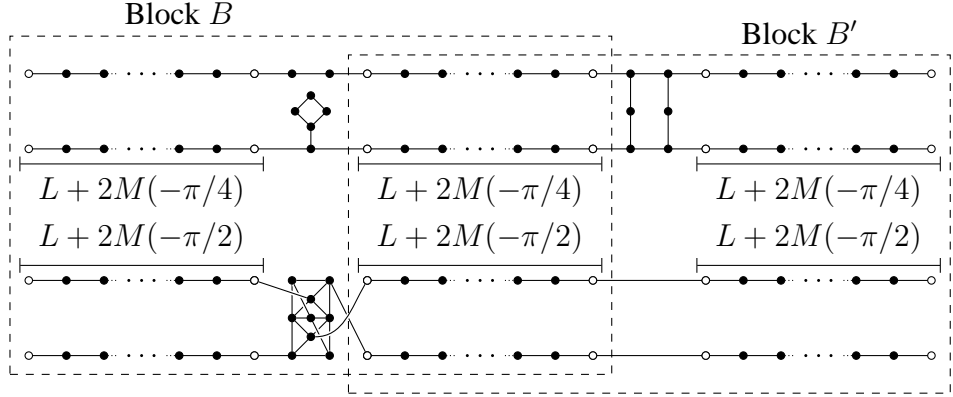

\subsection{Initial state, final measurement, and error bound}

We prepare an initial state of $n+1$ spatially separated wave packets, where each wave packet is localized on an input path of the first block of the graph.  To simplify notation, we assume that the first block is of type I, taking the unitaries $U_i$ and $V_\med$ to be the identity if required.  Label the vertices of the $2n$ input paths in the first block corresponding to computational qubits as $\ket{x,q}^{j}_{\mathrm{in}}$, where $j\in\{1,\ldots,n\}$ indexes the qubit, $q\in\{0,1\}$ labels the computational basis state, and $x\in\{1,\ldots,2M(-{\pi}/{4})+L\}$ denotes the position along the path.  Similarly, label the vertices of the input paths for the mediator qubit as $\ket{x,q}^{\med}_{\mathrm{in}}$ where $x \in \{1,\ldots,2M(-\pi/2)+L\}$.  Here the vertices with $x=1$ are the rightmost vertices on the input paths. We define the logical input states for each of the qubits as
\begin{align*}
|0_{\text{in}}\rangle^j & =  \frac{1}{\sqrt{L}}\sum_{x=M\left(-\frac{\pi}{4}\right)+1}^{M\left(-\frac{\pi}{4}\right)+L}e^{i\frac{\pi}{4}x}|x,0\rangle^{j}_{\mathrm{in}} &
|1_{\text{in}}\rangle^j & =  \frac{1}{\sqrt{L}}\sum_{x=M\left(-\frac{\pi}{4}\right)+1}^{M\left(-\frac{\pi}{4}\right)+L}e^{i\frac{\pi}{4}x}|x,1\rangle^{j}_{\mathrm{in}}
\end{align*}
and the logical input states for the mediator as
\begin{align*}
|0_{\text{in}}\rangle^{\med} & =  \frac{1}{\sqrt{L}}\sum_{x=M\left(-\frac{\pi}{2}\right)+1}^{M\left(-\frac{\pi}{2}\right)+L}e^{i\frac{\pi}{2}x}|x,0\rangle^{\med}_{\mathrm{in}} &
|1_{\text{in}}\rangle^{\med} & =  \frac{1}{\sqrt{L}}\sum_{x=M\left(-\frac{\pi}{2}\right)+1}^{M\left(-\frac{\pi}{2}\right)+L}e^{i\frac{\pi}{2}x}|x,1\rangle^{\med}_{\mathrm{in}}.
\end{align*}

The system is initialized at time $t=0$ in the computational basis state $\ket{00\ldots 0}$, encoded as
\[
|\psi \left(0\right)\rangle =\Sym(|0_{\text{in}}\rangle^1\ldots|0_{\text{in}}\rangle^n|0_{\text{in}}\rangle^{\med}).
\]
Here $\Sym$ is a linear operator that symmetrizes if the particles are bosons and antisymmetrizes if the particles are fermions.  It is defined by
\[
  \Sym\left( |a_1\rangle|a_2\rangle\ldots|a_{n+1}\rangle \right) =\frac{1}{\sqrt{(n+1)!}}\sum_{\pi\in S_{n+1}} (\pm1)^{\text{sgn}(\pi) }|a_{\pi(1)}\rangle |a_{\pi(2)}\rangle\ldots|a_{\pi(n+1)}\rangle
\]
where the $\pm$ is $+$ for bosons and $-$ for fermions.

We evolve the initial state for a time $T$ according to the Schr\"{o}dinger equation with the Hamiltonian $H_G^{(n+1)}$ on a graph $G$ built by composing $g_{\mathrm{I}}$ blocks of type I and $g_{\mathrm{II}}$ blocks of type II as described in the previous section.  The total evolution time $T$ is given by equation \eq{total_time}. Labeling the vertices on the output paths of the final block (again assumed to be type I) as $|x,q\rangle^{j}_{\text{out}}$,  where now vertices with $x=1$ are the leftmost vertices on the output paths, the logical output states for the computational qubits are defined as
\begin{align*}
|0_{\text{out}}\rangle^j & =  \frac{1}{\sqrt{L}}\sum_{x=M\left(-\frac{\pi}{4}\right)+1}^{M\left(-\frac{\pi}{4}\right)+L}e^{-i\frac{\pi}{4}x}|x,0\rangle^{j}_{\text{out}} &
|1_{\text{out}}\rangle^j & =  \frac{1}{\sqrt{L}}\sum_{x=M\left(-\frac{\pi}{4}\right)+1}^{M\left(-\frac{\pi}{4}\right)+L}e^{-i\frac{\pi}{4}x}|x,1\rangle^{j}_{\text{out}}
\end{align*}
with similar definitions for the mediator qubit. Letting $U_C$ be the logical $n$-qubit unitary that the graph is intended to implement, the desired output state is 
\[
|\phi\rangle =\Sym\left(\sum_{\vec{z}\in \{0,1\}^n} \langle z^{1}z^{2}\ldots z^{n}| U_C|00\ldots 0\rangle |z^{1}_{\text{out}}\rangle^1\ldots|z^{n}_{\text{out}}\rangle^n|0_{\text{out}}\rangle^{\med}\right).
\]
where $z^i$ is the $i$th bit of $\vec{z}$.  A final measurement of the encoded quantum state in the computational basis is performed by measuring the locations of the particles at the end of the time evolution.

We prove in the next section that 
\begin{equation}
\left\Vert e^{-iH_G^{(n+1)} T} |\psi\left(0\right)\rangle - e^{i\gamma} |\phi\rangle\right\Vert = \O\left(gn\Norm{H_G^{(n+1)}} L^{-{1}/{4}}\right) 
\label{eq:errorbound}
\end{equation}
where $\gamma$ is an irrelevant overall phase.   For the Bose-Hubbard model and for the models with nearest-neighbor interactions that we consider, $\norm{H_G^{(n+1)}} = \O(n^2)$, so by taking $L=\O(n^{12}g^{4})$ we can make the error arbitrarily small. Using this bound on $L$, the total number of vertices required in our construction is $\O(n^{13}g^5)$ and the total evolution time is $\O(n^{12}g^5)$. Although these bounds are sufficient to establish universality with only polynomial overhead, we expect that they can be improved significantly.

\section{Analysis of wave packet scattering}\label{sec:analysis}

In this section we prove the error bound \eq{errorbound} stated above. The proof relies on the following technical results.

In \sec{Single-Particle-Wavetrain} we prove the following theorem about single-particle wave packet scattering on an \emph{infinite} graph of the form shown in \fig{graph}. The proof is based on a calculation from reference \cite{FGG08}.

\begin{theorem}
\label{thm:singlepart}Let $\widehat{G}$ be an $(N+m)$-vertex graph. Let $G$ be a graph obtained from $\widehat{G}$ by attaching semi-infinite paths to $N$ of its vertices, as shown in \fig{graph}, and let $S$ be the corresponding S-matrix. Let $H_{G}^{(1)}$ be the quantum walk Hamiltonian of equation \eq{single_particle_ham}. Let $k\in(-\pi,0)$, $M,L\in\natural$, $j\in\{1,\ldots,N\}$, and
\[
|\psi^{j}(0)\rangle=\frac{1}{\sqrt{L}}\sum_{x=M+1}^{M+L}e^{-ikx}|x,j\rangle.
\]
Let $c_{0}$ be a constant independent of $L$. Then, for all $0\leq t\leq c_{0}L$,
\[
\left\Vert e^{-iH_{G}^{(1)}t}|\psi^{j}(0)\rangle-|\alpha^{j}(t)\rangle\right\Vert =\O(L^{-{1}/{4}})\]
where 
\begin{align*}
|\alpha^{j}(t)\rangle & =  \frac{1}{\sqrt{L}}e^{-2it\cos k}\sum_{x=1}^{\infty}\sum_{q=1}^{N}\left(\delta_{qj} e^{-ikx}R(x-\left\lfloor 2t\sin k\right\rfloor)+S_{qj}(k)e^{ikx}R(-x-\left\lfloor 2t\sin k\right\rfloor)\right)|x,q\rangle
\end{align*}
with
\begin{align*}
R(l) & =  \begin{cases}
1 & \text{if }l\in\{M+1,M+2,\ldots,M+L\}\\
0 & \text{otherwise.}\end{cases}
\end{align*}
\end{theorem}

The approximation $|\alpha^j(t)\rangle$ to the time-evolved state has two terms, one corresponding to the incoming wave packet and one corresponding to a superposition of outgoing wave packets. Before scattering, the wave packet is supported entirely on path $j$ and the second term is zero. Likewise, after scattering the first term is zero and the particle is outgoing along the semi-infinite paths. As we can see from the above expression, the square wave packets involved move with speed $2|{\sin k}|$ as expected (and to a good approximation, movement occurs at discrete times).

In \sec{Two-Particle-Wavetrain} we prove the following result about two wave packets moving past each other on an infinite path.  While the same proof applies to wave packets with other momenta, for concreteness we assume that the momenta are $k_1={-\pi}/{2}$ and $k_2={\pi}/{4}$ (so $p_1=-k_1-k_2={\pi}/{4}$ and $p_2=(k_2-k_1)/2={3\pi}/{8}$).

\begin{theorem}
\label{thm:twopart}Let $H^{(2)}$ be a two-particle Hamiltonian of the form \eq{twoham} with interaction range at most $C$, i.e., $\mathcal{V}(|r|)=0$ for all $|r|>C$. Let $\theta_{\pm}(p_1,p_2)$ be given by equation \eq{delta_pm}. Define $\theta=\theta_{\pm}({\pi}/{4},{3\pi}/{8})$. Let $L\in\natural$, let $M\in\{C+1,C+2,\ldots\}$, and define
\begin{align*}
|\chi_{z,k}\rangle & =  \frac{1}{\sqrt{L}}\sum_{x=z-L}^{z-1}e^{ikx}|x\rangle\\
|\psi(0)\rangle & =  \frac{1}{\sqrt{2}}\left(|\chi_{-M,-\frac{\pi}{2}}\rangle|\chi_{M+L+1,\frac{\pi}{4}}\rangle 
	\pm |\chi_{M+L+1,\frac{\pi}{4}}\rangle|\chi_{-M,-\frac{\pi}{2}}\rangle\right).
\end{align*}
Let $c_{0}$ be a constant independent of $L$. Then, for all $0\leq t\leq c_{0}L$, we have
\[
\left\Vert e^{-iH^{(2)}t}|\psi(0)\rangle-|\alpha(t)\rangle\right\Vert =\O(L^{-{1}/{4}}),
\]
where 
\begin{equation}
|\alpha(t)\rangle=\sum_{x,y}a_{xy}(t)|x,y\rangle,
\label{eq:alpha}
\end{equation}
$a_{xy}(t)=\pm a_{yx}(t)$ for $x \ne y$, and, for $x\leq y$, 
\begin{align}
a_{xy}(t) & =  \frac{1}{\sqrt{2}L}e^{-\sqrt{2}it}\left[e^{-i \pi x/2} e^{i \pi y/4} F(x,y,t) 
    \pm e^{i\theta} e^{i \pi x/4} e^{-i \pi y/2} F(y,x,t) \right]
\label{eq:a_xy}
\end{align}
where
\begin{align*}
F(u,v,t) & =  \begin{cases}
	1 & \text{if }u-2 \lfloor t \rfloor\in\{-M-L,\ldots,-M-1\}\text{ and }v+2\left\lfloor \frac{t}{\sqrt{2}}
	\right\rfloor \in\{M+1,\ldots,M+L\}\\
	0 & \text{otherwise.}\end{cases}
\end{align*}
\end{theorem}
\vspace{1cm}

Note that for $t>t_{0}$, where $t_{0}$ is the smallest time that satisfies 
\[
2\left(\left\lfloor \frac{t_0}{\sqrt{2}}\right\rfloor + \lfloor t_0 \rfloor \right)\geq 2M+2L+1,
\]
we have $F(x,y,t) =0$ for $x\leq y$.
Thus, for $t>t_{0}$, 
\begin{equation}
|\alpha(t)\rangle=\frac{e^{i\theta}}{\sqrt{2}}e^{-\sqrt{2}it}\left(|\chi_{-M+2\left\lfloor t
  	\right\rfloor ,-\frac{\pi}{2}}\rangle|\chi_{M+L+1-2\left\lfloor {t}/{\sqrt{2}}\right\rfloor ,
  	\frac{\pi}{4}}\rangle \pm|\chi_{M+L+1-2\left\lfloor {t}/{\sqrt{2}}\right\rfloor ,\frac{\pi}{4}}\rangle|
  	\chi_{-M+2\left\lfloor t\right\rfloor ,-\frac{\pi}{2}}\rangle\right)
\label{eq:t_biggerthan_t0-1}
\end{equation}
which describes two wave packets of length $L$ moving apart. In particular, note that the prefactor includes an overall phase of $e^{i\theta}$ arising from the interaction between the particles.

Note that in Theorems \ref{thm:singlepart} and \ref{thm:twopart}, the big-$\O$ notation includes the dependence on $M$. In particular, the error bounds hold even if $M$ is chosen to depend on $L$. When applying these Theorems, we choose $M$ to be proportional to $L$.

These results about wave packet scattering on infinite graphs are not sufficient to prove error bounds for our scheme, as the graphs we construct are finite. However, a wave packet propagates with a finite speed, so that if we evolve for a finite amount of time then we expect a wave packet to only ``see'' a small part of the graph. The next ingredient in our analysis is a ``truncation lemma'' that we use to make this precise. The lemma, which we prove in \sec{Truncation-Lemma}, is as follows:

\begin{lemma}[Truncation Lemma]
\label{lem:trunc}
Let $H$ be a Hamiltonian acting on a Hilbert
space $\H$ and let $\ket{\Phi}\in\H$ be a normalized state. Let
$\K$ be a subspace of $\H$, let $P$ be the projector onto $\K$,
and let $\tilde{H}=PHP$ be the Hamiltonian within this subspace. Suppose
that, for some $T>0$, $W\in\{H,\tilde{H}\}$, $N_0\in\natural$,
and $\delta>0$, we have, for all $0\leq t\leq T$, 
\begin{align*}
e^{-iWt}|\Phi\rangle & = |\gamma(t)\rangle+|\epsilon(t)\rangle \text{ with }
\left\Vert |\epsilon(t)\rangle\right\Vert \leq \delta
\end{align*}
and
\begin{align*}
  (1-P) H^{r}|\gamma(t)\rangle & = 0 \text{ for all } r\in\{0,1,\ldots, N_0-1\}.
\end{align*}
Then, for all $0\leq t \leq T$, 
\[
  \Norm{\left(e^{-iHt}-e^{-i\tilde{H}t}\right)|\Phi\rangle}
  \leq \left(\frac{4e\norm{H}t}{N_0} + 2 \right) 
        \left(\delta + 2^{-N_0}(1+\delta)\right).
\]
\end{lemma}
\subsection{Truncating the semi-infinite paths}\label{sec:truncating}

\lem{trunc} lets us prove analogs of \thm{singlepart} and \thm{twopart} when the graphs involved have been truncated to have finitely many vertices.

For example, consider the case where $H=H_G^{(1)}$ is the Hamiltonian \eq{single_particle_ham} for a single particle on a graph of the form shown in \fig{graph}. Let $G(K)$ be the finite graph obtained from $G$ by truncating each of the paths to have a total length $K=\Omega(L)$ (so that the endpoints of the paths are labeled $(K,j)$ for $j\in\{1,\ldots,N\}$), and choose $\tilde{H}=H_{G(K)}^{(1)}$. Let the subspace $\K$ be spanned by basis states corresponding to vertices in $G(K)$.  Let $|\Phi\rangle=|\psi^j(0)\rangle$ be the same initial state as in \thm{singlepart}. We choose the evolution time $T$ so that for $0\leq t\leq T$, the time-evolved state remains far from the vertices labeled $(K,j)$ (for each $j\in\{1,\ldots,N\}$), and thus far from the effect of truncating the paths.  More precisely, we choose $T=\O(L)$ so that, for times $0\leq t\leq T$, the state $|\alpha^j (t)\rangle$ from \thm{singlepart} has no amplitude on vertices within a distance $N_0=\Omega(L)$ from the endpoints of the paths. For such times $t$ we have
\[
\left(1-P\right)H^r |\alpha^j \left(t\right)\rangle=0 \text{ for all } 0\leq  r < N_0.
\]
This allows us to use the lemma with $W=H=H_G^{(1)}$, $|\gamma(t)\rangle=|\alpha^j(t)\rangle$, and the bound $\delta=\O(L^{-{1}/{4}})$ from \thm{singlepart}. The truncation lemma then says that, for $0\leq t\leq T$,
\begin{equation}
\left\Vert \left(e^{-iH_{G}^{(1)}t}-e^{-iH_{G(K)}^{(1)}t}\right)|\psi^{j}(0)\rangle \right\Vert=\O(L^{-{1}/{4}})
\label{eq:trunc_paths}
\end{equation}
so, for $0\leq t\leq T$,
\[
\left\Vert e^{-iH_{G(K)}^{(1)}t}|\psi^{j}(0)\rangle-|\alpha^{j}(t)\rangle \right\Vert =\O(L^{-{1}/{4}}).
\]
In other words, for small enough evolution times, the conclusion of \thm{singlepart} still holds if we replace the full Hamiltonian $H^{(1)}_G$ with the truncated Hamiltonian $H_{G(K)}^{(1)}$.

We can also use the truncation lemma to extend \thm{twopart} to the case where the infinite path has been truncated to a finite path.  Let $W=H=H^{(2)}$  be the two-particle Hamiltonian \eq{twoham} and define the truncated Hamiltonian $\tilde{H}^{(2)}= P H^{(2)} P$, where
\begin{align*}
 P = \sum_{x,y = -K_1}^{K_2} \ket{x,y} \bra{x,y}
\end{align*}
with $K_1, K_2=\Omega(L)$. We take $|\Phi\rangle =|\psi(0)\rangle$ as in \thm{twopart} and choose the evolution time $T=\O(L)$ so that, for all times $0\leq t\leq T$, the state $|\alpha (t)\rangle$ has no amplitude on states where either particle is located within a distance $N_0=\Omega(L)$ from the endpoints of the truncated path. With these choices, and letting $\ket{\gamma(t)} = \ket{\alpha(t)}$ and $\delta = \O(L^{-1/4})$, we get that
\[
\left\Vert e^{-i\tilde{H}^{(2)}t}|\psi(0)\rangle-|\alpha(t)\rangle\right\Vert =\O(L^{-{1}/{4}})
\]
for $0\leq t\leq T$, which is the conclusion of \thm{twopart} but now applied to the truncated Hamiltonian $\tilde{H}^{(2)}$.

\subsection{Wave packet propagation on more complicated graphs}\label{sec:more_complicated_graphs}

The results of the previous section apply to wave packet scattering on finite graphs with long paths attached. The intuition behind these results is that, as long as the wave packets never get close to the ends of the paths, it does not matter whether we evolve using the Hamiltonian for the finite or the infinite graph. We now take this argument a step further. We expect that attaching another finite graph to the ends of the truncated paths will not substantially alter the time evolution as long as the wave packets never get close to the ends of the paths.

We can use the truncation lemma a second time to make this intuition precise. For example, consider the case of a single-particle wave packet scattering on a graph $G(K)$ of the form described in the previous section. Let $|\Phi\rangle=|\psi^j (0)\rangle$ be the initial state of the particle as defined in \thm{singlepart}, and choose the time $T=\O(L)$ so that for all $0\leq t\leq T$, $|\alpha^j(t)\rangle$ has no amplitude on vertices within a distance $N_0=\Omega(L)$ from the endpoints of the paths. We have already shown in the previous section that under these conditions the truncated Hamiltonian $H_{G(K)}^{(1)}$ generates approximately the same time evolution as the infinite Hamiltonian $H_G^{(1)}$ (up to an error term that is  $\O(L^{-{1}/{4}})$).  Let $G^{\ast}$ be a graph obtained from $G(K)$ by attaching a finite graph to the endpoint vertices of $G(K)$ (the vertices labeled $(K,1),\ldots,(K,N)$).

Now apply the truncation lemma using $W=\tilde{H}=H_{G(K)}^{(1)}$ and $H=H_{G^\ast}^{(1)}$, again letting $P$ project onto $G(K)$.  With $\ket{\gamma(t)}=\ket{\alpha^j(t)}$ and the bound $\delta =\O(L^{-{1}/{4}})$ from the previous section, the truncation lemma gives (for $0\leq t \leq T$)
\[
\left\Vert e^{-iH_{G(K)}^{(1)}t}|\psi^{j}(0)\rangle-e^{-iH_{G^\ast}^{(1)}t}|\psi^{j}(0)\rangle \right\Vert =\O\left( \left\Vert H_{G^\ast}^{(1)}\right\Vert L^{-{1}/{4}}\right)
\]
and hence
\[
\left\Vert e^{-iH_{G^\ast}^{(1)}t}|\psi^{j}(0)\rangle-|\alpha^{j}(t)\rangle \right\Vert =\O\left( \left\Vert H_{G^\ast}^{(1)}\right\Vert L^{-{1}/{4}}\right).
\]
We see that the evolution of a wave packet for small times depends only on a portion of the graph (up to error terms bounded as above). This allows us to analyze a single-particle wave packet transmitting through a complicated graph consisting of subgraphs $G_1(K)$ and $G_2(K)$ that overlap on long paths (e.g., the graph associated with the upper encoded qubit in \fig{CombiningBlocks}).  For the purpose of analysis, we divide the total evolution time into intervals and evolve the state according to the Hamiltonian of a specified subgraph during each interval. Later we will see that this approach can even be used to approximate the time evolution of a system of more than one particle propagating through a graph.

In \sec{1qubit_calc} we describe in full detail the finite graph implementing a single-qubit gate on one logical or mediator qubit. Then in \sec{2qubit_calc} we perform a similar calculation for the two-qubit $\CD$ gate. Finally, in \sec{block_by_block} we show how to combine these bounds to analyze a graph corresponding to an entire quantum circuit.

\subsection{Single-qubit gates} \label{sec:1qubit_calc}

Let us now apply the results of \sec{truncating} to scattering on a graph implementing a single-qubit gate in our encoding. To approximate a single-qubit gate $U$ on either a computational or the mediator qubit, we use a graph $G(K)$ of the form shown in \fig{single_qubit_graph}, with 4 paths of length $K$ extending outward from a subgraph $\widehat{G}$. We consider scattering at momentum $k=-{\pi}/{2}$ (in which case $U$ is either the identity or the Hadamard gate, up to a global phase) or $k=-{\pi}/{4}$ (in which case $U$ is either the phase gate, the identity gate, or the basis-changing gate, up to a global phase). We treat both choices for the momentum $k$ in this section. By construction, the subgraph $\widehat{G}$ has a $4\times4$ S-matrix of the form \eq{S_matrix_circuit} at momentum $k$, with lower left submatrix equal to $U$. In this section we prove results about the single-particle evolution generated by the Hamiltonian $H_{G(K)}^{(1)}$. We apply these results in \sec{block_by_block} where we analyze the full $(n+1)$-particle system.

By choosing $K$ to depend on the momentum $k$, we can compensate for differing particle speeds and have the scattering process occur over a fixed amount of time, $t_{\mathrm{I}}={3L}/{2}$. As discussed in \sec{description}, we choose $K(k)=2M(k)+L$ where 
\begin{align*}
M(-{\pi}/{2}) & =  L &
M(-{\pi}/{4}) & =  \left\lceil \left(\frac{3\sqrt{2} - 2}{4}\right)L\right\rceil .
\end{align*}
 
A logical input state $a|0\rangle+b|1\rangle$ is encoded using single-particle
states $|0_{\text{\text{in}}}\rangle$ and $|1_{\text{in}}\rangle$
that only have support on the top left path and bottom left
path, respectively. The encoded states differ depending on whether
$k=-{\pi}/{2}$ or $-{\pi}/{4}$, and are defined as
\begin{align*}
|0_{\text{in}}\rangle & =  \frac{1}{\sqrt{L}}\sum_{x=M(k)+1}^{M(k)+L}e^{-ikx}|x,1\rangle &
|1_{\text{in}}\rangle & =  \frac{1}{\sqrt{L}}\sum_{x=M(k)+1}^{M(k)+L}e^{-ikx}|x,2\rangle.
\end{align*}
Starting at $t=0$ from the superposition 
\[
|\psi(0)\rangle=a|0_{\text{\text{in}}}\rangle+b|1_{\text{\text{in}}}\rangle,
\]
the computation proceeds by evolving $\ket{\psi(0)}$ with the time-independent Hamiltonian $H_{G(K)}^{(1)}$, corresponding to a quantum walk on $G(K)$. Here $K=K(k)$ but we leave the argument implicit. After time $t_{\mathrm{I}}$, the state is
\[
|\psi_{K}(t_{\mathrm{I}})\rangle=e^{-iH_{G(K)}^{(1)} t_{\mathrm{I}}}|\psi(0)\rangle.
\]
The state $|\psi_{K}(t_{\mathrm{I}})\rangle$ is approximated
by 
\begin{equation}
|\psi_{\text{out}}\rangle=\left(U_{00}a+U_{01}b\right)|0_{\text{out}}\rangle+\left(U_{10}a+U_{11}b\right)|1_{\text{out}}\rangle
\label{eq:time_evolved}
\end{equation}
where $|0_{\text{out}}\rangle$ and $|1_{\text{out}}\rangle$ are defined as
\begin{align*}
|0_{\text{out}}\rangle & =  e^{-2it_{\mathrm{I}}\cos k}\frac{1}{\sqrt{L}}\sum_{x=M(k)+1}^{M(k)+L}e^{ikx}|x,3\rangle &
|1_{\text{out}}\rangle & =  e^{-2it_{\mathrm{I}}\cos k}\frac{1}{\sqrt{L}}\sum_{x=M(k)+1}^{M(k)+L}e^{ikx}|x,4\rangle.
\end{align*}
Using the results of \sec{truncating}, the error in approximating $|\psi_{K}(t_{\mathrm{I}})\rangle$
by $|\psi_{\text{out}}\rangle$ goes to zero polynomially quickly
as $L$ grows: 
\begin{equation}
\left\Vert |\psi_{K}(t_{\mathrm{I}})\rangle-|\psi_{\text{out}}\rangle\right\Vert =\O(L^{-{1}/{4}}).
\label{eq:error_single_qubit}
\end{equation}
The effect of evolving the input state $|\psi(0)\rangle$ for time $t_{\mathrm{I}}$ is depicted in \fig{1cartoon}.

\begin{figure}
\centering
\capstart
\begin{tikzpicture}[scale = 1.2,
    thin,
    attach/.style={circle,fill=white,draw=black,
      inner sep=1.25pt,minimum size=0pt},
    vert/.style={circle,draw=black,fill=black,
      inner sep=1.25pt,minimum size=0pt},
    attach/.style={circle,draw=black,fill=white,
      inner sep=1.25pt,minimum size=0pt},
    dots/.style={circle,fill=black,
      inner sep=.5pt,minimum size=0pt}]

   \draw (4.75,.35) circle (.83) node {$\widehat{G}$};

   \foreach \y in {0,.7}{
   \begin{scope}[yshift = \y cm]
   \draw (0,0) -- (1.15,0);
   \draw[densely dotted] (1,0) -- (1.32,0);
   \draw[densely dotted] (1.68,0) -- (2,0);
   \draw (1.85,0) -- (4,0);
   \begin{scope}[xshift=-.5cm]
   \draw (6,0) -- (8.15,0);
   \draw[densely dotted] (8,0) -- (8.32,0);
   \draw[densely dotted] (8.68,0) -- (9,0);
   \draw (8.85,0) -- (10,0);
   \end{scope}
   \foreach \x in {1,2,3,6.5,7.5,8.5}{
   \node at (\x, 0) [vert] {};}

   \node at (4,0) [attach] {};
   \node at (5.5,0) [attach]{};
   \node at (0,0) [attach] {};
   \node at (9.5,0) [attach] {};

   \end{scope}}

   \node[anchor=north] at (3.75,0) {(1,2)};
   \node[anchor=north] at (2.75,0) {(2,2)};
   \node[anchor=north] at (1.75,0) {(3,2)};

   \node[anchor=north] at (-.25,0) {($K$,2)};
   
   \node[anchor=south] at (3.75,.7) {(1,1)};
   \node[anchor=south] at (2.75,.7) {(2,1)};
   \node[anchor=south] at (1.75,.7) {(3,1)};

   \node[anchor=south] at (-.25,.7) {($K$,1)};
   
   \node[anchor=south] at (5.75,.7) {(1,3)};
   \node[anchor=south] at (6.75,.7) {(2,3)};
   \node[anchor=south] at (7.75,.7) {(3,3)};

   \node[anchor=south] at (9.75,.7) {($K$,3)};
      
   \node[anchor=north] at (5.75,0) {(1,4)};
   \node[anchor=north] at (6.75,0) {(2,4)};
   \node[anchor=north] at (7.75,0) {(3,4)};
;
   \node[anchor=north] at (9.75,0) {($K$,4)};
\end{tikzpicture}

\caption{A graph $G(K)$ used to perform a single-qubit gate on an encoded qubit. }
\label{fig:single_qubit_graph}
\end{figure}
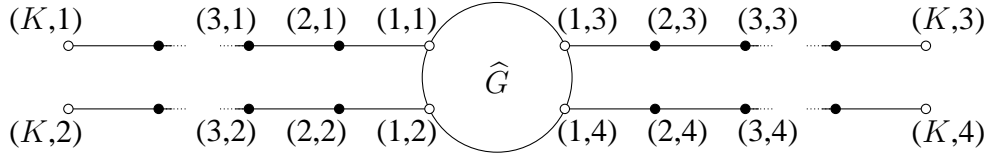

To prove equation \eq{error_single_qubit}, we write
\[
\left\Vert |\psi_{K}(t_{\mathrm{I}})\rangle-|\psi_{\text{out}}\rangle\right\Vert \leq\left\Vert |\psi_{K}(t_{\mathrm{I}})\rangle-|\psi_{\infty}(t_{\mathrm{I}})\rangle\right\Vert +\left\Vert |\psi_{\infty}(t_{\mathrm{I}})\rangle-|\psi_{\text{out}}\rangle\right\Vert .\]
We then apply the bounds
\begin{equation}
\left\Vert |\psi_{\infty}(t_{\mathrm{I}})\rangle-|\psi_{\text{out}}\rangle\right\Vert =\O(L^{-{1}/{4}})\label{eq:bound1}
\end{equation}
and
\begin{equation}
\left\Vert |\psi_{K}(t_{\mathrm{I}})\rangle-|\psi_{\infty}(t_{\mathrm{I}})\rangle\right\Vert =\O(L^{-{1}/{4}}).\label{eq:bound2}
\end{equation}
Equation \eq{bound1} follows from \thm{singlepart}.
Applying the theorem with $T=t_{\mathrm{I}}$ and $M=M(k)$, we find
\begin{equation}
\left\Vert |\psi_{\infty}(t_{\mathrm{I}})\rangle-a|\alpha^{0}(t_{\mathrm{I}})\rangle-b|\alpha^{1}(t_{\mathrm{I}})\rangle\right\Vert =\O(L^{-{1}/{4}}),\label{eq:bound_from_thm}
\end{equation}
where  
\begin{align*}
|\alpha^{0}(t_{\mathrm{I}})\rangle & =  \frac{1}{\sqrt{L}}e^{-2it_{\mathrm{I}}\cos k}\sum_{x = -\lfloor 2t_{\mathrm{I}} \sin k \rfloor-M(k)-L}^{-\lfloor 2t_{\mathrm{I}} \sin k \rfloor-M(k)-1}e^{ikx}\left(U_{00}|x,3\rangle+U_{10}|x,4\rangle\right) \\
& = \frac{1}{\sqrt{L}}e^{-2it_{\mathrm{I}}\cos k}\sum_{x=M(k)+1}^{M(k)+L}e^{ikx}\left(U_{00}|x,3\rangle+U_{10}|x,4\rangle\right) + \O(L^{-1/2})
\end{align*}
and similarly
\begin{align*}
|\alpha^{1}(t_{\mathrm{I}})\rangle & = \frac{1}{\sqrt{L}}e^{-2it_{\mathrm{I}}\cos k}\sum_{x=M(k)+1}^{M(k)+L}e^{ikx}\left(U_{01}|x,3\rangle+U_{11}|x,4\rangle\right) + \O(L^{-1/2})
\end{align*}
where the $\O(L^{-{1}/{2}})$ error terms arise from approximating the upper and lower limits of the summation. Comparing with equation \eq{time_evolved}, we see that \eq{bound_from_thm} implies \eq{bound1}.

The bound \eq{bound2} follows from equation \eq{trunc_paths} (the truncation lemma with $|\Phi\rangle  = |\psi(0)\rangle$, $H = H_{G(\infty)}^{(1)}$, $\tilde{H}  =  H_{G(K)}^{(1)}$, $W =  H$, $\delta = \O(L^{-{1}/{4}})$, $N_0 =  M(k)$, and $T = t_{\mathrm{I}}$).

\subsection{A two-qubit gate}\label{sec:2qubit_calc}

The $\CD$ gate is implemented using the graph shown in \fig{Graph-used-to-1}. In this section we specify the logical input states, the logical output states, the distances $X$, $Z$, and $W$ appearing in the figure, and the total evolution time. With these choices, we show that a $\CD$ gate is applied to the logical states at the end of the time evolution under the quantum walk Hamiltonian (up to error terms that are $\O(L^{-{1}/{4}})$). The results of this section pertain to the two-particle Hamiltonian $H^{(2)}_{G'}$ for the graph $G'$ shown in \fig{Graph-used-to-1}.

The logical input states are
\begin{equation*}
|0_{\text{in}}\rangle^c=\frac{1}{\sqrt{L}}\sum_{x=M(-\frac{\pi}{4})+1}^{M(-\frac{\pi}{4})+L}e^{-i\frac{\pi}{4}x}|x,1\rangle \qquad |1_\text{in}\rangle^c=\frac{1}{\sqrt{L}}\sum_{x=M(-\frac{\pi}{4})+1}^{M(-\frac{\pi}{4})+L}e^{-i\frac{\pi}{4}x}|x,2\rangle
\end{equation*}
for the computational qubit and
\begin{equation*}
|0_{\text{in}}\rangle^\med=\frac{1}{\sqrt{L}}\sum_{y=M(-\frac{\pi}{2})+1}^{M(-\frac{\pi}{2})+L}e^{-i\frac{\pi}{2}y}|y,4\rangle \qquad |1_\text{in}\rangle^\med=\frac{1}{\sqrt{L}}\sum_{y=M(-\frac{\pi}{2})+1}^{M(-\frac{\pi}{2})+L}e^{-i\frac{\pi}{2}y}|y,3\rangle
\end{equation*}
 for the mediator qubit. We define symmetrized (or antisymmetrized) logical input states for $a,b\in\{0,1\}$ as
\begin{align*}
|a b_{\text{in}}\rangle^{c,\med} &=\text{Sym}(|a_{\text{in}}\rangle^c |b_{\text{in}}\rangle^\med )\\
& =\frac{1}{\sqrt{2}} \left(|a_{\text{in}}\rangle^c |b_{\text{in}}\rangle^\med \pm |b_{\text{in}}\rangle^\med|a_{\text{in}}\rangle^c\right).
\end{align*}

\begin{figure}
\centering
\capstart
\begin{tikzpicture} [scale=1.5,
	vert/.style={circle, draw=black, fill=black,inner sep=1pt, minimum width=0pt},
	dots/.style={circle, fill=black,inner sep=1pt, minimum width=0pt},
	switch/.style={circle,draw=black,inner sep=6pt,minimum width=0pt},
	attach/.style={circle,fill=white, draw=black,inner sep=1pt, minimum width=0pt}]

  \foreach \y in {0,0.5,3,3.5}{
  \begin{scope}[yshift=\y cm]
  \foreach \x in {0,1,2,6,7,8}{
	\begin{scope}[xshift = \x cm]
	  \node at (0,0) [vert]{};
	\end{scope}}
	\draw (0,0) -- (2,0);
	\draw (6,0) -- (8,0);
	\draw[densely dotted] (2,0)--(2.2,0);
	\draw[densely dotted] (5.8,0)--(6,0);
	\node at (0,0) [attach] {};
	\node at (8,0) [attach] {};
 \end{scope}}

  \node (switch2) at (4,0.5)[switch]{};
  \node (switch1) at (4,3)[switch]{};
  
  \node at (3,0.5) [vert]{};
  \node at (3,3) [vert]{};
  \node at (5,0.5) [vert]{};
  \node at (5,3) [vert]{};
  
  \draw (switch2.west) -- (3,0.5);
  \draw[densely dotted] (3,0.5) -- (2.8,0.5);
  \draw (switch2.east) -- (5,0.5);
  \draw[densely dotted] (5,0.5) -- (5.2,0.5);
  
  \draw (switch1.west) -- (3,3);
  \draw[densely dotted] (3,3) -- (2.8,3);
  \draw (switch1.east) -- (5,3);
  \draw[densely dotted] (5,3) -- (5.2,3);
  
  \node at (3.5,0) [vert]{};
  \node at (4.5,0) [vert]{};
  \node at (3.5,3.5) [vert]{};
  \node at (4.5,3.5) [vert]{};
  
  \draw (3.5,0) -- (4.5,0);
  \draw (3.5,3.5) -- (4.5,3.5);
  \draw[densely dotted] (3.3,0) -- (4.7,0);
  \draw[densely dotted] (3.3,3.5) -- (4.7,3.5);

  \node at (4,1.25) [vert] {};
  \node at (4,2.25) [vert] {};
  
  \draw (switch2.north) -- (4,1.25);
  \draw[densely dotted] (4,1.25) -- (4,1.45);
  
  \draw (switch1.south) -- (4,2.25);
  \draw[densely dotted] (4,2.25)-- (4,2.05);
  
  \draw[line width=.7pt] (switch1.west) to[out=0,in=90] (switch1.south);
  \draw[line width=2.1pt] (switch1.east) to[out=180,in=90] (switch1.south);
  \draw[line width=.7pt,white] (switch1.east) to[out=180,in=90] (switch1.south);

  \draw[line width=.7pt] (switch2.east) to[out=180,in=-90] (switch2.north);
  \draw[line width=2.1pt] (switch2.west) to[out=0,in=-90] (switch2.north);
  \draw[line width=.7pt,white] (switch2.west) to[out=0,in=-90] (switch2.north);
  
  \node at (switch2.north) [attach] {};
  \node at (switch1.south) [attach] {};
  \node at (switch2.east) [attach] {};
  \node at (switch1.east) [attach] {};
  \node at (switch2.west) [attach] {};
  \node at (switch1.west) [attach] {};
  
  \node at (0,0) [below] {(1,4)};
  \node at (1,0) [below] {(2,4)};
  \node at (2,0) [below] {(3,4)};
  \node at (3.5,0) [below] {};
  \node at (4.5,0) [below] {};
  \node at (6,0) [below] {($2X\!+\! Z\! + \! 4$,4)};
  \node at (8,0) [below] {($2X\! +\! Z\! +\! 6$,4)};
  
  \node at (0,3.5) [above] {(1,1)};
  \node at (1,3.5) [above] {(2,1)};
  \node at (2,3.5) [above] {(3,1)};
  \node at (3.5,3.5) [above] {};
  \node at (4.5,3.5) [above] {};
  \node at (6,3.5) [above] {($2W\! + \! Z\! + \! 2$,1)};
  \node at (8,3.5) [above] {($2W\! + \! Z\! + \! 4$,1)};
  
  \node at (0,3) [below] {(1,2)};
  \node at (1,3) [below] {(2,2)};
  \node at (2,3) [below] {(3,2)};
  \node at (3,3) [below] {($W\!-\! 1$,2)};
  \node[anchor=south east] at (3.87,3) {($W$,2)};
  
  \node at (0,.5) [above] {(1,3)};
  \node at (1,.5) [above] {(2,3)};
  \node at (2,.5) [above] {(3,3)};
  \node at (3,.5) [above] {($X\! -\! 1$,3)};
  \node at (3.87,.5) [below left] {($X$,3)};
  
  \node at (8,.5) [above] {($2W\! + \! Z\! + \! 4$,2)};
  \node at (6,.5) [above] {($2W\! + \! Z\! + \! 2$,2)};
  \node at (4.13,.5) [below right] {($W\! + \! Z\! + \! 5$,2)};  
  
  \node at (8,3) [below] {($2X\! + \! Z\! + \! 6$,3)};
  \node at (6,3) [below] {($2X\! + \! Z\! + \! 4$,3)};
  \node at (4.13,3) [above right] {($X\! + \! Z\! + \! 7$,3)};
  
  \node at (4,2.87) [below right] {(1,5)};
  \node at (4,2.25) [below right] {(2,5)};
  
  \node at (4,.63) [above right] {($Z$,5)};
  \node at (4,1.25) [above right] {($Z\! -\! 1$,5)};
  
  \node at (-.4,0) [left] {$0_{\med,\text{in}}$};
  \node at (-.4,.5) [left] {$1_{\med,\text{in}}$};
  \node at (-.4,3) [left] {$1_{c,\text{in}}$};
  \node at (-.4,3.5) [left] {$0_{c,\text{in}}$};
  
  \node at (8.7,0) [right] {$0_{\med,\text{out}}$};
  \node at (8.7,.5) [right] {$1_{c,\text{out}}$};
  \node at (8.7,3) [right] {$1_{\med,\text{out}}$};
  \node at (8.7,3.5) [right] {$0_{c,\text{out}}$};
\end{tikzpicture}
\caption{Graph $G'$ used to implement the $\CD$ gate. The integers $Z$, $X$, and $W$ are specified in equations \eq{Z_eq}, \eq{X_eq}, and \eq{W_eq}, respectively.}
\label{fig:Graph-used-to-1}
\end{figure}
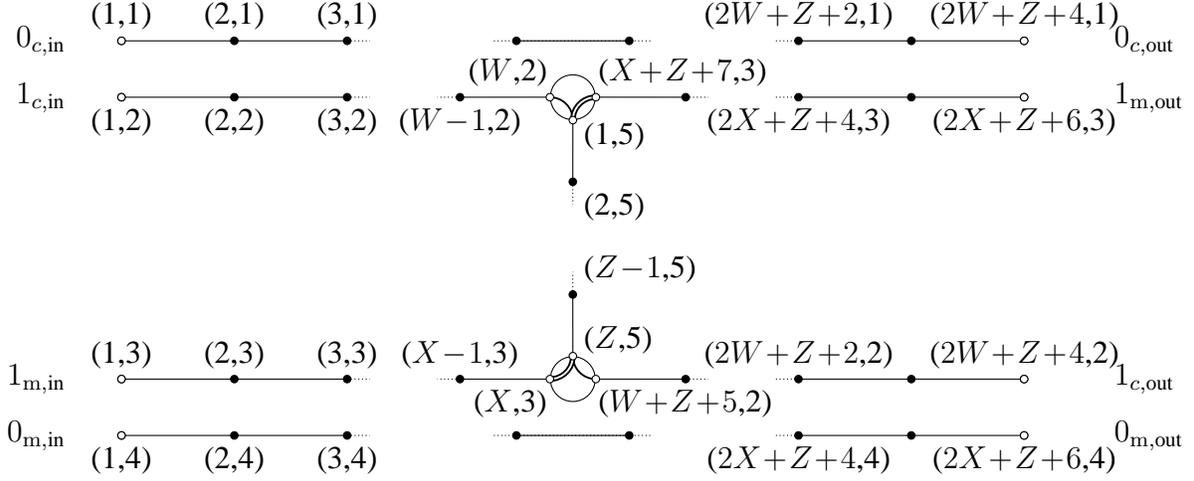

We choose the distances $Z$, $X$, and $W$ from \fig{Graph-used-to-1}
to be \begin{align}
Z & = 4L \label{eq:Z_eq} \\
X & = d_{2}+L+M\left(-\frac{\pi}{2}\right) \label{eq:X_eq}\\
W & = d_{1}+L+M\left(-\frac{\pi}{4}\right) \label{eq:W_eq}
\end{align}
where
\begin{align*}
d_{1} & = M\left(-\frac{\pi}{4}\right) \\
d_{2} & = \left\lceil \frac{5L+2d_{1}}{\sqrt{2}}-\frac{5}{2}L\right\rceil. \end{align*}
With these choices, a wave packet moving with speed $\sqrt{2}$ travels
a distance $Z+2d_{1}+L=5L+2d_{1}$ in approximately the same time that
a wave packet moving with speed $2$ takes to travel a distance $Z+2d_{2}+L=5L+2d_{2}$,
since
\[
t_{\mathrm{II}}=\frac{5L+2d_{1}}{\sqrt{2}}\approx\frac{5L+2d_{2}}{2}.
\]

We claim that the logical input states evolve into logical output states (defined below) with a phase of $e^{i\theta}$ applied in the case where both particles are in the logical state $1$.  Specifically,
\begin{align}
\left\Vert e^{-iH_{G'}^{(2)}t_{\mathrm{II}}}|00_{\text{in}}\rangle^{c,\med}-|00_{\text{out}}\rangle^{c,\med}\right\Vert  & = \O(L^{-{1}/{4}})\label{eq:bound00}\\
\left\Vert e^{-iH_{G'}^{(2)}t_{\mathrm{II}}}|01_{\text{in}}\rangle^{c,\med}-|01_{\text{out}}\rangle^{c,\med}\right\Vert  & = \O(L^{-{1}/{4}})\label{eq:bound01}\\
\left\Vert e^{-iH_{G'}^{(2)}t_{\mathrm{II}}}|10_{\text{in}}\rangle^{c,\med}-|10_{\text{out}}\rangle^{c,\med}\right\Vert  & = \O(L^{-{1}/{4}})\label{eq:bound10}\\
\left\Vert e^{-iH_{G'}^{(2)}t_{\mathrm{II}}}|11_{\text{in}}\rangle^{c,\med} - e^{i\theta}|11_{\text{out}}\rangle^{c,\med}\right\Vert  & = \O(L^{-{1}/{4}})\label{eq:bound11}
\end{align}
 where, letting $Q_{1}=2W+Z+4-M\left(-{\pi}/{4}\right)-L$ and $Q_{2}=2X+Z+6-M\left(-{\pi}/{2}\right)-L$,
\begin{align*}
|0_\text{out}\rangle^c &=\frac{e^{-it_{\mathrm{II}}\sqrt{2}}}{\sqrt{L}}\sum_{x=Q_{1}+1}^{Q_{1}+L}e^{-i\frac{\pi}{4}x}|x,1\rangle &
|1_\text{out}\rangle^c &=\frac{e^{-it_{\mathrm{II}}\sqrt{2}}}{\sqrt{L}}\sum_{x=Q_{1}+1}^{Q_{1}+L}e^{-i\frac{\pi}{4}x}|x,2\rangle \\
|0_\text{out}\rangle^\med &=\frac{1}{\sqrt{L}}\sum_{y=Q_{2}+1}^{Q_{2}+L}e^{-i\frac{\pi}{2}y}|y,4\rangle &|1_\text{out}\rangle^\med &=\frac{1}{\sqrt{L}}\sum_{y=Q_{2}+1}^{Q_{2}+L}e^{-i\frac{\pi}{2}y}|y,3\rangle
\end{align*}
and $|a b_\text{out}\rangle^{c,\med}=\text{Sym}\left(|a_{\text{out}}\rangle^c|b_{\text{out}}\rangle^\med\right)$. 

Note that the input states are wave packets located a distance $M(k)$ from the ends of the input paths on the left-hand side of the graph in \fig{Graph-used-to-1}. Similarly, the output logical states are wave packets located a distance $M(k)$ from the ends of the output paths on the right-hand side.

The first three bounds \eq{bound00}, \eq{bound01}, and \eq{bound10} are relatively easy to show, since in each case the two particles are supported on disconnected subgraphs and therefore do not interact. In each of these three cases we can simply analyze the propagation of the one-particle starting states through the graph. The symmetrized (or antisymmetrized) starting state then evolves into the symmetrized (or antisymmetrized) tensor product of the two output states.

For example, with input state $|00_{\text{in}}\rangle^{c,\med}$, the evolution of the particle with momentum $-{\pi}/{4}$ occurs only on the top path and the evolution of the particle with momentum $-{\pi}/{2}$ occurs only on the bottom path. Starting from the initial state $|0_\text{in}\rangle^c$ and evolving for time $t_{\mathrm{II}}$ with the single-particle Hamiltonian for the top path, we obtain the final state
\[
|0_\text{out}\rangle^c+\O(L^{-{1}/{4}})
\]
using the method of \sec{truncating}. Similarly, starting from the initial state $|0_\text{in}\rangle^\med$ and evolving for time $t_{\mathrm{II}}$ with the single-particle Hamiltonian
for the bottom path of the graph we obtain the final state
\[
|0_\text{out}\rangle^{\med}+\O(L^{-{1}/{4}}).
\]
Putting these bounds together we get the bound \eq{bound00}.

In the case where the input state is $|10_{\text{in}}\rangle^{c,\med}$ (or $|01_{\text{in}}\rangle^{c,\med}$) the single-particle evolution for the particle with momentum $-{\pi}/{4}$ (or $-{\pi}/{2}$) is slightly more complicated, as in this case the particle moves through the momentum switches and the vertical path. The S-matrix of the momentum switch at the relevant momenta is given by equation \eq{switch_S}. At momentum $-{\pi}/{4}$, the momentum switch has the same S-matrix as a path with $4$ vertices (including the input and output vertices). At momentum $-{\pi}/{2}$, it has the same S-matrix as a path with $5$ vertices (including input and output vertices). Note that our labeling of vertices on the output paths (in \fig{Graph-used-to-1}) takes this into account. The first vertices on the output paths connected to the momentum switches are labeled $(X+Z+7,3)$ and $(W+Z+5,2)$, respectively, reflecting the fact that a particle with momentum $-{\pi}/{4}$ has traveled $W$ vertices on the input path, $Z$ vertices through the middle segment, and has effectively traveled an additional $4$ vertices inside the two switches. Similarly, a particle with momentum $-{\pi}/{2}$ effectively sees an additional $6$ vertices from the two momentum switches.

To get the bound \eq{bound10} we have to analyze the single-particle evolution
for the computational particle initialized in the state $|1_\text{in}\rangle^c$. 
We claim that, after time $t_{\mathrm{II}}$, the time-evolved state is
\[
|1_\text{out}\rangle^c+\O(L^{-{1}/{4}}).
\]
It is easy to see why this should be the case in light of our discussion above: when scattering at momentum $-{\pi}/{4}$, the graph in \fig{Graph-used-to-1} is equivalent to one where each momentum switch is replaced by a path with $2$ internal vertices connecting the relevant input/output vertices.

To make this precise, we use the method described in \sec{more_complicated_graphs} for analyzing scattering through sequences of overlapping graphs using the truncation lemma. Here we should choose subgraphs $G_{1}$ and $G_{2}$ of the graph $G'$ in \fig{Graph-used-to-1} that overlap on the vertical path but where each subgraph contains only one of the momentum switches. A convenient choice is to take $G_{1}$ to be the subgraph containing the top switch and the paths connected to it (the vertices $(1,2),\ldots,(W,2)$, $(1,5),\ldots,(Z,5)$ and $(X+Z+7,3),\ldots,(2X+Z+6,3)$). Similarly, choose $G_{2}$ to be the bottom switch along with the three paths connected to it. The graphs $G_{1}$ and $G_{2}$ both contain the vertices $(1,5),\ldots,(Z,5)$ along the vertical path. Break up the total evolution time into two intervals $[0,t_{\alpha}]$ and $[t_{\alpha},t_{\mathrm{II}}]$. Choose $t_{\alpha}$ so that the wave packet, evolved for this time with $H_{G_1}^{(1)}$, travels through the top switch and ends up a distance $\Theta(L)$ from each switch, partway along the vertical path (up to terms bounded as $\O(L^{-{1}/{4}})$, as in \sec{truncating}). With this choice, the single-particle evolution with the Hamiltonian for the full graph is approximated by the evolution with $H_{G_1}^{(1)}$ on this time interval (see \sec{more_complicated_graphs}). At time $t_\alpha$, the particle is outgoing with respect to scattering from the graph $G_1$, but incoming with respect to $G_2$. On the interval $[t_{\alpha},t_{\mathrm{II}}]$ the time evolution is approximated by evolving the state with $H_{G_2}^{(1)}$. During this time interval the particle travels through the bottom switch onto the final path, and at $t_{\mathrm{II}}$ is a distance $M(-{\pi}/{4})$ from the endpoint of the output path. Both switches have the same S-matrix (at momentum $-{\pi}/{4}$) as a path of length $4$, so this analysis gives the output state $|10_\text{out}\rangle^{c,\med}$ up to terms bounded as $\O(L^{-{1}/{4}})$, establishing \eq{bound10}. For the bound \eq{bound01}, we apply a similar analysis to the trajectory of the mediator particle.

The case where the input state is $|11_{\text{in}}\rangle^{c,\med}$ is more involved but proceeds similarly. In this case, to analyze the time evolution we divide the time interval $[0,t_{\mathrm{II}}]$ into three segments $[0,t_{A}]$, $[t_{A},t_{B}]$, and $[t_{B},t_{\mathrm{II}}]$. For each of these three time intervals we choose a subgraph $G_{A}$, $G_{B}$, $G_C$ of the graph $G'$ in \fig{Graph-used-to-1} and we approximate the time evolution by evolving with the Hamiltonian on the associated subgraph. We then use the truncation lemma to show that, on each time interval, the evolution generated by the Hamiltonian for the appropriate subgraph approximates the evolution generated by the full Hamiltonian, with error $\O(L^{-{1}/{4}})$. Up to these error terms, at times $t=0$, $t=t_A$, $t=t_B$, and $t=t_{\mathrm{II}}$ the time-evolved state 
\[
e^{-iH_{G'}^{(2)}t}|11_{\text{in}}\rangle^{c,\med}
\]
has both particles in square wave packet states, each with support only on $L$ vertices of the graph, as depicted in \fig{11_scattering_cartoon}.

We take $G_A$ to be the subgraph obtained from $G'$ by removing the vertices labeled $(\lceil 1.85L\rceil,5)\allowbreak, \ldots,\allowbreak (\lceil 1.90 L\rceil,5)$ in the vertical path. By removing this interval of consecutive vertices, we disconnect the graph into two components where the initial state $|11_{\text{in}}\rangle^{c,\med}$ has one particle in each component. This could be achieved by removing a single vertex, but instead we remove an interval of approximately $0.05L$ vertices to separate the components of $G_A$ by more than the interaction range $C$ (for sufficiently large $L$), simplifying our use of the truncation lemma.

 We choose $t_{A}={3L}/{2}$. Consider the time evolution of the initial state $|11_\text{in}\rangle^{c,\med}$ with the two-particle Hamiltonian $H_{G_A}^{(2)}$ for time $t_A$. The states $|1_ {\text{in}}\rangle^c$ and $|1_\text{in}\rangle^{\med}$ are supported on disconnected components of the graph $G_A$, so we can analyze the time evolution of the state $|11_\text{in}\rangle^{c,\med}$ under $H_{G_A}^{(2)}$ by analyzing two single-particle problems, using the results of \sec{truncating} for each particle. During the interval $[0,t_A]$,  each particle passes through one switch, ending up a distance $\Theta(L)$ from the switch that it passed through and $\Theta(L)$ from the vertices that have been removed, as shown in \fig{11_scattering_cartoon}(b) (with error at most $\O(L^{-{1}/{4}})$). Up to these error terms, the support of each particle remains at least $N_0=\Theta(L)$ vertices from the endpoints of the graph, so we can apply the truncation lemma using $H=H_{G'}^{(2)}$, $W=\tilde{H}=H_{G_A}^{(2)}$, $T=t_{\mathrm{A}}$, and $\delta=\O(L^{-{1}/{4}})$. Here $P$ is the projector onto states where both particles are located at vertices of $G_A$. We have $P H_{G'}^{(2)}P=H_{G_A}^{(2)}$ since the number of vertices in the removed segment is greater than the interaction range $C$. Applying the truncation lemma gives
\[
\left\Vert e^{-iH_{G_A}^{(2)}t_A}|11_\text{in}\rangle^{c,\med}-e^{-iH_{G'}^{(2)}t_A}|11_\text{in}\rangle^{c,\med}\right\Vert=\O(L^{-{1}/{4}}).
\]

We approximate the evolution on the interval $[t_A,t_B]$ using the two-particle  Hamiltonian $H_{G_B}^{(2)}$, where $G_B$  is the vertical path $(1,5),\ldots,(Z,5)$. Using the result of \sec{truncating}, we know that (up to terms bounded as $\O(L^{-{1}/{4}})$) the wave packets move with their respective speeds and acquire a phase of $e^{i\theta}$ as they pass each other. We choose $t_B={5L}/{2}$ so that during the evolution the wave packets have no support on vertices within a distance $\Theta(L)$ from the endpoints of the vertical segment where the graph has been truncated (again up to terms bounded as $\O(L^{-{1}/{4}})$). Using $H_{G_B}^{(2)}$ (rather than $H_{G'}^{(2)}$) to evolve the state on this interval, we incur errors bounded as $\O(L^{-{1}/{4}})$ (using the truncation lemma with $N_0=\Theta(L)$, $W=\tilde{H}=H_{G_B}^{(2)}$, $H=H_{G'}^{(2)}$, and $\delta=\O(L^{-{1}/{4}})$).

We choose $G_C=G_A$; in the final interval $[t_{B},t_{\mathrm{II}}]$ we evolve using the Hamiltonian  $H_{G_A}^{(2)}$ again, and we use the truncation lemma as we did for the first interval. The initial state is approximated by two wave packets supported on disconnected sections of $G_A$ and the evolution of this initial state reduces to two single-particle scattering problems. During the interval $[t_B,t_{\mathrm{II}}]$, each particle passes through a second switch, and at time $t_{\mathrm{II}}$ is a distance $M(k)$ from the end of the appropriate output path. 

Our analysis shows that for the input state $|11_\text{in}\rangle^{c,\med}$ the only effect of the interaction is to alter the global phase of the final state by a factor of $e^{i\theta}$ relative to the case where no interaction is present, up to error terms bounded as $\O(L^{-{1}/{4}})$. This establishes equation \eq{bound11}. In \fig{11_scattering_cartoon} we illustrate the movement of the two wave packets through the graph when the initial state is $|11_\text{in}\rangle^{c,\med}$.

\begin{figure}
\centering
\capstart
\begin{tikzpicture}[   scale=0.7,   
	dots/.style={circle,draw=black,fill=black,inner sep=0pt,minimum size=.5 mm}, 
	splitter/.style={circle,draw=black,fill=white,inner sep=0pt,minimum size=4mm},   
	attach/.style={circle,draw=black,fill=white,inner sep=0pt,minimum size=1mm}]
	
\begin{scope}[yshift= 6 cm]
  \node at (-.75,4) {(a)};

  \draw (0,0) to (10,0);   
  \draw (0,3) to (10,3);
  \foreach \x in {0,.2,.4,...,10}{
  \node at (\x, 0) [dots] {};
  \node at (\x, 3) [dots] {};
  }
  \foreach \y in {0,.2,.4,...,3}
  \node at (5,\y) [dots] {};
  
  \node (splitter42) at (5,0) [splitter] {};   
  \node (splitter41) at (5,3) [splitter] {};   
  \draw (splitter41.south) to (splitter42.north);      
  \draw (splitter41.west) to[out=0,in=90] (splitter41.south);   
  \draw[line width=1.2pt] (splitter41.south) to[out=90,in=180]           
  (splitter41.east);   
  \draw[line width=.4pt,white] (splitter41.south) to[out=90,in=180] (splitter41.east);
  \draw[line width=1.2pt] (splitter42.west) to[out=0,in=-90]           (splitter42.north);   
  \draw[line width=.4pt,white] (splitter42.west) to[out=0,in=-90]           (splitter42.north);   
  \draw (splitter42.north) to[out=-90,in=180] (splitter42.east);
  \node at (splitter41.east) [attach]{};
  \node at (splitter41.west) [attach]{};
  \node at (splitter41.south) [attach]{};
  \node at (splitter42.east) [attach]{};
  \node at (splitter42.west) [attach]{};
  \node at (splitter42.north) [attach]{};
  
  \node at (0,0) [attach]{};
  \node at (0,3) [attach]{};
  \node at (10,0) [attach]{};
  \node at (10,3) [attach]{};  

  \draw (1.05,3.15) to (1.25,3.15) to (1.25,3.6) to node[above]          
    {${\pi}/{4}\rightarrow$} (3.25,3.6) to (3.25,3.15) to (3.45,3.15);
  \draw (1.05,0.15) to (1.25,0.15) to (1.25,0.6) to node[above]          
    {${\pi}/{2}\rightarrow$} (3.25,0.6) to (3.25,0.15) to (3.45,0.15);
  \draw[|-|] (0,2.5) to node[below] {$M(-\frac{\pi}{4})$} (1.25,2.5);   
  \draw[|-|] (1.25,2.5) to node[below]{$L$}   (3.25,2.5);   
  \draw[|-|] (3.25,2.5) to node[below]{$d_1$} (4.75,2.5);
  \draw[|-|] (0,-.5) to node[below] {$M(-\frac{\pi}{2})$} (1.25,-.5);   
  \draw[|-|] (1.25,-.5) to node[below]{$L$}   (3.25,-.5);   
  \draw[|-|] (3.25,-.5) to node[below]{$d_2$} (4.75,-.5);
\end{scope}
  

\begin{scope}[xshift= 11.5cm, yshift=6cm]
  \node at (-.75,4) {(b)};

  \draw (0,0) to (10,0);   
  \draw (0,3) to (10,3);
  \foreach \x in {0,.2,.4,...,10}{
  \node at (\x, 0) [dots] {};
  \node at (\x, 3) [dots] {};
  }
  \foreach \y in {0,.2,.4,...,3}
  \node at (5,\y) [dots] {};
  
  \node (splitter42) at (5,0) [splitter] {};   
  \node (splitter41) at (5,3) [splitter] {};   
  \draw (splitter41.south) to (splitter42.north);      
  \draw (splitter41.west) to[out=0,in=90] (splitter41.south);   
  \draw[line width=1.2pt] (splitter41.south) to[out=90,in=180]           
  (splitter41.east);   
  \draw[line width=.4pt,white] (splitter41.south) to[out=90,in=180] (splitter41.east);
  \draw[line width=1.2pt] (splitter42.west) to[out=0,in=-90]           (splitter42.north);   
  \draw[line width=.4pt,white] (splitter42.west) to[out=0,in=-90]           (splitter42.north);   
  \draw (splitter42.north) to[out=-90,in=180] (splitter42.east);
  \node at (splitter41.east) [attach]{};
  \node at (splitter41.west) [attach]{};
  \node at (splitter41.south) [attach]{};
  \node at (splitter42.east) [attach]{};
  \node at (splitter42.west) [attach]{};
  \node at (splitter42.north) [attach]{};
  
  \node at (0,0) [attach]{};
  \node at (0,3) [attach]{};
  \node at (10,0) [attach]{};
  \node at (10,3) [attach]{};
  \begin{scope}[yshift = -10 cm]
    \draw (5.1,11.7) to (5.1,11.8) to (5.5,11.8) to node[right]          
      {$\frac{\pi}{4}\downarrow$} (5.5,12.4) to (5.1,12.4) to (5.1,12.5);
    \draw (4.9,11.3) to (4.9,11.2) to (4.5,11.2) to node[left]            
      {$\frac{\pi}{2}\uparrow$} (4.5,10.6) to (4.9,10.6) to (4.9,10.5);
    \draw[|-|] (7.95,12.85) to node[right] {$\approx 0.56L$} (7.95,12.4);   
    \draw[|-|] (7.35,12.4)  to node[right] {$L$}   (7.35,11.8);   
    \draw[|-|] (7.95,11.8)  to node[right] {$\approx 1.27L$} (7.95,11.2);   
    \draw[|-|] (7.35,11.2)  to node[right] {$L$}   (7.35,10.6);   
    \draw[|-|] (7.95,10.6)  to node[right] {$\approx 0.17L$} (7.95,10.15);
   \end{scope}  
\end{scope} 
  

  \node at (-.75,4) {(c)};

  \draw (0,0) to (10,0);   
  \draw (0,3) to (10,3);
  \foreach \x in {0,.2,.4,...,10}{
  \node at (\x, 0) [dots] {};
  \node at (\x, 3) [dots] {};
  }
  \foreach \y in {0,.2,.4,...,3}
  \node at (5,\y) [dots] {};
  
  \node (splitter42) at (5,0) [splitter] {};   
  \node (splitter41) at (5,3) [splitter] {};   
  \draw (splitter41.south) to (splitter42.north);      
  \draw (splitter41.west) to[out=0,in=90] (splitter41.south);   
  \draw[line width=1.2pt] (splitter41.south) to[out=90,in=180]           
  (splitter41.east);   
  \draw[line width=.4pt,white] (splitter41.south) to[out=90,in=180] (splitter41.east);
  \draw[line width=1.2pt] (splitter42.west) to[out=0,in=-90]           (splitter42.north);   
  \draw[line width=.4pt,white] (splitter42.west) to[out=0,in=-90]           (splitter42.north);   
  \draw (splitter42.north) to[out=-90,in=180] (splitter42.east);
  \node at (splitter41.east) [attach]{};
  \node at (splitter41.west) [attach]{};
  \node at (splitter41.south) [attach]{};
  \node at (splitter42.east) [attach]{};
  \node at (splitter42.west) [attach]{};
  \node at (splitter42.north) [attach]{};
 
  \node at (0,0) [attach]{};
  \node at (0,3) [attach]{};
  \node at (10,0) [attach]{};
  \node at (10,3) [attach]{};
  \begin{scope}[yshift = - 5cm]
  \draw (5.1,6.3) to (5.1,6.20) to (5.5,6.20) to node[right]            {$\frac{\pi}{4}\downarrow$} (5.5,5.6) to (5.1,5.6) to (5.1,5.5);
  \draw (4.9,6.7) to (4.9,6.8) to (4.5,6.8) to node[left]            {$\frac{\pi}{2}\uparrow$} (4.5,7.4) to (4.9,7.4) to (4.9,7.5);
  \draw[|-|] (7.95,7.85) to node[right] {$\approx 0.83L$} (7.95,7.4);   
  \draw[|-|] (7.35,7.4)  to node[right] {$L$}   (7.35,6.8);   
  \draw[|-|] (7.95,6.8)  to node[right] {$\approx 0.14L$} (7.95,6.2);   
  \draw[|-|] (7.35,6.2)  to node[right] {$L$}   (7.35,5.6);   
  \draw[|-|] (7.95,5.6)  to node[right] {$\approx 1.03L$} (7.95,5.15);
  \end{scope}
  
\begin{scope}[xshift=11.5cm]
  \node at (-.75,4) {(d)};

  \draw (0,0) to (10,0);   
  \draw (0,3) to (10,3);
  \foreach \x in {0,.2,.4,...,10}{
  \node at (\x, 0) [dots] {};
  \node at (\x, 3) [dots] {};
  }
  \foreach \y in {0,.2,.4,...,3}
  \node at (5,\y) [dots] {};
  
  \node (splitter42) at (5,0) [splitter] {};   
  \node (splitter41) at (5,3) [splitter] {};   
  \draw (splitter41.south) to (splitter42.north);      
  \draw (splitter41.west) to[out=0,in=90] (splitter41.south);   
  \draw[line width=1.2pt] (splitter41.south) to[out=90,in=180]           
  (splitter41.east);   
  \draw[line width=.4pt,white] (splitter41.south) to[out=90,in=180] (splitter41.east);
  \draw[line width=1.2pt] (splitter42.west) to[out=0,in=-90]           (splitter42.north);   
  \draw[line width=.4pt,white] (splitter42.west) to[out=0,in=-90]           (splitter42.north);   
  \draw (splitter42.north) to[out=-90,in=180] (splitter42.east);
  \node at (splitter41.east) [attach]{};
  \node at (splitter41.west) [attach]{};
  \node at (splitter41.south) [attach]{};
  \node at (splitter42.east) [attach]{};
  \node at (splitter42.west) [attach]{};
  \node at (splitter42.north) [attach]{};
 
  \node at (0,0) [attach]{};
  \node at (0,3) [attach]{};
  \node at (10,0) [attach]{};
  \node at (10,3) [attach]{};
  \draw (6.55,3.15) to (6.75,3.15) to (6.75,3.6) to node[above]            {${\pi}/{2}\rightarrow$} (8.75,3.6) to (8.75,3.15) to (8.95,3.15);
  \draw (6.55,.15) to (6.75,.15) to (6.75,.6) to node[above]            {${\pi}/{4}\rightarrow$} (8.75,.6)to (8.75,.15) to (8.95,.15);
  \begin{scope}[yshift=-.25cm]
  \draw[|-|] (5.25,2.75) to node[below] {$d_2$} (6.75,2.75);   
  \draw[|-|] (6.75,2.75) to node[below]{$L$}   (8.75,2.75);   
  \draw[|-|] (8.75,2.75) to node[below]{$M(-\frac{\pi}{2})$} (10,2.75);
  \draw[|-|] (5.25,-.25) to node[below] {$d_1$} (6.75,-.25);   
  \draw[|-|] (6.75,-.25) to node[below]{$L$}   (8.75,-.25);   
  \draw[|-|] (8.75,-.25) to node[below]{$M(-\frac{\pi}{4})$} (10,-.25);
  \end{scope}
\end{scope}
\end{tikzpicture}

\caption{This picture illustrates the scattering process for two wave packets that are incident on the input paths as shown in figure (a) at time $t=0$. Figure (b) shows the location of the two wave packets after a time $t_{A}={3L}/{2}$ and figure (c) shows the wave packets after a time $t_{B}=t_{A}+L$. After the particles pass one another they acquire an overall phase of $e^{i\theta}$. Figure (d) shows the final configuration of the wave packets after a total evolution time $t_{\mathrm{II}}={(Z+2d_{1}+L)}/{\sqrt{2}}$.}
\label{fig:11_scattering_cartoon}
\end{figure}
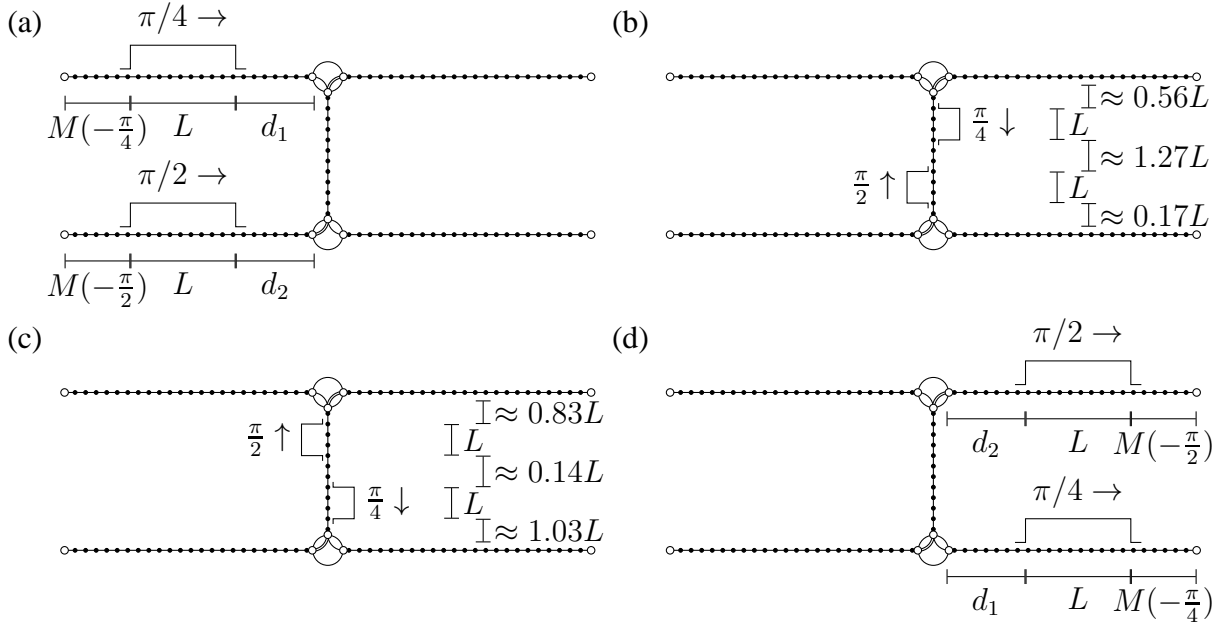

\subsection{Block-by-block analysis of the full graph for a circuit}\label{sec:block_by_block}

In this section we discuss how blocks such as those shown in \fig{squint} and \fig{CPintall} act on encoded data. For each type of block, we first consider the time evolution generated by the Hamiltonian $H_{\text{block}}$ for the block in isolation, i.e., with nothing connected on either side. The truncation lemma lets us use our results about $H_\text{block}$ to prove results about $H_G^{(n+1)}$, the Hamiltonian for the full graph $G$ where the input and output paths of the block are connected to other graphs.

\subsubsection{Blocks applying single-qubit gates}

First consider a block of type I (as shown in \fig{squint}). The results of \sec{1qubit_calc} show that if the input and output paths of this block are not connected to anything, the block applies the correct logical single-qubit gates to each encoded qubit. Define single-particle logical input states $|z_{\text{in}}\rangle^j$ and output states $|z_{\text{out}}\rangle^j$ for each computational qubit ($z\in\{0,1\}$ and $j\in \{1,\ldots,n\}$) and the mediator qubit ($j=\med$) as in \sec{1qubit_calc}.

Suppose that the state of the system at time $t=0$ is
\[
|\psi(0)\rangle=\Sym\left(\sum_{\vec{z}\in\{0,1\}^{n+1}}\alpha_{\vec{z}}|z_{\text{in}}^{1}\rangle^1\ldots|z_{\text{in}}^{n}\rangle^n|z_{\text{in}}^{n+1}\rangle^{\med}\right).
\]
Evolving the state $|\psi(0)\rangle$ for time $t_{\mathrm{I}}={3L}/{2}$
using the Hamiltonian $H_{\text{block}}$ for the block
we get $|\psi(t_{\mathrm{I}})\rangle=e^{-iH_{\text{block}}t_{\mathrm{I}}}|\psi(0)\rangle$.
Letting $h_{1},h_{2},\ldots,h_{n+1}$ be the single-particle Hamiltonians associated with the $n+1$ encoded qubits (all supported on disconnected components of the graph),
\begin{align*}
|\psi(t_{\mathrm{I}})\rangle & = \Sym\left(\sum_{\vec{z}\in\{0,1\}^{n+1}} \alpha_{\vec{z}} \, e^{-ih_{1}t_{\mathrm{I}}} |z_{\text{in}}^{1}\rangle^1 \ldots e^{-ih_{n}t_{\mathrm{I}}} |z_{\text{in}}^{n}\rangle^n e^{-ih_{n+1}t_{\mathrm{I}}}|z_{\text{in}}^{n+1}\rangle^{\med}\right)\\
 & = \Sym\left(\sum_{\vec{z},\vec{x}\in\{0,1\}^{n+1}} \alpha_{\vec{z}} \langle x^{1}|U_{1}|z^{1}\rangle|x_{\text{out}}^{1}\rangle^1\ldots\langle x^{n}|U_{n}|z^{n}\rangle|x_{\text{out}}^{n}\rangle^n\langle x^{n+1}|V_{\med}|z^{n+1}\rangle|x_{\text{out}}^{n+1}\rangle^{\med}\right)\\
& \qquad +\O(nL^{-{1}/{4}})\\
 & = \ket{\mu(t_{\mathrm{I}})} + \O(nL^{-{1}/{4}})
\end{align*}
where in the last line we have defined $\ket{\mu(t_{\mathrm{I}})}$ to be the logical output state with the appropriate unitaries applied. In the above we used the fact that 
\[
e^{-ih_w t_{\mathrm{I}}}|z_{\text{in}}\rangle^w=\sum_{x=0}^{1}\langle x|U_w|z\rangle |x_{\text{out}}\rangle^w+\O(L^{-{1}/{4}})
\]
as proved in \sec{1qubit_calc}.  Note that the error terms for each of the $n+1$ encoded qubits add linearly to give a total error $\O(nL^{-{1}/{4}})$.

Of course we are interested in the scenario where the input and output paths of this graph are connected to other graphs as specified in \sec{description}.
 Let $H_{G}^{(n+1)}$ be the full Hamiltonian including these connections. We use the truncation lemma to show that $e^{-iH_G^{(n+1)} t_{\mathrm{I}}}|\psi(0)\rangle$ is approximated by $|\mu(t_{\mathrm{I}})\rangle$. We let $P$ project onto the subspace of states where all $n+1$ particles are located on vertices of the graph within the block,  $H=H_G^{(n+1)}$, and $W=\tilde{H}=H_{\text{block}}$.\footnote{To ensure that $P H_G^{(n+1)} P = H_{\text{block}}$ we take $L$ large enough that, for any two vertices $i,j$ in the block that are associated with two different encoded qubits, there are no paths of length $\leq C$  in the full graph $G$ that connect $i$ and $j$. This is true provided $L\geq C$  (since any such path has to contain at least $L$ vertices of another block that are part of a $\CD$ gate subgraph). Since $C$ is a constant and $L$ grows with the size of the computation, this requirement is automatically satisfied. Note that there are no paths of length $\leq C$ that connect two vertices on different input (or output) paths for the \textit{same} encoded qubit since the subgraphs we use to implement single-qubit gates have no paths of length $\leq C$ between two input or two output vertices---see the beginning of \sec{build}.} For times $0\leq t\leq t_{\mathrm{I}}$, we have
\begin{align}
|\psi(t)\rangle &= \Sym\left(\sum_{\vec{z}\in\{0,1\}^{n+1}} \alpha_{\vec{z}} \, e^{-ih_{1}t }|z_{\text{in}}^{1}\rangle^1 \ldots e^{-ih_{n}t }|z_{\text{in}}^{n}\rangle^n e^{-ih_{n+1}t}|z_{\text{in}}^{n+1}\rangle^{\med}\right)\nonumber \\
&=|\mu(t)\rangle+|\epsilon (t)\rangle\label{eq:mut}
\end{align}
where $\left\Vert \epsilon(t)\rangle\right\Vert=\O(nL^{-{1}/{4}})$ and $|\mu(t)\rangle$ only has support on states where each particle is located at least a distance $N_0=M(-{\pi}/{4})$ from the endpoints of the input/output paths of the block (note that $M(-{\pi}/{4}) < M(-{\pi}/{2})$). In particular,
\[
(1-P)\left(H_G^{(n+1)}\right)^r|\mu(t)\rangle=0 \text{ for all } 0\leq r< M\left(-{\pi}/{4}\right).
\]
Equation \eq{mut} again follows from the results of \sec{1qubit_calc} (using the results of that section to approximate $e^{-ih_w t}|z_{\text{in}}\rangle^w$ for each $w$ and $t\leq t_{\mathrm{I}}$).  Applying the truncation lemma we then obtain 
\begin{align*}
\left\Vert e^{-iH_{G}^{(n+1)}t_{\mathrm{I}}}|\psi(0)\rangle- \ket{\mu(t_{\mathrm{I}})}\right\Vert  &=\O\left(\left\Vert H_{G}^{(n+1)}\right\Vert nL^{-{1}/{4}}\right).
\end{align*}

\subsubsection{Blocks applying two-qubit gates}

Now consider a block of type II (as shown in \fig{CPintall}). Without loss of generality we assume that the computational qubit involved in the gate is the $n$th encoded qubit. We label the vertices on the input and output paths of the
two qubits involved in the gate as in \fig{Graph-used-to-1}. For the
$n-1$ computational qubits not involved in the gate we label the states on each of the $2(n-1)$ paths as $|x,q\rangle^j$ for $j\in\{1,\ldots,n-1\}$, $q\in\{0,1\}$, and $x\in\{1,\ldots,2W+Z+4\}$. Here $x=1$ is the leftmost vertex on each path. Define the input and output states for each of the $n-1$ computational qubits not involved in the gate as
\begin{align*}
|0_{\text{in}}\rangle^j & =  \frac{1}{\sqrt{L}}\sum_{x=M(-\frac{\pi}{4})+1}^{M(-\frac{\pi}{4})+L}e^{-i\frac{\pi}{4}x}|x,0\rangle^{j}&
|1_{\text{in}}\rangle^j & =  \frac{1}{\sqrt{L}}\sum_{x=M(-\frac{\pi}{4})+1}^{M(-\frac{\pi}{4})+L}e^{-i\frac{\pi}{4}x}|x,1\rangle^{j}\\
|0_{\text{out}}\rangle^j &= \frac{e^{-it_{\mathrm{II}}\sqrt{2}}}{\sqrt{L}}\sum_{x=Q_{1}+1}^{Q_{1}+L}e^{-i\frac{\pi}{4}x}|x,0\rangle^j &
|1_{\text{out}}\rangle^j &= \frac{e^{-it_{\mathrm{II}}\sqrt{2}}}{\sqrt{L}}\sum_{x=Q_{1}+1}^{Q_{1}+L}e^{-i\frac{\pi}{4}x}|x,1\rangle^j,
\end{align*}
where $Q_1$ is defined in section \sec{2qubit_calc}. Similarly, define the logical input and output states for the $n$th computational qubit and the mediator as in \sec{2qubit_calc}.
The state at time $t=0$ has the form
 \begin{align*}
|\psi(0)\rangle&=\Sym\left(\sum_{\vec{z}\in\{0,1\}^{n+1}} \alpha_{\vec{z}} |z_{\text{in}}^{1}\rangle^1\ldots|z_{\text{in}}^{n-1}\rangle^{n-1}|z^n_{\text{in}}\rangle^n |z^{n+1}_{\text{in}}\rangle^{\med}\right)\\
& =\frac{1}{\sqrt{2}} \Sym\left(\sum_{\vec{z}\in\{0,1\}^{n+1}} \alpha_{\vec{z}} |z_{\text{in}}^{1}\rangle^1\ldots|z_{\text{in}}^{n-1}\rangle^{n-1}|z^n z^{n+1}_{\text{in}}\rangle^{n,\med}\right)
\end{align*}
where in the last line we have written the state in terms of the symmetrized (or antisymmetrized) logical states $|00_{\text{in}}\rangle^{n,\med},\allowbreak |01_{\text{in}}\rangle^{n,\med},\allowbreak |10_{\text{in}}\rangle^{n,\med},\allowbreak |11_{\text{in}}\rangle^{n,\med}$ as defined in \sec{2qubit_calc}. The prefactor of $1/\sqrt{2}$ arises from the fact that these logical states for the two encoded qubits $n$ and $\med$ are already symmetrized (or antisymmetrized).

Evolving this state for time $t_{\mathrm{II}}$ using the Hamiltonian
of the block (without anything connected on either side) gives $|\psi(t_{\mathrm{II}})\rangle=e^{-iH_{\text{block}}t_{\mathrm{II}}}|\psi(0)\rangle$.  Let $h_1,h_2,\ldots,h_{n-1}$ be the single-particle Hamiltonians for the components of the graph corresponding to the $n-1$ encoded qubits that are not involved in the gate. Let $h_{n,\med}$ be the two-particle Hamiltonian for the two encoded qubits on which the gate acts. Then 
\begin{align*}
|\psi(t_{\mathrm{II}})\rangle &= \frac{1}{\sqrt{2}}\Sym\Bigg(\sum_{\vec{z}\in\{0,1\}^{n+1}}  \alpha_{\vec{z}} \, e^{-ih_1 t_{\mathrm{II}}}|z_{\text{in}}^{1}\rangle^1\ldots e^{-ih_{n-1} t_{\mathrm{II}}}|z_{\text{in}}^{n-1}\rangle^{n-1} e^{-ih_{n,\med}t_{\mathrm{II}}}|z^n z^{n+1}_{\text{in}}\rangle^{n,\med} \Bigg).
\end{align*}
Using the results of \sec{2qubit_calc} this is 
\begin{align*}
|\psi(t_{\mathrm{II}})\rangle &=\frac{1}{\sqrt{2}}\Sym\Bigg(\sum_{\vec{z}\in\{0,1\}^{n+1}} \alpha_{\vec{z}} \, e^{i\theta z^n z^{n+1}} |z_{\text{out}}^{1}\rangle^1\ldots |z_{\text{out}}^{n-1}\rangle^{n-1}|z^n z^{n+1}_{\text{out}}\rangle^{n,\med}\Bigg)+ \O(nL^{-{1}/{4}})\\
&= \ket{\kappa(t_{\mathrm{II}})} + \O(n L^{-{1}/{4}})
\end{align*}
where $\ket{\kappa(t_{\mathrm{II}})}$ is the encoded logical state with the unitary applied.

Now connecting this block to other blocks on either side, and letting $H_{G}^{(n+1)}$ be the full Hamiltonian, we can apply the truncation lemma just as we did in the previous section for blocks of type I.  This gives 
\begin{align*}
\left\Vert e^{-iH_{G}^{(n+1)}t_{\mathrm{II}}}|\psi(0)\rangle - \ket{\kappa(t_{\mathrm{II}})} \right\Vert &= \O\left(\left\Vert H_{G}^{(n+1)}\right\Vert nL^{-{1}/{4}}\right).
\end{align*}

\subsubsection{Piecing blocks together}

We are now ready to show how to apply two unitaries in series, by concatenating blocks of type I or II as shown in \fig{squint} and \fig{CPintall}, respectively. We concatenate two blocks $B$ and $B'$ by removing some of the vertices on the input paths of $B'$ and then connecting the output paths of $B$ to the resulting graph (so that the two blocks overlap on the removed vertices). We remove the leftmost $2 M(k) + L$ vertices from each input path of $B'$ (this is a different number of vertices depending on whether the paths are associated with a mediator qubit or a computational qubit) and then connect every output path of $B$ to the corresponding shortened input path in $B'$. Note that if $B'$ is of type I, we remove all vertices from each input path, attaching the output path of $B$ to inputs of the individual gate subgraphs of $B'$.

Assuming that our wave packet initially starts as the correct input for a single block (i.e., a symmetrized logical state with one particle on each pair of qubit paths), we want to show the wave packet first moves through $B$, undergoing the appropriate unitary, and then moves through $B'$, undergoing the second unitary, with small total error. Define $t_1$ and $t_2$ to be $t_{\mathrm{I}}$ or $t_{\mathrm{II}}$ depending on the block type of $B$ and $B'$, respectively, and let the initial state of the system be $\ket{\psi(0)}$, a symmetrized logical state as described previously. From the results of the previous two sections, the state of the system after time $t_1$ is $\ket{\text{out}_1}+ \O(\Vert H_G^{(n+1)}\Vert n L^{{-1}/{4}})$. Here $\ket{\text{out}_1}$ is a symmetrized logical state in which we have applied the unitaries of $B$ to the initial logical state. Furthermore, each individual wave packet is located a distance $M(k)$ from the end of $B$, and has length $L$. Since $B$ and $B'$ overlap on $2M(k) + L$ vertices, we see that each wave packet is located a distance $M(k)$ inside of $B'$, so the output state for $B$ corresponds with an input state for $B'$ up to an irrelevant global phase. After time $t_1 + t_2$, the wave function is then $\ket{\text{out}_2} + \O(\Vert H_G^{(n+1)}\Vert n L^{{-1}/{4}})$, where $\ket{\text{out}_2}$ is a logical state in which the unitaries of $B'$ and $B$ have been applied in sequence.  Now generalizing to the case of $g$ blocks in series,  we see that the total error in the final output state is $\O(g \Vert H_G^{(n+1)}\Vert n L^{-{1}/{4}})$, as claimed in equation \eq{errorbound}.
\section{Single-particle wave packet scattering on infinite graphs}\label{sec:Single-Particle-Wavetrain}

In this section we prove \thm{singlepart}. The proof is based on (and follows closely) the calculation from the appendix of reference \cite{FGG08}.

Recall from \eq{single_particle_states} that the scattering eigenstates of $H_{G}^{(1)}$ have the
form
\[
\langle x,q|\mathrm{sc}_{j}(k)\rangle=e^{-ikx}\delta_{qj}+e^{ikx}S_{qj}(k)\]
 for each $k\in(-\pi,0)$. 

Before delving into the proof, we first establish that the state $|\alpha^{j}(t)\rangle$ is approximately normalized. This state is not normalized at all times $t$. However, $\langle\alpha^{j}(t)|\alpha^{j}(t)\rangle=1+\O(L^{-1})$, as we now show:
\begin{align*}
\langle\alpha^{j}(t)|\alpha^{j}(t)\rangle & =\frac{1}{L} \sum_{x=1}^{\infty}\left|e^{-ikx}R(x-\left\lfloor 2t\sin k\right\rfloor)+S_{jj}(k)e^{ikx}R(-x-\left\lfloor 2t\sin k\right\rfloor)\right|^{2} \\
 &\quad   +\frac{1}{L}\sum_{q\neq j}\sum_{x=1}^{\infty}|S_{qj}(k)|^{2} R(-x-\left\lfloor 2t\sin k\right\rfloor) \\
 & = \frac{1}{L}\sum_{x=1}^{\infty}\left[R(x-\left\lfloor 2t\sin k\right\rfloor )+R(-x-\left\lfloor 2t\sin k\right\rfloor)\right] \\
 &  \quad +\frac{1}{L}\sum_{x=1}^{\infty}\left(e^{-2ikx}S_{jj}^{\ast}(k)+e^{2ikx}S_{jj}(k)\right)R(x-\left\lfloor 2t\sin k\right\rfloor)R(-x-\left\lfloor 2t\sin k\right\rfloor) \\
 & = 1\! +\! \frac{1}{L}\sum_{x=1}^{\infty} \!
 	\left(e^{-2ikx}S_{jj}^{\ast}(k)\! +\! e^{2ikx}S_{jj}(k)\right) \!
   R(x\! -\! \left\lfloor 2t\sin k\right\rfloor)
   R(-x\! -\!\left\lfloor 2t\sin k\right\rfloor)
   \! +\!\O(L^{-1})
\end{align*}
where we have used unitarity of $S$ in the second step.
When it is nonzero, the second term can be written as
\[
\frac{1}{L}\sum_{x=1}^{b}\left(e^{-2ikx}S_{jj}^{\ast}(k)+e^{2ikx}S_{jj}(k)\right)
\]
where $b$ is the maximum positive integer such that $\{-b,b\}\subset\{M+1+\left\lfloor 2t\sin k\right\rfloor,\ldots,M+L+\left\lfloor 2t\sin k\right\rfloor \}$. Performing
the sums, we get
\begin{align*}
\left|\frac{1}{L}\sum_{x=1}^{b}\left(e^{-2ikx}S_{jj}^{\ast}(k)+e^{2ikx}S_{jj}(k)\right)\right| 
&= \frac{1}{L} \left|S_{jj}^{\ast}(k)e^{-2ik} \frac{e^{-2ikb}-1}{e^{-2ik}-1} 
+
S_{jj}(k)e^{2ik} \frac{e^{2ikb}-1}{e^{2ik}-1} \right|\\
 & \leq \frac{2}{L|{\sin k}|} .
\end{align*}
Thus we have $\langle\alpha^{j}(t)|\alpha^{j}(t)\rangle=1+\O(L^{-1})$.

\begin{proof}[Proof of \thm{singlepart}]
Define
\[
|\psi^{j}(t)\rangle=e^{-iH_{G}^{(1)}t}|\psi^{j}(0)\rangle
\]
and
\[
  \Pi_\epsilon = \int_{-\epsilon}^{\epsilon} \frac{d\phi}{2\pi} \sum_{q=1}^{N}|\mathrm{sc}_{q}(k+\phi)\rangle\langle\mathrm{sc}_{q}(k+\phi)|
\]
where we take $\epsilon=\frac{\left|\sin k\right|}{2\sqrt{L}}$.
Observe that $\Pi_\epsilon$ is a projection (i.e., $\Pi_\epsilon^2 = \Pi_\epsilon$), as can be shown using the delta-function normalization of the scattering states.
Thus we can write
\[
|\psi^{j}(t)\rangle=|w^{j}(t)\rangle+|v^{j}(t)\rangle
\]
where
\begin{align*}
|w^{j}(t)\rangle
&= \Pi_\epsilon |\psi^j(t)\rangle \\
&=\int_{-\epsilon}^{\epsilon}\frac{d\phi}{2\pi}e^{-2it\cos\left(k+\phi\right)}\sum_{q=1}^{N}|\mathrm{sc}_{q}(k+\phi)\rangle\langle\mathrm{sc}_{q}(k+\phi)|\psi^{j}(0)\rangle
\end{align*}
and $\langle w^{j}(t)|v^{j}(t)\rangle=0$.
Now
\[
\langle\mathrm{sc}_{q}(k+\phi)|\psi^{j}(0)\rangle=\frac{1}{\sqrt{L}}\sum_{x=M+1}^{M+L}\left(e^{i\phi x}\delta_{qj}+e^{-i\left(2k+\phi\right)x}S_{qj}^{\ast}(k+\phi)\right),\]
 so \[
|w^{j}(t)\rangle=|w_{A}^{j}(t)\rangle+\sum_{q=1}^{N}|w_{B}^{q,j}(t)\rangle\]
 where \begin{align*}
|w_{A}^{j}(t)\rangle & = \int_{-\epsilon}^{\epsilon}\frac{d\phi}{2\pi}e^{-2it\cos\left(k+\phi\right)}f(\phi)|\mathrm{sc}_{j}(k+\phi)\rangle\\
|w_{B}^{q,j}(t)\rangle & = \int_{-\epsilon}^{\epsilon}\frac{d\phi}{2\pi}e^{-2it\cos\left(k+\phi\right)}g_{qj}(\phi)|\mathrm{sc}_{q}(k+\phi)\rangle
\end{align*}
with
\begin{align*}
f(\phi) & = \frac{1}{\sqrt{L}}\sum_{x=M+1}^{M+L}e^{i\phi x}\\
g_{qj}(\phi) & = \frac{1}{\sqrt{L}}\sum_{x=M+1}^{M+L}e^{-i\left(2k+\phi\right)x}S_{qj}^{\ast}(k+\phi).\end{align*}
We will see that $\ket{\psi^j(t)} \approx \ket{w^j(t)} \approx \ket{w^j_A(t)} \approx \ket{\alpha^j(t)}$.

We have
\begin{align*}
\langle w_{A}^{j}(t)|w_{A}^{j}(t)\rangle & = \int_{-\epsilon}^{\epsilon}\frac{d\phi}{2\pi}\left|f(\phi)\right|^{2}
 = \frac{1}{L} \int_{-\epsilon}^{\epsilon}\frac{d\phi}{2\pi}\frac{\sin^{2}(\frac{1}{2}L\phi)}{\sin^{2}(\frac{1}{2}\phi)},
 \end{align*}
but
\[
\frac{1}{L} \int_{-\pi}^{\pi}\frac{d\phi}{2\pi}\frac{\sin^{2}(\frac{1}{2}L\phi)}{\sin^{2}(\frac{1}{2}\phi)}=1
\]
and
\begin{align}
\frac{1}{L}\left(\int_{\epsilon}^{\pi}+\int_{-\pi}^{-\epsilon}\right)\frac{d\phi}{2\pi}\frac{\sin^{2}(\frac{1}{2}L\phi)}{\sin^{2}(\frac{1}{2}\phi)} & = \frac{2}{L} \int_{\epsilon}^{\pi}\frac{d\phi}{2\pi}\frac{\sin^{2}(\frac{1}{2}L\phi)}{\sin^{2}(\frac{1}{2}\phi)}\nonumber \\
 & \leq \frac{2}{L}\int_{\epsilon}^{\pi}\frac{d\phi}{2\pi}\frac{\pi^{2}}{\phi^{2}}\nonumber \\
 & \leq \frac{\pi}{L\epsilon}.\label{eq:fbound}
\end{align}
Therefore
\[
1\geq\langle w_{A}^{j}(t)|w_{A}^{j}(t)\rangle\geq1-\frac{\pi}{L\epsilon}.
\]
Similarly,
\[
\langle w_{B}^{qj}(t)|w_{B}^{qj}(t)\rangle=\int_{-\epsilon}^{\epsilon}\frac{d\phi}{2\pi}\frac{\left|S_{qj}(k+\phi)\right|^{2}}{L}\frac{\sin^{2}(\frac{1}{2}L(2k+\phi))}{\sin^{2}(\frac{1}{2}(2k+\phi))},\]
and, using the unitarity of $S$,
\begin{align*}
\sum_{q=1}^{N}\langle w_{B}^{qj}(t)|w_{B}^{qj}(t)\rangle & = \frac{1}{L}\int_{-\epsilon}^{\epsilon}\frac{d\phi}{2\pi}\frac{\sin^{2}(\frac{1}{2}L(2k+\phi))}{\sin^{2}(\frac{1}{2}(2k+\phi))}\\
 & \leq \frac{1}{L} \int_{-\epsilon}^{\epsilon}\frac{d\phi}{2\pi}\frac{1}{\sin^{2}(\frac{1}{2}(2k+\phi))}.
\end{align*}
Now $|{\sin(k+{\phi}/{2}) - \sin k}| \leq {|\phi|}/{2}$ (by the mean value theorem), so
\begin{align*}
\sin^{2}\left(k+\frac{\phi}{2}\right) & \geq \left(\left|\sin k\right|-\left|\frac{\phi}{2}\right|\right)^{2}.
\end{align*}
Since $\epsilon=\frac{\left|\sin k\right|}{2\sqrt{L}}<\left|\sin k\right|$
we then have \begin{align*}
\sum_{q=1}^{N}\langle w_{B}^{qj}(t)|w_{B}^{qj}(t)\rangle & \leq \frac{1}{L} \int_{-\epsilon}^{\epsilon}\frac{d\phi}{2\pi}\frac{4}{\sin^{2}k}\\
 & = \frac{4\epsilon}{\pi L\sin^{2}k}.\end{align*}
 Hence \begin{align*}
\langle w^{j}(t)|w^{j}(t)\rangle & \geq \langle w_{A}^{j}(t)|w_{A}^{j}(t)\rangle-2\left|\sum_{q=1}^{N}\langle w_{A}^{j}(t)|w_{B}^{qj}(t)\rangle\right|\\
 & \geq 1-\frac{\pi}{L\epsilon}-2\left\Vert \sum_{q=1}^{n}|w_{B}^{qj}(t)\rangle\right\Vert \\
 & \geq 1-\frac{\pi}{L\epsilon}-2\sum_{q=1}^{n}\left\Vert |w_{B}^{qj}(t)\rangle\right\Vert \\
 & \geq 1-\frac{\pi}{L\epsilon}-4\sqrt{\frac{\epsilon N}{\pi L\sin^{2}k}},\end{align*}
so
\[
\langle v^{j}(t)|v^{j}(t)\rangle\leq\frac{\pi}{L\epsilon}+4\sqrt{\frac{\epsilon N}{\pi L\sin^{2}k}}
\]
since $\langle v^{j}(t)|v^{j}(t)\rangle+\langle w^{j}(t)|w^{j}(t)\rangle=1$.
Thus
\begin{align*}
\left\Vert |\psi^{j}(t)\rangle-|w_{A}^{j}(t)\rangle\right\Vert  & = \left\Vert |v^{j}(t)\rangle+\sum_{q=1}^{N}|w_{B}^{qj}(t)\rangle\right\Vert \\
 & \leq \left(\frac{\pi}{L\epsilon}+4\sqrt{\frac{\epsilon N}{\pi L\sin^{2}k}}\right)^{\frac{1}{2}}+2\sqrt{\frac{\epsilon N}{\pi L\sin^{2}k}}.
\end{align*}
With our choice $\epsilon=\frac{\left|\sin k\right|}{2\sqrt{L}}$,
we have $\norm{|\psi^{j}(t)\rangle-|w_{A}^{j}(t)\rangle} =\O(L^{-{1}/
{4}})$.
We now show that
\begin{equation}
\left\Vert |w_{A}^{j}(t)\rangle-|\alpha^{j}(t)\rangle\right\Vert =\O(L^{-{1}/{4}}).\label{eq:alpha_w_bound}\end{equation}
Letting \[
P=\sum_{q=1}^{N}\sum_{x=1}^{\infty}|x,q\rangle\langle x,q|\]
be the projector onto the semi-infinite paths, to show equation \eq{alpha_w_bound}
we use the bounds
\begin{equation}
\left\Vert \left(1-P\right) |w_{A}^{j}(t)\rangle\right\Vert=\O\left(\frac{\log L}{\sqrt{L}}\right)\label{eq:boundinsidegraph}
\end{equation}
and
\begin{equation}
\left\Vert P|w_{A}^{j}(t)\rangle-|\alpha^{j}(t)\rangle\right\Vert = \O(L^{-{1}/{4}}).
\label{eq:bound_inside_lines}
\end{equation}
Equation \eq{alpha_w_bound} follows from these bounds since $(1-P)|\alpha^{j}(t)\rangle=0$.  

To get equation \eq{boundinsidegraph}, write
\[
\langle w_{A}^{j}(t)|1-P|w_{A}^{j}(t)\rangle=\int_{D_\epsilon}\frac{d\phi d\tilde{\phi}}{4\pi^2}e^{-2it\cos\left(k+\phi\right)+2it\cos(k+\tilde{\phi})}f(\phi)f^{\ast}(\tilde{\phi})\langle\mathrm{sc}_{j}(k+\tilde{\phi})|1-P|\mathrm{sc}_{j}(k+\phi)\rangle.\]
 Using \lem{defn_scatteringstates}, there is a constant
$\lambda$ such that $|\langle\mathrm{sc}_{j}(k+\tilde{\phi})|1-P|\mathrm{sc}_{j}(k+\phi)\rangle|<\lambda^{2}m$,
so 
\begin{align}
\left|\langle w_{A}^{j}(t)|1-P|w_{A}^{j}(t)\rangle\right| & \leq \lambda^{2}m\int_{-\epsilon}^{\epsilon}\frac{d\phi}{2\pi}\int_{-\epsilon}^{\epsilon}\frac{d\tilde{\phi}}{2\pi} |f(\phi)f^{\ast}(\tilde{\phi})| \nonumber \\
 & = \lambda^{2}m\left(\int_{-\epsilon}^{\epsilon}\frac{d\phi}{2\pi} |f(\phi)|\right)^{2}.\label{eq:proj_eqn}
\end{align}
 Now 
\begin{align}
\int_{-\epsilon}^{\epsilon}\frac{d\phi}{2\pi}\left|f(\phi)\right| & = \int_{-\epsilon}^{\epsilon}\frac{d\phi}{2\pi\sqrt{L}}\left|\frac{\sin\frac{\phi L}{2}}{\sin\frac{\phi}{2}}\right|\nonumber \\
 & = 2\int_{0}^{d}\frac{d\phi}{2\pi\sqrt{L}}\left|\frac{\sin\frac{\phi L}{2}}{\sin\frac{\phi}{2}}\right|+2\int_{d}^{\epsilon}\frac{d\phi}{2\pi\sqrt{L}}\left|\frac{\sin\frac{\phi L}{2}}{\sin\frac{\phi}{2}}\right|\text{ for any }d\in(0,\epsilon]\nonumber \\
 & \leq \frac{d\sqrt{L}}{\pi}+2\int_{d}^{\epsilon}\frac{d\phi}{2\pi\sqrt{L}}\frac{\pi}{\phi}\nonumber \\
 & \leq \frac{d\sqrt{L}}{\pi}+\frac{\log\frac{1}{d}}{\sqrt{L}}\label{eq:eqn_d}
\end{align}
 where in the last line we used the fact that $\epsilon<1.$ Setting
$d=\Theta({1}/{L})$, using equation \eq{proj_eqn}, and taking the square root of both sides, we get equation \eq{boundinsidegraph}. 

We now prove the bound \eq{bound_inside_lines}. Noting that
\[
\frac{1}{\sqrt{L}}R(l)=\int_{-\pi}^{\pi}\frac{d\phi}{2\pi}e^{-i\phi l}f(\phi),
\]
we write
\begin{align}
\langle x,q|\alpha^{j}(t)\rangle & = e^{-2it\cos k}\left(\delta_{qj}e^{-ikx}\int_{-\pi}^{\pi}\frac{d\phi}{2\pi}e^{-i\phi\left(x-\left\lfloor 2t\sin k\right\rfloor \right)}f(\phi)\right.\nonumber\\
&\qquad\qquad\qquad  \left. +S_{qj}(k)e^{ikx}\int_{-\pi}^{\pi}\frac{d\phi}{2\pi}e^{-i\phi\left(-x-\left\lfloor 2t\sin k\right\rfloor \right)}f(\phi)\right).\label{eq:alpha_mat_elements}
\end{align}
On the other hand,
\begin{equation}
\langle x,q|w_{A}^{j}(t)\rangle=\int_{-\epsilon}^{\epsilon}\frac{d\phi}{2\pi}e^{-2it\cos\left(k+\phi\right)}f(\phi)\left(e^{-i\left(k+\phi\right)x}\delta_{qj}+e^{i\left(k+\phi\right)x}S_{qj}(k+\phi)\right).\label{eq:w_a_mat_elements}
\end{equation}
Using equations \eq{alpha_mat_elements} and \eq{w_a_mat_elements}
we can write\[
P|w_{A}^{j}(t)\rangle=|\alpha^{j}(t)\rangle+\sum_{i=1}^{7}|c_{i}^{j}(t)\rangle\]
 where $P|c_{i}^{j}(t)\rangle=|c_{i}^{j}(t)\rangle$ and \begin{align*}
\langle x,q|c_{1}^{j}(t)\rangle & = \delta_{qj}e^{-2it\cos k}e^{-ikx}\int_{-\pi}^{\pi}\frac{d\phi}{2\pi}e^{-i\phi x}f(\phi)\left(e^{2it\phi\sin k}-e^{i\phi\left\lfloor 2t\sin k\right\rfloor }\right)\\
\langle x,q|c_{2}^{j}(t)\rangle & = S_{qj}(k)e^{-2it\cos k}e^{ikx}\int_{-\pi}^{\pi}\frac{d\phi}{2\pi}e^{i\phi x}f(\phi)\left(e^{2it\phi\sin k}-e^{i\phi\left\lfloor 2t\sin k\right\rfloor }\right)\\
\langle x,q|c_{3}^{j}(t)\rangle & = -\delta_{qj}e^{-2it\cos k}e^{-ikx}\left(\int_{\epsilon}^{\pi}+\int_{-\pi}^{-\epsilon}\right)\frac{d\phi}{2\pi}e^{-i\phi x}f(\phi)e^{2it\phi\sin k}\\
\langle x,q|c_{4}^{j}(t)\rangle & = -S_{qj}(k)e^{-2it\cos k}e^{ikx}\left(\int_{\epsilon}^{\pi}+\int_{-\pi}^{-\epsilon}\right)\frac{d\phi}{2\pi}e^{i\phi x}f(\phi)e^{2it\phi\sin k}\\
\langle x,q|c_{5}^{j}(t)\rangle & = \delta_{qj}e^{-ikx}\int_{-\epsilon}^{\epsilon}\frac{d\phi}{2\pi}e^{-i\phi x}f(\phi)\left(e^{-2it\cos\left(k+\phi\right)}-e^{-2it\cos k+2it\phi\sin k}\right)\\
\langle x,q|c_{6}^{j}(t)\rangle & = S_{qj}(k)e^{ikx}\int_{-\epsilon}^{\epsilon}\frac{d\phi}{2\pi}e^{i\phi x}f(\phi)\left(e^{-2it\cos\left(k+\phi\right)}-e^{-2it\cos k+2it\phi\sin k}\right)\\
\langle x,q|c_{7}^{j}(t)\rangle & = e^{ikx}\int_{-\epsilon}^{\epsilon}\frac{d\phi}{2\pi}e^{i\phi x}e^{-2it\cos\left(k+\phi\right)}f(\phi)\left(S_{qj}(k+\phi)-S_{qj}(k)\right).
\end{align*}
We now bound the norm of each of these states:
\begin{align*}
\langle c_{1}^{j}(t)|c_{1}^{j}(t)\rangle & = \sum_{q=1}^{N}\sum_{x=1}^{\infty}\left|\delta_{qj}e^{-2it\cos k}e^{-ikx}\int_{-\pi}^{\pi}\frac{d\phi}{2\pi}e^{-i\phi x}f(\phi)\left(e^{2it\phi\sin k}-e^{i\phi\left\lfloor 2t\sin k\right\rfloor }\right)\right|^{2}\\
 & \leq \sum_{q=1}^{N}\sum_{x=-\infty}^{\infty}\left|\delta_{qj}e^{-2it\cos k}e^{-ikx}\int_{-\pi}^{\pi}\frac{d\phi}{2\pi}e^{-i\phi x}f(\phi)\left(e^{2it\phi\sin k}-e^{i\phi\left\lfloor 2t\sin k\right\rfloor }\right)\right|^{2}\\
 & = \int_{-\pi}^{\pi}\frac{d\phi}{2\pi}\left|f(\phi)\right|^{2}\left|e^{2it\phi\sin k}-e^{i\phi\left\lfloor 2t\sin k\right\rfloor }\right|^{2}\\
 & \leq \int_{-\pi}^{\pi}\frac{d\phi}{2\pi}\left|f(\phi)\right|^{2}\left(2t\phi\sin k-\left\lfloor 2t\sin k\right\rfloor \phi\right)^{2}\\
 & \leq \int_{-\pi}^{\pi}\frac{d\phi}{2\pi}\left|f(\phi)\right|^{2}\phi^{2}
\end{align*}
where we have used the facts that $|{e^{is}-1}|^{2} \leq s^{2}$
for $s\in\mathbb{R}$ and $\left|2t\sin k -\left\lfloor 2t\sin k\right\rfloor \right|<1$. In the above we made the following replacement under the integral:
\[
\sum_{x=-\infty}^{\infty} e^{i(\phi-\tilde{\phi})x}=2\pi\delta(\phi-\tilde{\phi}) \text{ for } \phi,\tilde{\phi}\in (-\pi,\pi).
\]
We use this repeatedly in the following calculations.
Continuing, we get
\begin{align*}
\langle c_{1}^{j}(t)|c_{1}^{j}(t)\rangle & \leq \frac{1}{L} \int_{-\pi}^{\pi}\frac{d\phi}{2\pi}\frac{\sin^{2}(\frac{1}{2}L\phi)}{\sin^{2}(\frac{1}{2}\phi)}\phi^{2}\\
 & \le \frac{1}{L} \int_{-\pi}^{\pi}\frac{d\phi}{2\pi}\frac{1}{\sin^{2}(\frac{1}{2}\phi)}\phi^{2}\\
 & \leq \frac{\pi^{2}}{L}
\end{align*}
using the fact that $\sin^{2}({\phi}/{2})\geq{\phi^{2}}/{\pi^{2}}$
for $\phi\in[-\pi,\pi]$. Similarly we bound $\langle c_{2}^{j}(t)|c_{2}^{j}(t)\rangle\leq{\pi^{2}}/{L}$. 

Using equation \eq{fbound} we get 
\begin{align*}
\langle c_{3}^{j}(t)|c_{3}^{j}(t)\rangle & \leq \left(\int_{\epsilon}^{\pi}+\int_{-\pi}^{-\epsilon}\right)\frac{d\phi}{2\pi}\left|f(\phi)\right|^{2}\\
 & \leq \frac{\pi}{L\epsilon}\end{align*}
and similarly for $\langle c_{4}^{j}(t)|c_{4}^{j}(t)\rangle$.
Next, we have
\begin{align*}
\langle c_{5}^{j}(t)|c_{5}^{j}(t)\rangle & \leq \int_{-\epsilon}^{\epsilon}\frac{d\phi}{2\pi}\left|f(\phi)\right|^{2}\left|e^{-2it\cos\left(k+\phi\right)}-e^{-2it\cos k+2it\phi\sin k}\right|^{2}\\
 & \leq \int_{-\epsilon}^{\epsilon}\frac{d\phi}{2\pi}\left|f(\phi)\right|^{2}\left(2t\cos\left(k+\phi\right)-2t\cos k+2t\phi\sin k\right)^{2}\\
 & = \int_{-\epsilon}^{\epsilon}\frac{d\phi}{2\pi}\left|f(\phi)\right|^{2}\left(2t\cos k\left(\cos\phi-1\right)+2t\sin k\left(\phi-\sin\phi\right)\right)^{2}\\
 & \leq \int_{-\epsilon}^{\epsilon}\frac{d\phi}{2\pi}\left|f(\phi)\right|^{2}4t^{2}\phi^{4}\\
 & = \frac{4t^{2}}{L}\int_{-\epsilon}^{\epsilon}\frac{d\phi}{2\pi}\frac{\sin^{2}(\frac{1}{2}L\phi)}{\sin^{2}(\frac{1}{2}\phi)}\phi^{4}\\
 & \leq \frac{4t^{2}}{L}\int_{-\epsilon}^{\epsilon}\frac{d\phi}{2\pi}\pi^{2}\phi^{2}\\
 & = \frac{4\pi}{3L}t^{2}\epsilon^{3}\end{align*}
 and we have the same bound for $|c_{6}^{j}(t)\rangle$. Finally, 
\begin{align*}
\langle c_{7}^{j}(t)|c_{7}^{j}(t)\rangle & \leq \int_{-\epsilon}^{\epsilon}\frac{d\phi}{2\pi}\left|f(\phi)\right|^{2}\sum_{q=1}^{N}\left|S_{qj}(k+\phi)-S_{qj}(k)\right|^{2}.\end{align*}
Now, for each $q\in\{1,\ldots,N\}$,
\[
\left|S_{qj}(k+\phi)-S_{qj}(k)\right| \leq \Gamma |\phi|
\]
where the Lipschitz constant
\[
\Gamma = \max_{q,j\in\{1,\ldots,N\}} \max_{k' \in [-\pi,\pi]}\left|\frac{d}{dk'}S_{qj}(k')\right|
\]
is well defined since each matrix element $S_{qj}(k')$ is a
bounded rational function of $e^{ik'}$ (see Section 3 of \cite{Childs_Gosset}). Hence
\begin{align*}
\langle c_{7}^{j}(t)|c_{7}^{j}(t)\rangle & \leq \int_{-\epsilon}^{\epsilon}\frac{d\phi}{2\pi}\left|f(\phi)\right|^{2}N\Gamma^{2}\phi^{2}\\
 & = \frac{N\Gamma^{2}}{L}\int_{-\epsilon}^{\epsilon}\frac{d\phi}{2\pi}\frac{\sin^{2}(\frac{1}{2}L\phi)}{\sin^{2}(\frac{1}{2}\phi)}\phi^{2}\\
 & \leq \frac{N\Gamma^{2}}{L}\int_{-\epsilon}^{\epsilon}\frac{d\phi}{2\pi}\pi^{2}\\
 & = N\Gamma^{2}\frac{\pi\epsilon}{L}.
\end{align*}
Now using the bounds on the norms of each of these states we get
\begin{align*}
\left\Vert P|w_{A}^{j}(t)\rangle-|\alpha^{j}(t)\rangle\right\Vert  & \leq 2\frac{\pi}{\sqrt{L}}+2\sqrt{\frac{\pi}{L\epsilon}}+2\sqrt{\frac{4\pi}{3L}t^{2}\epsilon^{3}}+\sqrt{N\Gamma^{2}\frac{\pi\epsilon}{L}}\\
 & = \O(L^{-{1}/{4}})
\end{align*}
using the choice $\epsilon=\frac{|{\sin p}|}{2\sqrt{L}}$ and the fact that $t = \O(L)$. 
\end{proof}


\section{Two-particle wave packet scattering on an infinite path}\label{sec:Two-Particle-Wavetrain}

In this section we prove \thm{twopart}. The main proof appears in \sec{twopartproof}, relying on several technical lemmas proved in \sec{techlem}. The proof follows the method used in the single-particle case, which is based on the calculation from the appendix of reference \cite{FGG08}.

Recall from \eq{symscatter} that for each $p_{1}\in(-\pi,\pi)$ and $p_{2}\in(0,\pi)$ there is an eigenstate $|\mathrm{sc}(p_{1};p_{2})\rangle_{\pm}$ of $H^{(2)}$ of the form 
\begin{align}
  \langle x,y|\mathrm{sc}(p_{1};p_{2})\rangle_\pm &=
    \frac{e^{-ip_{1}\left(\frac{x+y}{2}\right)}}{\sqrt{2}} 
    \begin{cases} e^{-i p_2 (x-y)} \pm  e^{i \theta_{\pm}(p_1,p_2)} e^{i p_2 (x - y)} & \text{if } x - y \leq -C\\
     e^{- i p_2 (x-y)}e^{i\theta_{\pm}(p_1,p_2)} \pm e^{i p_2 (x-y)} & \text{if }x-y \geq C\\
    f(p_1,p_2,x-y) \pm f(p_1,p_2,y-x)& \text{if }|x-y|< C \end{cases}\label{eq:sc}
\end{align}
where
\begin{align*}
e^{i \theta_\pm (p_{1},p_{2})} & =  T(p_1,p_2) \pm R(p_1,p_2),
\end{align*}
$C$ is the range of the interaction, $T$ and $R$ are the transmission
and reflection coefficients of the interaction at the chosen momentum, $f$
describes the amplitudes of the scattering state within the interaction range, and the $\pm$ depends
on the type of particle ($+$ for bosons, $-$ for fermions).  The state
$|\mathrm{sc}(p_{1};p_{2})\rangle_\pm$ satisfies 
\[
H^{(2)}|\mathrm{sc}(p_{1};p_{2})\rangle_{\pm} = 4\cos \f{p_{1}}{2}\cos p_{2} |\mathrm{sc}(p_{1};p_{2})\rangle_\pm
\]
and is delta-function normalized as
\begin{equation}
_{\pm}\langle\mathrm{sc}(p_{1}';p_{2}')|\mathrm{sc}(p_{1};p_{2})\rangle_{\pm}=
  4 \pi^{2} \delta(p_{1} - p_{1}') \delta(p_{2} - p_{2}').
\label{eq:del_func_norm}
\end{equation}

\subsection{Proof of \thm{twopart}}
\label{sec:twopartproof}

Let
\[
  \Pi_\epsilon = \iint_{D_{\epsilon}} \frac{d\phi_{1}d\phi_{2}} 
    {4\pi^{2}} |\mathrm{sc}(p_{1}+\phi_{1};p_{2}+\phi_{2})\rangle_{\pm} {}_{\pm}\langle\mathrm{sc}(p_{1}+\phi_{1};p_{2}+\phi_{2})|
\]
with $D_{\epsilon}=\left[-\epsilon,\epsilon\right]\times\left[-\epsilon,\epsilon\right]$, $p_{1}= {\pi}/{2}-{\pi}/{4}={\pi}/{4}$, and $p_{2}=({\pi}/{2} +
{\pi}/{4})/2={3\pi}/{8}$.
By the delta-function normalization of the scattering states (equation
\eq{del_func_norm}), $\Pi_\epsilon$ is a projection.
Thus we can write
\begin{align*}
|\psi(t)\rangle & = e^{-iH^{(2)}t}|\psi(0)\rangle = |\psi_{1}(t)\rangle+|\psi_{2}(t)\rangle
\end{align*}
where 
\begin{align*}
  \ket{\psi_{1}(t)}
  &= \Pi_\epsilon |\psi(t)\rangle \\
  &= \iint_{D_{\epsilon}} \frac{d\phi_{1}d\phi_{2}} 
    {4\pi^{2}} e^{-it 4\cos(\frac{p_{1}}{2}+\frac{\phi_{1}}{2})
      \cos(p_{2} + \phi_{2})}|\mathrm{sc}(p_{1}+\phi_{1};p_{2}+ 
    \phi_{2})\rangle_{\pm} \\
    &\qquad {}_{\pm}\langle\mathrm{sc}(p_{1}+\phi_{1};p_{2}+\phi_{2})|\psi(0)\rangle
\end{align*}
and $|\psi_{2}(t)\rangle$ is orthogonal to $|\psi_{1}(t)\rangle$.
We take $\epsilon=a/\sqrt{L}$ for some constant $a$. Using equation \eq{sc}
we get 
\[
  \ket{\psi_{1}(t)}=|\psi_{A}(t)\rangle\pm|\psi_{B}(t)\rangle
\]
where
\begin{align}
  |\psi_{A}(t)\rangle & = \iint_{D_{\epsilon}} \frac{d\phi_{1}d\phi_{2}}
    {4\pi^{2}} e^{-it 4\cos(\frac{\pi}{8}+\frac{\phi_{1}}{2})
    \cos(\frac{3\pi}{8}+\phi_{2})}
     A(\phi_{1},\phi_{2})|\mathrm{sc}(\tfrac{\pi}{4}+\phi_{1};
     \tfrac{3\pi}{8}+\phi_{2})\rangle_{\pm} \label{eq:psiA} \\
 \ket{\psi_{B}(t)} & = \iint_{D_{\epsilon}} \frac{d\phi_{1}d\phi_{2}}
    {4\pi^{2}} e^{-it 4\cos(\frac{\pi}{8}+\frac{\phi_{1}}{2})
    \cos(\frac{3\pi}{8}+\phi_{2})}  e^{-i\theta_{\pm}(\tfrac{\pi}{4}+\phi_{1},\tfrac{3\pi}{8}+\phi_{2} )} \nonumber\\
 	&\qquad B(\phi_{1},\phi_{2},\tfrac{3\pi}{8})
 	|\mathrm{sc}(\tfrac{\pi}{4} + \phi_{1};\tfrac{3\pi}{8} + \phi_{2})\rangle_\pm
\end{align}
with
\begin{align}
A(\phi_{1},\phi_{2}) & = \frac{1}{L}\sum_{x=-(M+L)}^{-(M+1)}\sum_{y=M+1}^{M+L} 
    e^{i\phi_{1}\frac{x+y}{2}}e^{i\phi_{2}\left(x-y\right)}\label{eq:A}\\
B(\phi_{1},\phi_{2},k) & = \frac{1}{L}\sum_{x=-(M+L)}^{-(M+1)}
    \sum_{y=M+1}^{M+L}e^{i\phi_{1}\frac{x+y}{2}} 	
    e^{i\left(\phi_{2}+2 k\right)\left(y-x\right)}.\nonumber 
\end{align}
Again using the delta-function normalization of the scattering states, we get 
\begin{align*}
\langle\psi_{B}(t)|\psi_{B}(t)\rangle & = \iint_{D_{\epsilon}}\frac{d\phi_{1}d\phi_{2}}{4\pi^{2}}\left|B(\phi_{1},\phi_{2},\tfrac{3\pi}{8})\right|^{2}\\
 & \leq \frac{16\pi^{2}}{L^{2}\epsilon^{2}}
\end{align*}
by \lem{Let--and} (provided $\epsilon<{3\pi}/{8}$, which holds for $L$ sufficiently large).
Similarly, 
\begin{align*}
1 &\geq\langle\psi_{A}(t)|\psi_{A}(t)\rangle \\
&= \iint_{D_{\epsilon}}\frac{d\phi_{1}d\phi_{2}}{4\pi^{2}}\left|A(\phi_{1},\phi_{2})\right|^{2}\\
 & \geq 1-\frac{4\pi}{L\epsilon}
\end{align*}
(from the first two facts in \lem{Let--and}) and therefore
\begin{align*}
\langle\psi_{1}(t)|\psi_{1}(t)\rangle & = \langle\psi_{A}(t)|\psi_{A}(t)\rangle+\langle\psi_{B}(t)|\psi_{B}(t)\rangle+\langle\psi_{A}(t)|\psi_{B}(t)\rangle+\langle\psi_{B}(t)|\psi_{A}(t)\rangle\\
 & \geq 1-\frac{4\pi}{L\epsilon}-2\left|\langle\psi_{A}(t)|\psi_{B}(t)\rangle\right|\\
 & \geq 1-\frac{4\pi}{L\epsilon}-2\left|\langle\psi_{A}(t)|\psi_{A}(t)\rangle\right|^{\frac{1}{2}}\left|\langle\psi_{B}(t)|\psi_{B}(t)\rangle\right|^{\frac{1}{2}}\\
 & \geq 
 1-\frac{12\pi}{L\epsilon}.
\end{align*}

Hence 
\[
  \langle\psi_{2}(t)|\psi_{2}(t)\rangle\leq\frac{12\pi}{L\epsilon}
\]
since 
\[
\langle\psi(t)|\psi(t)\rangle=\langle\psi_{1}(t)|\psi_{1}(t)\rangle+\langle\psi_{2}(t)|\psi_{2}(t)\rangle=1.
\]
Thus
\begin{align*}
\left\Vert \,|\psi(t)\rangle-|\psi_{A}(t)\rangle\right\Vert  & = \left\Vert \,|\psi_{B}(t)\rangle+|\psi_{2}(t)\rangle\right\Vert \\
 & \leq \left\Vert \,|\psi_{B}(t)\rangle\right\Vert +\left\Vert \,|\psi_{2}(t)\rangle\right\Vert \\
 & \leq \frac{4\pi}{L\epsilon}+\sqrt{\frac{12\pi}{L\epsilon}}.
\end{align*}
Now 
\begin{align*}
\left\Vert \,|\psi(t)\rangle-|\alpha(t)\rangle\right\Vert  & \leq \left\Vert \,|\psi(t)\rangle-|\psi_{A}(t)\rangle\right\Vert +\left\Vert \,|\psi_{A}(t)\rangle-|\alpha(t)\rangle\right\Vert \\
 & \leq \frac{4\pi}{L\epsilon}+\sqrt{\frac{12\pi}{L\epsilon}}+\left\Vert \,|\psi_{A}(t)\rangle-|\alpha(t)\rangle\right\Vert \\
 & = \O(L^{-{1}/{4}})+\left\Vert \,|\psi_{A}(t)\rangle-|\alpha(t)\rangle\right\Vert 
\end{align*}
using our choice $\epsilon=a/\sqrt{L}$. To complete the
proof, we now show that the second term in this expression is bounded by $\O(L^{-{1}/{4}})$.

\begin{lemma}
With $|\psi_{A}(t)\rangle$ and $|\alpha(t)\rangle$ defined through
equations \eq{psiA} and \eq{alpha}, respectively, and with $t\leq c_{0}L$ (for some constant $c_{0}$), 
$\norm{ |\psi_{A}(t)\rangle-|\alpha(t)\rangle } = \O(L^{-{1}/{4}})$.
\end{lemma}
\begin{proof}
To simplify matters, note that for $x\neq y$, $\langle x,y|\alpha(t)\rangle=\pm\langle y,x|\alpha(t)\rangle$ and $\langle x,y|\psi_{A}(t)\rangle=\pm\langle y,x|\psi_{A}(t)\rangle$.  Taking $C$ to be the maximum range of the interaction in our Hamiltonian, we have
\[
\left\Vert \,|\psi_{A}(t)\rangle-|\alpha(t)\rangle\right\Vert \leq2\left\Vert P_{1}|\psi_{A}(t)\rangle-P_{1}|\alpha(t)\rangle\right\Vert  + \left\Vert P_2 \ket{\psi_A(t)}\right\Vert + \left \Vert P_2 \ket{\alpha(t)}\right\Vert,
\]
where
\begin{equation}
     P_{1}=\sum_{y-x \geq C}|x,y\rangle\langle x,y| 
     \qquad P_2 = \sum_{|x-y| < C} \ket{x,y}\bra{x,y}.
\label{eq:P1P2}
\end{equation}
Now, for $y-x\geq C$, 
\begin{align*}
\langle x,y|\psi_{A}(t)\rangle & = \iint_{D_{\epsilon}}
  \frac{d\phi_{1}d\phi_{2}}{4\pi^{2}}
  e^{-it 4\cos(\frac{\pi}{8}+\frac{\phi_{1}}{2})
        \cos(\frac{3\pi}{8}+\phi_{2})} A(\phi_{1},\phi_{2})\frac{e^{-i\left(\frac{\pi}{4}+\phi_{1}\right)
       \left(\frac{x+y}{2}\right)}}{\sqrt{2}} \\
 & \qquad \left(e^{i\left(\frac{3\pi}
       {8}+\phi_{2}\right)\left(y-x\right)}\pm e^{-i\left(\frac{3\pi}{8}
       +\phi_{2}\right)
     \left(y-x\right)+
     i\theta_{\pm}(\frac{\pi}{4}+\phi_{1},\frac{3\pi}{8}+\phi_{2})}
      \right)\\
 & = \iint_{D_{\epsilon}}\frac{d\phi_{1}d\phi_{2}}{4\pi^{2}}
      \bigg[\f{1}{\sqrt{2}} e^{-it 4\cos(\frac{\pi}{8} + \frac{\phi_{1}}{2}) 
      \cos(\frac{3\pi}{8}+\phi_{2})} A(\phi_{1},\phi_{2}) \\
&    \qquad \left( e^{-i\pi x/2} e^{i\pi y/4} e^{-i\phi_{1}\left(\frac{x+y}{2}\right)}
 	 e^{i\phi_{2}\left(y-x\right)}\right.\\
& \qquad\quad\left.\pm  e^{i \pi x/4} e^{-i\pi y/2} e^{-i\phi_{1}
      \left(\frac{x+y}{2}\right)}
 	 e^{-i\phi_{2}\left(y-x\right)} e^{i \theta_{\pm}\left(\frac{\pi}{4} 
             + \phi_1 , \frac{3\pi}{8} + 
 	 \phi_2\right)} \right)\bigg].
\end{align*}

From \lem{a_xy}, for $x\leq y$, the state $\ket{\alpha(t)}$ takes the form
\begin{align*}
\langle x,y|\alpha(t)\rangle & = \frac{1}{\sqrt{2}}e^{-it\sqrt{2}}\left[ e^{-i\pi x /2}e^{i \pi y/4}  
	\left(\iint_{D_{\pi}}\frac{d\phi_{1}d\phi_{2}}{4\pi^{2}}\right.\right.\\
 &	\qquad  \left.\left. A(\phi_{1},\phi_{2}) 		
	e^{- i\phi_{1}\left(-\left\lfloor t\right\rfloor +\left\lfloor \frac{t}{\sqrt{2}}\right\rfloor +\frac{x+y}
	{2}\right)}e^{-2i\phi_{2}\left(-\left\lfloor t\right\rfloor -\left\lfloor \frac{t}{\sqrt{2}}\right\rfloor +\frac{x-y}
	{2}\right)}\right)\right.\\
 & \quad \pm e^{i\theta} e^{i\pi x/4}e^{-i \pi y/2} 
 	\left(\iint_{D_{\pi}}\frac{d\phi_{1}d\phi_{2}}{4\pi^{2}}\right.\\
 & \qquad \left.\left. A(\phi_{1},\phi_{2})
 	e^{-i\phi_{1}\left(-\left\lfloor t\right\rfloor +\left\lfloor \frac{t}{\sqrt{2}}\right\rfloor +\frac{x+y}{2}\right)}
 	e^{-2i\phi_{2}\left(-\left\lfloor t\right\rfloor -\left\lfloor \frac{t}{\sqrt{2}}\right\rfloor +\frac{y-x}
 	{2}\right)}\right)\right],
 \end{align*}
 where $D_{\pi}=[-\pi,\pi]\times[-\pi,\pi]$. Using these expressions for $\ket{\psi_A(t)}$ and $\ket{\alpha(t)}$,
we now write 
\[
P_{1}|\psi_{A}(t)\rangle-P_{1}|\alpha(t)\rangle=\pm |e_{1}(t)\rangle+|e_{2}(t)\rangle \pm|f_{1}(t)\rangle+|f_{2}(t)\rangle\pm|g_{1}(t)\rangle+|g_{2}(t)\rangle\pm|h(t)\rangle
\]
where each term in the above equation is supported only on states
$|x,y\rangle$ such that $y-x \geq C$, and (for $y - x \geq C$)
\begin{align*}
\langle x,y|e_{1}(t)\rangle & = 
	\frac{e^{i\theta}}{\sqrt{2}} e^{-it\sqrt{2}}
	e^{i \pi x/4}e^{-i\pi y/2}\iint_{D_{\pi}}\frac{d\phi_{1}d\phi_{2}}{4\pi^{2}}
	A(\phi_{1},\phi_{2}) \bigg[e^{-i\phi_{1}\left(-t+\frac{t}{\sqrt{2}}+\frac{x+y}{2}\right)}\\
&  \qquad\qquad e^{-2i\phi_{2}\left(-t-\frac{t}{\sqrt{2}}+\frac{y-x}{2}\right)}
	-e^{-i\phi_{1}\left(-\left\lfloor t\right\rfloor +\left\lfloor \frac{t}{\sqrt{2}}\right\rfloor +\frac{x+y}{2}
 		\right)}e^{-2i\phi_{2}\left(-\left\lfloor t\right\rfloor -\left\lfloor \frac{t}{\sqrt{2}}\right\rfloor 
 		+\frac{y-x}{2}\right)}\bigg]\\
\langle x,y|e_{2}(t)\rangle & = 
	\frac{1}{\sqrt{2}}e^{-it\sqrt{2}}e^{-i \pi x/2}e^{i\pi y/4}
	\iint_{D_{\pi}}\frac{d\phi_{1}d\phi_{2}}{4\pi^{2}}A(\phi_{1},\phi_{2})\bigg[e^{-i\phi_{1}
	\left(-t+\frac{t}{\sqrt{2}}+\frac{x+y}{2}\right)}\\
& \qquad\qquad e^{-2i\phi_{2}\left(-t-\frac{t}{\sqrt{2}}+\frac{x-y}{2}\right)}
		-e^{-i\phi_{1}\left(-\left\lfloor t\right\rfloor +\left\lfloor \frac{t}{\sqrt{2}}\right\rfloor +
		\frac{x+y}{2}\right)}e^{-2i\phi_{2}\left(-\left\lfloor t\right\rfloor -\left\lfloor 
		\frac{t}{\sqrt{2}}\right\rfloor +\frac{x-y}{2}\right)}\bigg]\\
\langle x,y|f_{1}(t)\rangle & = -
	\frac{e^{i\theta}}{\sqrt{2}} e^{-it\sqrt{2}}e^{i \pi x/4} 
	e^{-i\pi y/2}\iint_{D_{\pi}\setminus D_{\epsilon}}\frac{d\phi_{1}d\phi_{2}}{4\pi^{2}}
	A(\phi_{1},\phi_{2})\\
&  \qquad\qquad\qquad\qquad \qquad\qquad \qquad 
	e^{-i\phi_{1}\left(-t+\frac{t}{\sqrt{2}}+\frac{x+y}{2}\right)}e^{-2i\phi_{2}
	\left(-t-\frac{t}{\sqrt{2}}+\frac{y-x}{2}\right)}\\
\langle x,y|f_{2}(t)\rangle & =  -\frac{1}{\sqrt{2}}e^{-it\sqrt{2}}e^{-i\pi x/2}e^{i\pi y/4 }
	\iint_{D_{\pi}\setminus D_{\epsilon}}\frac{d\phi_{1}d\phi_{2}}{4\pi^{2}}
	A(\phi_{1},\phi_{2})\\
&  \qquad\qquad\qquad\qquad \qquad\qquad \qquad e^{-i\phi_{1}\left(-t+\frac{t}{\sqrt{2}}+\frac{x+y}{2}\right)}
	e^{-2i\phi_{2}\left(-t-\frac{t}{\sqrt{2}}+\frac{x-y}{2}\right)}\\
\langle x,y|g_{1}(t)\rangle & =   \frac{e^{i\theta}}{\sqrt{2}}e^{i \pi x/4}e^{- i \pi y/2}
	\iint_{D_{\epsilon}}\frac{d\phi_{1}d\phi_{2}}{4\pi^{2}}A(\phi_{1},\phi_{2})
	e^{-i\phi_{1}\left(\frac{x+y}{2}\right)}e^{-2i\phi_{2}\left(\frac{y-x}{2}\right)}\\
& \qquad \qquad 
	\left[e^{-it 4\cos(\frac{\pi}{8}+\frac{\phi_{1}}{2})\cos(\frac{3\pi}{8}
	+\phi_{2})} -e^{-it\left(\sqrt{2}+\sqrt{2}\left(\frac{\phi_{1}}{2}-\phi_{2}\right)
 	-2\left(\frac{\phi_{1}}{2}+\phi_{2}\right)\right)}\right]\\
\langle x,y|g_{2}(t)\rangle & =  \frac{1}{\sqrt{2}}e^{-i\pi x/2}e^{i\pi y/4}
	\iint_{D_{\epsilon}}\frac{d\phi_{1}d\phi_{2}}{4\pi^{2}}A(\phi_{1},\phi_{2})
	e^{-i\phi_{1}\left(\frac{x+y}{2}\right)}e^{-2i\phi_{2}\left(\frac{x-y}{2}\right)}\\
 & \qquad\qquad\left[e^{-it 4\cos(\frac{\pi}{8}+\frac{\phi_{1}}{2})\cos(\frac{3\pi}{8}+\phi_{2})} - e^{-it\left(\sqrt{2}+\sqrt{2}\left(\frac{\phi_{1}}{2}-\phi_{2}\right)
 	-2\left(\frac{\phi_{1}}{2}+\phi_{2}\right)\right)}\right]\\
\langle x,y|h(t)\rangle & =  \frac{1}{\sqrt{2}}e^{i\pi x/4}e^{-i \pi y/2}
	\iint_{D_{\epsilon}}\frac{d\phi_{1}d\phi_{2}}{4\pi^{2}}
	A(\phi_{1},\phi_{2})e^{-i\phi_{1}\left(\frac{x+y}{2}\right)}e^{-2i\phi_{2}\left(\frac{y-x}{2}\right)}\\
& \qquad \qquad e^{-it 4\cos(\frac{\pi}{8}+\frac{\phi_{1}}{2}) \cos(\frac{3\pi}{8}+\phi_{2})} \left(e^{i\theta_{\pm}(\frac{\pi}{4}+\phi_{1},\frac{3\pi}{8}+\phi_{2})}-e^{i\theta}\right).
\end{align*}
We now proceed to bound the norm of each of these states.  We repeatedly
make the following replacement inside the integrals (here $\phi_{1},\phi_{2},\tilde{\phi}_{1},\tilde{\phi}_2 \in (-\pi,\pi)$):
\begin{align*}
\sum_{x,y=-\infty}^{\infty}e^{ix\left(\frac{1}{2}\left(\phi_{1}-\tilde{\phi}_{1}\right)-\left(\phi_{2}-\tilde{\phi}_{2}\right)\right)}e^{iy\left(\frac{1}{2}\left(\phi_{1}-\tilde{\phi}_{1}\right)+\left(\phi_{2}-\tilde{\phi}_{2}\right)\right)}
 & =  4\pi^{2}\delta(\phi_{1}-\tilde{\phi}_{1}) \delta(\phi_{2}-\tilde{\phi}_{2}).
 \end{align*}
 Using this formula we get
 \begin{align*}
\langle e_{1}(t)|e_{1}(t)\rangle & =  \sum_{y-x\geq C}\langle e_{1}(t)|x,y\rangle\langle x,y|e_{1}(t)\rangle\\
   & \leq  \sum_{x=-\infty}^{\infty}\sum_{y=-\infty}^{\infty}\bigg|\frac{1}{\sqrt{2}}
	 	\iint_{D_{\pi}}\frac{d\phi_{1}d\phi_{2}}{4\pi^{2}}A(\phi_{1},\phi_{2})
	 	\bigg[e^{-i\phi_{1}\left(-t+\frac{t}{\sqrt{2}}+\frac{x+y}{2}\right)}\\
	&\qquad e^{-2i\phi_{2}\left(-t-\frac{t}
	 	{\sqrt{2}}+\frac{y-x}{2}\right)}-e^{-i\phi_{1}\left(-\left\lfloor t\right\rfloor +\left\lfloor \frac{t}{\sqrt{2}}\right\rfloor 
 		+\frac{x+y}{2}\right)}e^{-2i\phi_{2}\left(-\left\lfloor t\right\rfloor -\left\lfloor \frac{t}{\sqrt{2}}
 		\right\rfloor +\frac{y-x}{2}\right)}\bigg]\bigg|^{2}\\
 & =  \frac{1}{2}\iint_{D_{\pi}}\frac{d\phi_{1}d\phi_{2}}{4\pi^{2}}\left|A(\phi_{1},\phi_{2})
 		\right|^{2}\bigg|e^{-i\phi_{1}\left(-t+\frac{t}{\sqrt{2}}\right)}e^{-2i\phi_{2}\left(-t-\frac{t}
 		{\sqrt{2}}\right)}\\
 & \qquad -e^{-i\phi_{1}\left(-\left\lfloor t\right\rfloor +\left\lfloor \frac{t}
 		{\sqrt{2}}\right\rfloor \right)}e^{-2i\phi_{2}\left(-\left\lfloor t\right\rfloor -\left\lfloor 
 		\frac{t}{\sqrt{2}}\right\rfloor \right)}\bigg|^{2}.
\end{align*}
 Now use the fact that $\left|e^{-ic}-1\right|^{2}\leq c^{2}$ for
$c\in\mathbb{R}$ to get 
\begin{align*}
\langle e_{1}(t)|e_{1}(t)\rangle & \leq  \frac{1}{2}\iint_{D_{\pi}}\left(\frac{d\phi_{1}d\phi_{2}}
	{4\pi^{2}}\right)\left|A(\phi_{1},\phi_{2})\right|^{2}\Bigg(-\phi_{1}\left(-t+\frac{t}{\sqrt{2}}+
	\left\lfloor t\right\rfloor -\left\lfloor \frac{t}{\sqrt{2}}\right\rfloor \right)\\
 &  \quad \quad -2\phi_{2}\left(-t-\frac{t}{\sqrt{2}}+\left\lfloor t\right\rfloor +\left\lfloor \frac{t}{\sqrt{2}}\right\rfloor \right)\Bigg)^{2}\\
 & \leq  4\iint_{D_{\pi}}\frac{d\phi_{1}d\phi_{2}}{4\pi^{2}}\left|A(\phi_{1},\phi_{2})\right|^{2}\left(\phi_{1}^{2}+4\phi_{2}^{2}\right)
\end{align*}
using the Cauchy-Schwarz inequality and the fact that $\left|t-{t}/{\sqrt{2}}-\left\lfloor t\right\rfloor -\left\lfloor {t}/{\sqrt{2}}\right\rfloor \right|\leq2$.
So 
\begin{align*}
\langle e_{1}(t)|e_{1}(t)\rangle & \leq 4\left(\iint_{D_{\pi}\setminus D_{\epsilon}}\frac{d\phi_{1}d\phi_{2}}{4\pi^{2}}+\iint_{D_{\epsilon}}\frac{d\phi_{1}d\phi_{2}}{4\pi^{2}}\right)\left|A(\phi_{1},\phi_{2})\right|^{2}\left(\phi_{1}^{2}+4\phi_{2}^{2}\right)\\
 & \leq  4\cdot 5\pi^{2} \frac{4\pi}{L\epsilon} + 20\epsilon^{2}\\
 & =  \frac{80\pi^{3}}{L\epsilon}+20\epsilon^{2}
\end{align*}
where we have used \lem{Let--and} and the fact that $\phi_{1}^{2}+4\phi_{2}^{2}\leq 5\epsilon^{2}$ on $D_\epsilon$. Similarly, 
\[
\langle e_{2}(t)|e_{2}(t)\rangle\leq\frac{80\pi^{3}}{L\epsilon}+20\epsilon^{2}.
\]
 Now
 \begin{align*}
\langle f_{1}(t)|f_{1}(t)\rangle & \leq  \frac{1}{2}\iint_{D_{\pi}\setminus D_{\epsilon}}\frac{d\phi_{1}d\phi_{2}}{4\pi^{2}}\left|A(\phi_{1},\phi_{2})\right|^{2}\\
 & \leq  \frac{2\pi}{L\epsilon}
 \end{align*}
 by \lem{Let--and}, and similarly
 \[
\langle f_{2}(t)|f_{2}(t)\rangle\leq\frac{2\pi}{L\epsilon}.
\]
Moving on to the next term, 
\begin{align}
\langle g_{1}(t)|g_{1}(t)\rangle & \leq  \frac{1}{2}\iint_{D_{\epsilon}}\frac{d\phi_{1}d\phi_{2}}{4\pi^{2}}
	\left|A(\phi_{1},\phi_{2})\right|^{2}\Bigg|e^{-it 4\cos(\frac{\pi}{8}+\frac{\phi_{1}}{2})
	\cos(\frac{3\pi}{8}+\phi_{2})}\nonumber \\
 &  \qquad\qquad\qquad -e^{-it\left(\sqrt{2}+\sqrt{2}\left(\frac{\phi_{1}}{2}-\phi_{2}\right)-2\left(\frac{\phi_{1}}{2}+\phi_{2}\right)\right)}
 	\Bigg|^{2}\nonumber \\
 & \leq  \frac{1}{2}\iint_{D_{\epsilon}}\frac{d\phi_{1}d\phi_{2}}{4\pi^{2}}
 	\Bigg[\left|A(\phi_{1},\phi_{2})\right|^{2}t^{2} \left(4\cos\left(\frac{\pi}{8}+\frac{\phi_{1}}{2}\right)\cos\left(\frac{3\pi}{8}+\phi_{2}\right)\right.\nonumber\\
&\qquad\qquad\qquad\left.
 	-\sqrt{2}-\sqrt{2}\left(\frac{\phi_{1}}{2}-\phi_{2}\right)+2\left(\frac{\phi_{1}}{2}+\phi_{2}\right)\right)^{2}\Bigg]
 	\label{eq:g_bound}
\end{align}
using $\left|e^{-ic}-1\right|^{2}\leq c^{2}$ for $c\in\mathbb{R}$.
Now 
\begin{align*}
  & 4\cos\left(\frac{\pi}{8}+\frac{\phi_{1}}{2}\right)\cos\left(\frac{3\pi}{8}+\phi_{2}\right) \\
 &\quad =  
	2\cos\left(\frac{\pi}{2}+\frac{\phi_{1}}{2}+\phi_{2}\right)+2\cos\left(-\frac{\pi}{4}+\frac{\phi_{1}}{2}-\phi_{2}\right)\\
 &\quad =  - 2 \sin\left(\frac{\phi_1}{2} + \phi_2\right) + \sqrt{2} \cos\left(\frac{\phi_1}{2}-\phi_2\right) + 
 	\sqrt{2} \sin\left(\frac{\phi_1}{2} - \phi_2\right)
 \end{align*}
so 
\begin{align*}
 &   \left|4\cos\left(\frac{\pi}{8}+\frac{\phi_{1}}{2}\right)\cos\left(\frac{3\pi}{8}+\phi_{2}\right)-\sqrt{2}-\sqrt{2}
 	\left(\frac{\phi_{1}}{2}-\phi_{2}\right)+2\left(\frac{\phi_{1}}{2}+\phi_{2}\right)\right|\\
 & \quad \leq  \left|\sqrt{2}\left(\cos\left(\frac{\phi_{1}}{2}-\phi_{2}\right)-1\right)\right| 
 	+\left|\sqrt{2}\left(\sin\left(\frac{\phi_{1}}{2}-\phi_{2}\right)-\left(\frac{\phi_{1}}{2}-\phi_{2}\right)\right)\right|\\
 & \qquad +\left|2\left(\sin\left(\frac{\phi_{1}}{2}+\phi_{2}\right)-\left(\frac{\phi_{1}}{2}+\phi_{2}\right)\right)\right|\\
 & \quad \leq  \sqrt{2}\left(\frac{\phi_1}{2}-\phi_{2}\right)^{2}+\sqrt{2}\left(\frac{\phi_{1}}{2}
 	-\phi_{2}\right)^{2}+2\left(\frac{\phi_{1}}{2}+\phi_{2}\right)^{2} \\
 & \quad \leq  4\left(\left(\frac{\phi_{1}}{2}+\phi_{2}\right)^{2}+\left(\frac{\phi_{1}}{2}-\phi_{2}\right)^{2}\right),
\end{align*}
using $|{\cos x -1}|\leq x^2$ and $|{\sin x-x}|\leq x^2$ for $x\in \mathbb{R}$. Plugging this into equation \eq{g_bound}, we get 
\begin{align*}
\langle g_{1}(t)|g_{1}(t)\rangle & \leq \frac{1}{2}\iint_{D_{\epsilon}}
	\frac{d\phi_{1}d\phi_{2}}{4\pi^{2}} 16 \left|A(\phi_{1},\phi_{2})\right|^{2}t^{2}
	\left(\left(\frac{\phi_{1}}{2} + \phi_2\right)^2+\left(\frac{\phi_1}{2} - \phi_{2}\right)^2\right)^{2} \\
 & \leq  16 t^2 \iint_{D_{\epsilon}}\frac{d\phi_{1}d\phi_{2}}{4\pi^{2}}
 	\left|A(\phi_{1},\phi_{2})\right|^{2}\left(\left(\frac{\phi_1}{2} + \phi_2\right)^4 
 	+ \left(\frac{\phi_1}{2} - \phi_2\right)^4\right)\\
 & \leq  \frac{16 t^{2}}{L^2} \iint_{D_{\epsilon}}\frac{d\phi_{1}d\phi_{2}}{4\pi^{2}}
	  \frac{\sin^2(\frac{L}{2} [\frac{\phi_1}{2} + \phi_2])}
	 {\sin^2(\frac{1}{2} [\frac{\phi_1}{2} + \phi_2])}
	 \frac{\sin^2(\frac{L}{2} [-\frac{\phi_1}{2} + \phi_2])}
	 {\sin^2(\frac{1}{2} [-\frac{\phi_1}{2} + \phi_2])}\\
 & \qquad \left(\left(\frac{\phi_1}{2}+ \phi_2\right)^4 
 	+ \left(\frac{\phi_1}{2} - \phi_2\right)^4\right)
\end{align*}
where we used the Cauchy-Schwarz inequality in the second line and equation \eq{A_summed} in the last line.  Changing coordinates to
\[
  \alpha_1 = \phi_1 + \frac{ \phi_2}{2} \qquad \alpha_2 = \frac{\phi_1}{2} -\phi_2 
 \]
and realizing that $|\alpha_1|,|\alpha_2| < 3\epsilon/2$ for $(\phi_1,\phi_2)\in D_{\epsilon}$, we see that
\begin{align*}
  \bra{g_1(t)} g_1(t)\rangle &\leq \frac{16t^2}{L^2} \int_{-{3\epsilon/2}}^{3\epsilon/2} \frac{d\alpha_1}{2\pi}
      \int_{-{3\epsilon/2}}^{3\epsilon/2} \frac{d\alpha_2}{2\pi} 
      \frac{\sin^2(\frac{1}{2} L\alpha_1)}
	 {\sin^2(\frac{1}{2} \alpha_1)}
	 \frac{\sin^2(\frac{1}{2} L\alpha_2)}
	 {\sin^2(\frac{1}{2} \alpha_2)} \left(\alpha_1^4+\alpha_2^4\right)\\
   &= \frac{32t^2}{L^2} \int_{-{3\epsilon/2}}^{3\epsilon/2} \frac{d\alpha_1}{2\pi}
      \int_{-{3\epsilon/2}}^{3\epsilon/2} \frac{d\alpha_2}{2\pi}
      \frac{\sin^2(\frac{1}{2} L\alpha_1)}
	 {\sin^2(\frac{1}{2} \alpha_1)}
	 \frac{\sin^2(\frac{1}{2} L\alpha_2)}
	 {\sin^2(\frac{1}{2} \alpha_2)} \alpha_1^4\\
   &\leq \frac{32t^2}{L} \int_{-3\epsilon/2}^{3\epsilon/2} \frac{d\alpha_1}{2\pi}  
     \frac{\sin^2(\frac{1}{2} L\alpha_1)}
	 {\sin^2(\frac{1}{2} \alpha_1)} \alpha_1^4\\
   &\leq \frac{32 t^2}{L} \int_{-3\epsilon/2}^{3\epsilon/2} \frac{d\alpha_1}{2\pi} \frac{\pi^2}{\alpha_1^2} \alpha_1^4\\
   &= \frac{36 \pi t^2\epsilon^3}{L},
\end{align*}
with the same bound on $\langle g_{2}(t)|g_{2}(t)\rangle$.
 
Finally, 
\begin{align*}
\langle h(t)|h(t)\rangle & \leq  \frac{1}{2}\iint_{D_{\epsilon}}\frac{d\phi_{1}d\phi_{2}}{4\pi^{2}}
	\left|A(\phi_{1},\phi_{2})\right|^{2}\left|e^{i\theta_\pm (\tfrac{\pi}{4}+\phi_{1},
	\tfrac{3\pi}{8}+\phi_{2})}-e^{i\theta}\right|^{2}.
 \end{align*}
Recall that $e^{i\theta_\pm (p_1,p_2)}=T(p_1,p_2) \pm R(p_1,p_2)$ is obtained by solving for the effective single-particle S-matrix for the Hamiltonian \eq{vr_eqn}. For $p_1$ near ${\pi}/{4}$ we divide this Hamiltonian by $2\cos({p_1}/{2})$ to put it in the form considered in \cite{Childs_Gosset}, where the potential term is now $\mathcal{V}(|r|)/(2\cos({p_1}/{2}))$. The entries $T(p_1,p_2)$ and $R(p_1,p_2)$ of this S-matrix are bounded rational functions of $z=e^{ip_2}$ and $(2\cos({p_1}/{2}))^{-1}$ \cite{Childs_Gosset}, so they are differentiable as a function of $p_1$ and $p_2$ on some neighborhood $U$ of $({\pi}/{4},{3\pi}/{8})$ (and have bounded partial derivatives on this neighborhood).

For $\epsilon$ small enough that $D_\epsilon = ({\pi}/{4}, {3\pi}/{8}) \subset U$ we get, using the mean value theorem and the fact that $\theta=\theta_\pm ({\pi}/{4}, {3\pi}/{8})$,
\begin{align*}
\left|e^{i\theta_\pm (\tfrac{\pi}{4}+\phi_{1}, \tfrac{3\pi}{8}+\phi_{2})}-e^{i\theta}\right| & \leq \sqrt{\phi_1^2+\phi_2^2} \max_{U} \big|\vec{\nabla} e^{i\theta_\pm}\big| \quad \text{for }(\phi_1,\phi_2)\in D_{\epsilon}\\
& \leq  \epsilon \Gamma
\end{align*}
for some constant $\Gamma$ (independent of $L$).
Therefore
\begin{align*}
\langle h(t)|h(t)\rangle & \leq  \frac{1}{2}\iint_{D_{\epsilon}}\frac{d\phi_{1}d\phi_{2}}{4\pi^{2}}
	\left|A(\phi_{1},\phi_{2})\right|^{2} \epsilon^2 \Gamma^2\\
& \leq \frac{1}{2}\Gamma^2 \epsilon^2.
\end{align*}
 
Putting these bounds together, we get 
\begin{align*}
\Norm{P_1|\psi_A(t)\rangle-P_1|\alpha(t)\rangle}  
 & \leq \norm{|e_{1}(t)\rangle} + \norm{|e_{2}(t)\rangle} + 
        \norm{|f_{1}(t)\rangle} + \norm{|f_{2}(t)\rangle} \\
 &  \qquad + \norm{|g_{1}(t)\rangle} + \norm{|g_{2}(t)\rangle} 
           + \norm{|h(t)\rangle} \\
 & \leq 2\left(\frac{80\pi^{3}}{L\epsilon}+20\epsilon^{2}\right)^{\frac{1}{2}}
 	+2\left(\frac{2\pi}{L\epsilon}\right)^{\frac{1}{2}}+2\left(\frac{36\pi t^{2}\epsilon^{3}}{L}\right)^{\frac{1}{2}}+\frac{1}{\sqrt{2}}\Gamma \epsilon.
\end{align*}
 Letting $\epsilon={a}/{\sqrt{L}}$ and $t\leq c_{0}L$ we get 
\begin{equation}
\left\Vert \,P_1|\psi_A(t)\rangle-P_1|\alpha(t)\rangle\right\Vert = \O(L^{-{1}/{4}}).\label{eq:psiA_alpha}
\end{equation}

Since $P_2\ket{\alpha(t)}$ has support on at most $4CL$ basis states $|x,y\rangle$, and since $|\langle x,y|P_2|\alpha(t)\rangle|^2=\O( L^{-2})$, we get
\begin{equation}
  \left\Vert P_2\ket{\alpha(t)}\right\Vert = \O(L^{-{1}/{2}}).\label{eq:P2alpha_bound}
\end{equation}
Furthermore, \lem{psiA_P2} says that
\begin{equation}
\label{eq:boundinside}
\left\Vert P_2|\psi_{A}(t)\rangle\right\Vert=\O\left(\frac{\log L}{\sqrt{L}}\right).
\end{equation}
Using equations \eq{psiA_alpha}, \eq{P2alpha_bound}, and \eq{boundinside}, we find
\begin{align*}
\left\Vert |\psi_{A}(t)\rangle-|\alpha(t)\rangle\right\Vert  & \leq 2\left\Vert P_{1}|\psi_{A}(t)\rangle-P_{1}|\alpha(t)\rangle\right\Vert +\left\Vert P_{2}|\alpha(t)\rangle\right\Vert +\left\Vert P_{2}|\psi_{A}(t)\rangle\right\Vert \\
 & = \O(L^{-{1}/{4}}),
\end{align*}
which completes the proof. 
\end{proof}

\subsection{Technical lemmas}
\label{sec:techlem}

In this section we prove three lemmas that are used in the proof of \thm{twopart}.

\begin{lemma} \label{lem:alpha}
Let $|\alpha(t)\rangle$ be defined as in \thm{twopart}.
Then
\[
\langle\alpha(t)|\alpha(t)\rangle=1+\O(L^{-1}).
\]
\end{lemma}

\begin{proof}
Define
\[
\Pi=\sum_{x\leq y}|x,y\rangle\langle x,y|.
\]
Note that, since $\langle x,y|\alpha(t)\rangle=\pm\langle y,x|\alpha(t)\rangle$ for $x\neq y$,
\begin{align*}
\langle\alpha(t)|\alpha(t)\rangle & = 2\langle\alpha(t)|\Pi|\alpha(t)\rangle-\sum_{x=-\infty}^{\infty}\langle\alpha(t)|x,x\rangle\langle x,x|\alpha(t)\rangle\\
 & = 2\langle\alpha(t)|\Pi|\alpha(t)\rangle+\O(L^{-1})
\end{align*}
where the last line follows since $|\langle x,x|\alpha(t)\rangle|^{2}$
is nonzero for at most $L$ values of $x$ and $|\langle x,x|\alpha(t)\rangle|^{2}=\O(L^{-2})$.
We now show that
\[
\langle\alpha(t)|\Pi|\alpha(t)\rangle=\frac{1}{2}+O(L^{-1}).
\]
Note that
\begin{align*}
\langle\alpha(t)|\Pi|\alpha(t)\rangle &= \frac{1}{2L^{2}}\sum_{x\leq y}\Bigg(F(x,y,t) +F(y,x,t)  \\
&\qquad \pm e^{i\theta}e^{\frac{3i\pi}{4}x}e^{-\frac{3i\pi}{4}y}F(x,y,t) F(y,x,t) \\
&\qquad \pm e^{-i\theta}e^{-\frac{3i\pi}{4}x}e^{\frac{3i\pi}{4}y}F(x,y,t) F(y,x,t) \Bigg).
\end{align*}
Now
\[
\sum_{x\leq y}F(y,x,t) =\sum_{y\leq x}F(x,y,t) 
\]
and
\begin{align*}
\frac{1}{2L^{2}}\sum_{x\leq y}\left[F(x,y,t) +F(y,x,t) \right] &= \frac{1}{2L^{2}}\left(\sum_{x=-\infty}^{\infty}\sum_{y=-\infty}^{\infty}F(x,y,t) -\sum_{x=-\infty}^{\infty}F(x,x,t) \right)\\
 &= \frac{1}{2}+\O(L^{-1}).
\end{align*}

We now establish the bound
\[
\left|\frac{1}{2L^{2}}\sum_{x\leq y}e^{\frac{3i\pi}{4}x}e^{-\frac{3i\pi}{4}y}F(x,y,t) F(y,x,t) \right|=\O(L^{-1})
\]
to complete the proof.  To get this bound, note that both $F(x,y,t) =1$ and $F(y,x,t) =1$ if and only if $x,y \in \mathcal{B}$ where
\[
\mathcal{B}=\{-M-L+2\left\lfloor t\right\rfloor ,\ldots,-M-1+2\left\lfloor t\right\rfloor \} \cap \left\{M+1-2\left\lfloor \frac{t}{\sqrt{2}}\right\rfloor ,\ldots,M+L-2\left\lfloor \frac{t}{\sqrt{2}}\right\rfloor \right\}.
\]
Observe that
 \[
\mathcal{B}=\{j,j+1,\ldots,j+l\}
\]
for some $j,l\in\mathbb{Z}$ with $l<L$. So
\begin{align*}
\frac{1}{2L^{2}}\left|\sum_{x\leq y}e^{\frac{3i\pi}{4}x}e^{-\frac{3i\pi}{4}y}F(x,y,t) F(y,x,t) \right| & = \frac{1}{2L^{2}}\left|\sum_{x,y\in \mathcal{B},\, x\leq y}e^{\frac{3i\pi}{4}x}e^{-\frac{3i\pi}{4}y}\right|\\
 &  = \frac{1}{2L^{2}}\left|\sum_{y=j}^{j+l}\sum_{x=j}^{y}e^{\frac{3i\pi}{4}x}e^{-\frac{3i\pi}{4}y}\right|\\
 &= \frac{1}{2L^{2}}\left|\sum_{y=j}^{j+l}e^{-\frac{3i\pi}{4}y}e^{3i\frac{\pi}{4}j}\frac{e^{3i\frac{\pi}{4}\left(y+1-j\right)}-1}{e^{3i\frac{\pi}{4}}-1}\right|\\
 & \leq \frac{l+1}{2L^{2}}\frac{2}{\left|e^{3i\frac{\pi}{4}}-1\right|}\\
 &= \O(L^{-1})
 \end{align*}
since $l<L$.
\end{proof}
\begin{lemma}
\label{lem:Let--and}Let $k\in(-\pi,0)\cup(0,\pi)$ and $0<\epsilon<\min\left\{ \pi-|k|,|k|\right\}$.
Let
\begin{align*}
D_{\epsilon} & =  \left[-\epsilon,\epsilon\right]\times\left[-\epsilon,\epsilon\right]\\
D_{\pi} & =  \left[-\pi,\pi\right]\times\left[-\pi,\pi\right].
\end{align*}
Then 
\begin{align*}
\iint_{D_{\pi}}\frac{d\phi_{1}d\phi_{2}}{4\pi^{2}}|A(\phi_{1},\phi_{2})|^{2} & =  1\\
\iint_{D_{\pi}\setminus D_{\epsilon}}\frac{d\phi_{1}d\phi_{2}}{4\pi^{2}}|A(\phi_{1},\phi_{2})|^{2} & \leq  \frac{4\pi}{L\epsilon}\\
\iint_{D_{\epsilon}}\frac{d\phi_{1}d\phi_{2}}{4\pi^{2}}|B(\phi_{1},\phi_{2},k)|^{2} & \leq  \frac{16\pi^{2}}{L^{2}\epsilon^{2}}\end{align*}
where $A(\phi_1,\phi_2)$ and $B(\phi_1,\phi_2,k)$ are given by equation \eq{A}.
 \end{lemma}
\begin{proof}
Using equation \eq{A} we get
\begin{align*}
|A(\phi_{1},\phi_{2})|^{2} & =  \frac{1}{L^{2}} \sum_{x,\tilde{x}=-(M+L)}
  ^{-(M+1)} \sum_{y,\tilde{y}=M+1}^{M+L}
	e^{i\frac{\phi_{1}}{2}\left(x+y-(\tilde{x}+\tilde{y})\right)}
        e^{i\phi_{2}\left(x-y-(\tilde{x}-\tilde{y})\right)}.
\end{align*}
Now 
\[
\int_{-\pi}^{\pi}\frac{d\phi_{2}}{2\pi}e^{i\phi_{2}
  \left(x-y-\tilde{x}+\tilde{y}\right)}=\delta_{x-y,\tilde{x}-\tilde{y}},
\]
so (suppressing the limits of summation for readability) 
\begin{align*}
\iint_{D_{\pi}}\frac{d\phi_{1}d\phi_{2}}{4\pi^{2}}
   |A(\phi_{1},\phi_{2})|^{2} 
   & =  \frac{1}{L^{2}} \int_{-\pi}^{\pi}\frac{d\phi_{1}}{2\pi}
     \sum_{x,\tilde{x}}\sum_{y,\tilde{y}}	e^{i\phi_{1}\left(y-\tilde{y}\right)}
     \delta_{x-y,\tilde{x}-\tilde{y}} \\
 & =  \frac{1}{L^{2}}\sum_{x,\tilde{x}}\sum_{y,\tilde{y}}\delta_{y,\tilde{y}}
 	\delta_{x-y,\tilde{x}-\tilde{y}}\\
 & =  1
\end{align*}
which proves the first part.

By performing the sums in equation \eq{A} we get 
\begin{equation}
|A(\phi_{1},\phi_{2})|^{2}
=\frac{1}{L^{2}}\frac{\sin^{2}(\frac{1}{2}L[\frac{\phi_1}{2}+\phi_{2}])} {\sin^{2}(\frac{1}{2}[\frac{\phi_1}{2}+\phi_{2}])}
\frac{\sin^{2}(\frac{1}{2}L[\frac{\phi_1}{2} -\phi_{2}])}
{\sin^{2}(\frac{1}{2}[\frac{\phi_1}{2}-\phi_{2}])}.
\label{eq:A_summed}
\end{equation}
Letting $\alpha_{1}={\phi_1}/{2}+\phi_{2}$ and
$\alpha_{2}={\phi_1}/{2}-\phi_{2}$, we see that
$|\alpha_{1}|\leq 3\pi/2$, $|\alpha_{2}|\leq 3\pi/2$,
and $\alpha_{1}^{2}+\alpha_{2}^{2}\geq  5\epsilon^{2}/2$ whenever
$(\phi_{1},\phi_{2})\in D_{\pi}\setminus D_{\epsilon}$. Defining
$D_{3\pi/2}=[-3\pi/2,3\pi/2]^{2}$
we get $(\alpha_{1},\alpha_{2})\in D_{3\pi/2}\setminus D_{\epsilon}$
whenever $(\phi_{1},\phi_{2})\in D_{\pi}\setminus D_{\epsilon}$.
Hence 
\begin{align*}
\iint_{D_{\pi}\setminus D_{\epsilon}}\frac{d\phi_{1}d\phi_{2}}{4\pi^{2}}|A(\phi_{1},\phi_{2})|^{2} 
	& \leq \frac{1}{L^{2}} \iint_{D_{3\pi/2}\setminus D_{\epsilon}} \frac{d\alpha_{1}d\alpha_{2}}{4\pi^{2}}
	\frac{\sin^{2}(\frac{1}{2}L\alpha_{1})}{\sin^{2}
	(\frac{1}{2}\alpha_{1})}\frac{\sin^{2}(\frac{1}{2}L\alpha_{2})}
	{\sin^{2}(\frac{1}{2}\alpha_{2})}\\
 & \leq  \frac{4}{L}\left(\frac{1}{L}\int_{-\frac{3\pi}{2}}^{\frac{3\pi}{2}}\frac{d\alpha_{1}}{2\pi}\frac{\sin^{2}(
 	\frac{1}{2}L\alpha_{1})}{\sin^{2}(\frac{1}{2}\alpha_{1})}\right)
 	\left(\int_{\epsilon}^{3\pi/2}\frac{d\alpha_{2}}{2\pi}\frac{\sin^{2}(\frac{1}{2}L\alpha_{2})}
 	{\sin^{2}(\frac{1}{2}\alpha_{2})}\right)\\
 & \leq  \frac{4}{L}\left(\int_{-2\pi}^{2\pi}\frac{d\alpha_{1}}{2\pi}\frac{1}{L}\frac{\sin^{2}(\frac{1}{2}L\alpha_{1})}{\sin^{2}(\frac{1}{2}\alpha_{1})}\right)\left(\int_{\epsilon}^{\frac{3\pi}{2}}\frac{d\alpha_{2}}{2\pi}\frac{1}{\sin^{2}(\frac{1}{2}\alpha_{2})}\right)\\
 & = \frac{8}{L} \left(\int_{\epsilon}^{\pi}\frac{d\alpha_{2}}{2\pi}\frac{1}{\sin^{2}(\frac{1}{2}\alpha_{2})}+\int_{\pi}^{\frac{3\pi}{2}}\frac{d\alpha_{2}}{2\pi}\frac{1}{\sin^{2}(\frac{1}{2}\alpha_{2})}\right)\\
 & \leq  \frac{8}{L}\left(\int_{\epsilon}^{\pi}\frac{d\alpha_{2}}{2\pi}\frac{\pi^{2}}{\alpha_{2}^{2}}+2\int_{\pi}^{\frac{3\pi}{2}}\frac{d\alpha_{2}}{2\pi}\right)\\
 & =  \frac{4\pi}{L\epsilon}
 \end{align*}
which proves the second inequality (in the next-to-last line we
have used the fact that $\sin({x}/{2})>{x}/{\pi}$ for $x \in (0,\pi)$
and $\sin^{2}({x}/{2})>{1}/{2}$ for $x \in (\pi,{3\pi}/{2})$).

Now
\begin{align*}
|B(\phi_{1},\phi_{2},k)|^{2} & =  |A(\phi_{1},-\phi_{2}-2k)|^{2}\\
 & \leq  \frac{1}{L^{2}}\frac{1}{\sin^{2}(\frac{1}{2}
     [\frac{\phi_{1}}{2}+\phi_{2}+2k])}
 \frac{1}{\sin^{2}(\frac{1}{2}
     [-\frac{\phi_{1}}{2}+{\phi_{2}}+2k])}.
\end{align*}
If $(\phi_{1},\phi_{2})\in D_{\epsilon}$ then $|k|-{3\epsilon}/{4} 
\leq\left|\pm{\phi_{1}}/{4}+{\phi_{2}}/{2}+k\right|
\leq|k|+{3\epsilon}/{4}$.
Noting that $\epsilon$ is chosen such that $0 < \epsilon < 
\min \{\pi-|k|,|k|\}$, we get
\[
\frac{\epsilon}{4}\leq\left|\pm\frac{\phi_{1}}{4}+\frac{\phi_{2}}{2}+k\right|\leq\pi-\frac{\epsilon}{4}
\]
so 
\begin{align*}
|B(\phi_{1},\phi_{2},k)|^{2} & \leq  \frac{1}{L^{2}}\frac{1}{\sin^{4}(\frac{\epsilon}{4})}\\
 	& \leq  \frac{16\pi^{4}}{L^{2}\epsilon^{4}}
\end{align*}
and 
\begin{align*}
\iint_{D_{\epsilon}}\frac{d\phi_{1}d\phi_{2}}{4\pi^{2}}|B(\phi_{1},\phi_{2},k)|^{2} & \leq  \frac{1}{4\pi^{2}}\left(2\epsilon\right)^{2}\left(\frac{16\pi^{4}}{L^{2}\epsilon^{4}}\right)\\
 & =  \frac{16\pi^{2}}{L^{2}\epsilon^{2}}. \qedhere
 \end{align*}
\end{proof}

\begin{lemma}
\label{lem:a_xy}
Let $a_{xy}(t)$ be as in \thm{twopart}. For $x\leq y$,
\begin{align*}
a_{xy}(t) & =  \f{1}{\sqrt{2}} e^{- i t\sqrt{2}}\left[e^{-i \pi x/2} e^{i \pi y/4} \left(\iint_{D_{\pi}} 
	\f{d\phi_1 d\phi_2}{4\pi^2} A(\phi_1,\phi_2)\right.\right.\\
	& \qquad\quad \left.e^{-i \phi_1 \left( -\lfloor t\rfloor + \lfloor
	\frac{t}{\sqrt{2}}\rfloor + \frac{x+y}{2} \right)} e^{-2 i \phi_2 \left(-\lfloor t 
	\rfloor -\lfloor \frac{t}{\sqrt{2}}\rfloor + \frac{x-y}{2}\right)}\right)\\
& \quad \pm e^{i\theta} e^{i\pi x/4}e^{-i \pi y/2} \left(\iint_{D_{\pi}} \frac{d\phi_1 d\phi_2}{4\pi^2}
	A(\phi_1,\phi_2) \right.\\
&	\qquad\quad\left.\left.e^{- i \phi_1 \left(- \lfloor t\rfloor + \lfloor \frac{t}{\sqrt{2}}\rfloor + 
	\frac{x+y}{2}\right) } e^{-2 i \phi_2 \left(-\lfloor t \rfloor -\lfloor \frac{t}{\sqrt{2}}\rfloor
	+ \frac{y-x}{2}\right)}\right)
	\right].
\end{align*}
\end{lemma}

\begin{proof}
The lemma follows from \eq{a_xy} and the fact that, for any two numbers $\gamma_{1},\gamma_{2}$ such that $\gamma_{1}+\gamma_{2},\gamma_{1}-\gamma_{2}\in \mathbb{Z}$,
\[
\iint_{D_\pi}\frac{d\phi_{1} d\phi_{2}}{4\pi^2} A(\phi_{1},\phi_{2})e^{i\gamma_{1}\phi_{1}+2i \gamma_{2}\phi_{2}}=\begin{cases}
\frac{1}{L} &\text{ if }(-\gamma_{1}-\gamma_{2},-\gamma_{1}+\gamma_{2})\in \mathcal{S} \\
0 &\text{ otherwise}
\end{cases}
\]
where $\mathcal{S} = \{-M-L,\ldots, -M-1\} \times \{M+1,\ldots, M+L\}$.  To establish this formula, observe that
\begin{align*}
  \iint_{D_{\pi}}\frac{d\phi_1 d\phi_{2}}{4\pi^2}A(\phi_{1},\phi_{2})e^{ i\gamma_{1}\phi_{1}+2i\gamma_{2}\phi_{2}} 
  	& =  \frac{1}{L}\sum_{x=-M-L}^{-M-1}\sum_{y=M+1}^{M+L}\iint_{D_{\pi}}\frac{d\phi_1 d\phi_{2}}{4\pi^2}
  		e^{i\phi_{1}\left(\gamma_{1}+\frac{x+y}{2}\right)}e^{i\phi_{2}\left(x-y+2\gamma_{2}\right)}\\
 & =  \frac{1}{L}\sum_{x=-M-L}^{-M-1}\,\sum_{y=M+1}^{M+L}\int_{-\pi}^{\pi}\frac{d\phi_{1}}{2\pi}
 	e^{i\phi_{1}\left(\gamma_{1}+\frac{x-y}{2}\right)}\delta_{y,-x-2\gamma_{2}}.
\end{align*}
 Here we have performed the integral over $\phi_{2}$ using the fact
that $2\gamma_{2}$ is an integer.  We then have
\begin{align*}
\iint_{D_\pi}\frac{d\phi_{1} d\phi_{2}}{4\pi^2} A(\phi_{1},
	\phi_{2})e^{i\gamma_{1}\phi_{1}+2i\gamma_{2}\phi_{2}} 
 & =  \frac{1}{L}\sum_{x=-M-L}^{-M-1}\,
	\sum_{y=M+1}^{M+L}\int_{-\pi}^{\pi}\frac{d\phi_{1}}{2\pi}e^{i\phi_{1}
	\left(\gamma_{1}+x+\gamma_{2}\right)}\delta_{y,-x-2\gamma_{1}}\\
 & =  \frac{1}{L}\sum_{x=-M-L}^{-M-1}\,\sum_{y=M+1}^{M+L}\delta_{x,-\gamma_{1}-\gamma_{2}}\delta_{y,\gamma_{2}-\gamma_{1}}
 \end{align*}
as claimed.
\end{proof}

\begin{lemma}
\label{lem:psiA_P2}
Let $|\psi_{A}(t)\rangle$ be as in equation \eq{psiA} with $\epsilon=\frac{a}{\sqrt{L}}$ (for some constant $a$), and let $P_{2}$ be as in equation \eq{P1P2}. Then
\[
\left|\langle\psi_{A}(t)|P_{2}|\psi_{A}(t)\rangle\right|=\O\left(\frac{\log^2 L}{L}\right).
\]
\end{lemma}

\begin{proof}
First note that 
\begin{align*}
\left|\langle\psi_{A}(t)|P_{2}|\psi_{A}(t)\rangle\right| & \leq\int_{D_{\epsilon}}\frac{d\phi_{1}d\tilde{\phi}_{1}}{4\pi^{2}}\int_{D_{\epsilon}}\frac{d\phi_{2}d\tilde{\phi}_{2}}{4\pi^{2}} \left|A(\phi_{1},\phi_{2})^{\ast}A(\tilde{\phi}_{1},\tilde{\phi}_{2})\right| \\
& \qquad \left|_{\pm}\langle\mathrm{sc}(\tfrac{\pi}{4}+\phi_{1};\tfrac{3\pi}{8}+\phi_{2})|P_{2}|\mathrm{sc}(\tfrac{\pi}{4}+\tilde{\phi}_{1};\tfrac{3\pi}{8}+\tilde{\phi}_{2})\rangle_{\pm}\right|.
\end{align*}
Now 
\begin{align*}
_{\pm}\langle\mathrm{sc}(\tfrac{\pi}{4}+\phi_{1};\tfrac{3\pi}{8}+\phi_{2})|P_2|\mathrm{sc}(\tfrac{\pi}{4}+\tilde{\phi}_{1};\tfrac{3\pi}{8}+\tilde{\phi}_{2})\rangle_{\pm} & = {}_{\pm}\langle\psi(\tfrac{\pi}{4}+\phi_{1};\tfrac{3\pi}{8}+\phi_{2})|J|\psi(\tfrac{\pi}{4}+\tilde{\phi}_{1};\tfrac{3\pi}{8}+\tilde{\phi}_{2})\rangle_{\pm}
\end{align*}
 where
\[
|\psi(\tfrac{\pi}{4}+\phi_{1};\tfrac{3\pi}{8}+\phi_{2})\rangle_{\pm}=\frac{1}{\sqrt{2}}\left(|\psi(\tfrac{\pi}{4}+\phi_{1};\tfrac{3\pi}{8}+\phi_{2})\rangle\pm|\psi(\tfrac{\pi}{4}+\phi_{1};-\tfrac{3\pi}{8}-\phi_{2})\rangle\right),
\]
$|\psi(p_1;p_2)\rangle$ is defined in equation \eq{psip1p2}, and 
\begin{align*}
J & = \sum_{s\text{ even }}e^{i\left(\phi_{1}-\tilde{\phi}_{1}\right)\frac{s}{2}}\sum_{\text{$r$ even, $|r|\!<\!C$}}|r\rangle\langle r|+\sum_{\text{$s$ odd}} e^{i\left(\phi_{1}-\tilde{\phi}_{1}\right)\frac{s}{2}} \sum_{\text{$r$ odd, $|r|\!<\!C$}}|r\rangle\langle r|\\
 & = 2\pi\delta(\phi_{1}-\tilde{\phi}_{1})\sum_{|r|<C}|r\rangle\langle r|.
\end{align*}
Define 
\[
g(\phi_{1},\phi_{2},\tilde{\phi}_{2})=\sum_{|r|<C}{}_{\pm}\langle\psi(\tfrac{\pi}{4}+\phi_{1};\tfrac{3\pi}{8}+\phi_{2})|r\rangle\langle r|\psi(\tfrac{\pi}{4}+\phi_{1};\tfrac{3\pi}{8}+\tilde{\phi}_{2})\rangle_{\pm}.
\]
The function $g$ satisfies $|g(\phi_{1},\phi_{2},\tilde{\phi}_{2})|\leq4C\lambda(\phi_{1})^{2}$, where for each $\phi_{1}$, $\lambda(\phi_{1})$ is the constant
obtained by applying \lem{defn_scatteringstates} to the
single-particle Hamiltonian 
\[
H_{r}^{(1)}+\left(2\cos(\tfrac{\pi}{8}+\tfrac{\phi_{1}}{2})\right)^{-1}\sum_{r\in\mathbb{Z}}\mathcal{V}(|r|)|r\rangle\langle r|.
\]
Let $\rho=4C\max_{[-\epsilon,\epsilon]}\lambda(\phi_{1})^{2}$.
Then 
\begin{align*}
\left|\langle\psi_{A}(t)|P_{2}|\psi_{A}(t)\rangle\right| & \leq  \int_{D_{\epsilon}}\frac{d\phi_{1}d\phi_{2}}{4\pi^{2}}\int_{-\epsilon}^{\epsilon}\frac{d\tilde{\phi}_{2}}{2\pi}\left|A(\phi_{1},\phi_{2})^{\ast}A(\phi_{1},\tilde{\phi}_{2})\right| \rho\\
 & = \int_{D_{\epsilon}} \frac{d\phi_{1}d\phi_{2}}{4\pi^{2}}\frac{\rho}{L}\left|\frac{\sin(\frac{L}{2}(\frac{1}{2}\phi_{1}+\phi_{2}))}{\sin(\frac{1}{2}(\frac{1}{2}\phi_{1}+\phi_{2}))} \frac{\sin(\frac{L}{2}(\frac{1}{2}\phi_{1}-\phi_{2}))}{\sin(\frac{1}{2}(\frac{1}{2}\phi_{1}-\phi_{2}))}\right|\\
 & \quad  \int_{-\epsilon}^{\epsilon}\frac{d\tilde{\phi}_{2}}{2\pi L}\left|\frac{\sin(\frac{L}{2}(\frac{1}{2}\phi_{1}+\tilde{\phi}_{2}))}{\sin(\frac{1}{2}(\frac{1}{2}\phi_{1}+\tilde{\phi}_{2}))} \frac{\sin(\frac{L}{2}(\frac{1}{2}\phi_{1}-\tilde{\phi}_{2}))}{\sin(\frac{1}{2}(\frac{1}{2}\phi_{1}-\tilde{\phi}_{2}))}\right| \\
 & \leq  \rho\int_{D_{\epsilon}} \frac{d\phi_{1}d\phi_{2}}{4\pi^{2}L}\left|\frac{\sin(\frac{L}{2}(\frac{1}{2}\phi_{1}+\phi_{2}))}{\sin(\frac{1}{2}(\frac{1}{2}\phi_{1}+\phi_{2}))} \frac{\sin(\frac{L}{2}(\frac{1}{2}\phi_{1}-\phi_{2}))}{\sin(\frac{1}{2}(\frac{1}{2}\phi_{1}-\phi_{2}))}\right|\\
 & \quad \int_{-\epsilon}^{\epsilon}\frac{d\tilde{\phi}_{2}}{4\pi L}\left(\left|\frac{\sin(\frac{L}{2}(\frac{1}{2}\phi_{1}+\tilde{\phi}_{2}))}{\sin(\frac{1}{2}(\frac{1}{2}\phi_{1}+\tilde{\phi}_{2}))}\right|^{2}+\left|\frac{\sin(\frac{L}{2}(\frac{1}{2}\phi_{1}-\tilde{\phi}_{2}))}{\sin(\frac{1}{2}(\frac{1}{2}\phi_{1}-\tilde{\phi}_{2}))}\right|^{2}\right).
\end{align*}
 Now for each $\phi_{1}\in[-\epsilon,\epsilon]$,
 \begin{align*}
\int_{-\epsilon}^{\epsilon}\frac{d\tilde{\phi}_{2}}{2\pi L}\left|\frac{\sin(\frac{L}{2}(\frac{1}{2}\phi_{1}\pm\tilde{\phi}_{2}))}{\sin(\frac{1}{2}(\frac{1}{2}\phi_{1}\pm\tilde{\phi}_{2}))}\right|^{2} & =  \int_{-\epsilon\pm\frac{1}{2}\phi_{1}}^{\epsilon\pm\frac{1}{2}\phi_{1}}\frac{d\tilde{\phi}_{2}}{2\pi L}\left|\frac{\sin(\frac{L}{2}\tilde{\phi}_{2})}{\sin(\frac{1}{2}\tilde{\phi}_{2})}\right|^{2}\\
 & \leq  \int_{-\frac{3}{2}\epsilon}^{\frac{3}{2}\epsilon}\frac{d\tilde{\phi}_{2}}{2\pi L}\left|\frac{\sin(\frac{L}{2}\tilde{\phi}_{2})}{\sin(\frac{1}{2}\tilde{\phi}_{2})}\right|^{2}\\
 & \leq  \int_{-\pi}^{\pi}\frac{d\tilde{\phi}_{2}}{2\pi L}\left|\frac{\sin(\frac{L}{2}\tilde{\phi}_{2})}{\sin(\frac{1}{2}\tilde{\phi}_{2})}\right|^{2}\\
 & = 1,
\end{align*}
so 
\begin{equation}
\label{eq:psiA_Q_P2}
\left|\langle\psi_{A}(t)|P_{2}|\psi_{A}(t)\rangle\right|\leq\rho\int_{D_{\epsilon}}\frac{d\phi_{1}d\phi_{2}}{4\pi^{2}L}\left|\frac{\sin(\frac{L}{2}(\frac{1}{2}\phi_{1}+\phi_{2}))}{\sin(\frac{1}{2}(\frac{1}{2}\phi_{1}+\phi_{2}))} \frac{\sin(\frac{L}{2}(\frac{1}{2}\phi_{1}-\phi_{2}))}{\sin(\frac{1}{2}(\frac{1}{2}\phi_{1}-\phi_{2}))}\right|.
\end{equation}
Letting $\alpha_{1}=\frac{1}{2}\phi_{1}+\phi_{2}$ and $\alpha_{2}=\frac{1}{2}\phi_{1}-\phi_{2}$, we get 
\begin{align*}
&\frac{1}{L}\int_{D_{\epsilon}}\frac{d\phi_{1}d\phi_{2}}{4\pi^{2}}\left|\frac{\sin(\frac{L}{2}(\frac{1}{2}\phi_{1}+\phi_{2}))}{\sin(\frac{1}{2}(\frac{1}{2}\phi_{1}+\phi_{2}))} \frac{\sin(\frac{L}{2}(\frac{1}{2}\phi_{1}-\phi_{2}))}{\sin(\frac{1}{2}(\frac{1}{2}\phi_{1}-\phi_{2}))}\right| \\
&\quad \leq \frac{1}{L} \int_{-\frac{3}{2}\epsilon}^{\frac{3}{2}\epsilon}\frac{d\alpha_{1}}{2\pi}\left|\frac{\sin\left(\frac{L}{2}\alpha_{1}\right)}{\sin\left(\frac{1}{2}\alpha_{1}\right)}\right| \int_{-\frac{3}{2}\epsilon}^{\frac{3}{2}\epsilon}\frac{d\alpha_{2}}{2\pi}\left|\frac{\sin\left(\frac{L}{2}\alpha_{2}\right)}{\sin\left(\frac{1}{2}\alpha_{2}\right)}\right| \\
&\quad \leq \left(\frac{d\sqrt{L}}{\pi}+\frac{\log \frac{1}{d}}{\sqrt{L}}\right)^{2} ~\text{for $d\in(0,\tfrac{3\epsilon}{2}]$}
\end{align*}
where in the last line we used equation \eq{eqn_d}. Setting $d=\Theta({1}/{L})$ and using this bound in \eq{psiA_Q_P2} gives the desired result.
\end{proof}

\section{Truncation lemma}\label{sec:Truncation-Lemma}

To prove \lem{trunc} we use the following two propositions:
\begin{prop}
\label{pro:trunc_prop}Let $H$ be a Hamiltonian acting on a Hilbert
space $\H$, and let $\ket{\Phi}\in\H$ be a normalized state. Let
$\K$ be a subspace of $\H$ such that there exists an $N_0\in\natural$
so that for all $\ket{\alpha}\in \K^{\perp}$ and for all $n\in\{0,1,2,\ldots, N_0-1\}$, $\bra{\alpha}H^{n}\ket{\Phi}=0$.
Let $P$ be the projector onto $\K$ and let $\tilde{H}=PHP$ be
the Hamiltonian within this subspace.  Then
\[
\norm{e^{-it\tilde{H}}\ket{\Phi}-e^{-itH}\ket{\Phi}} \leq 2\left(\frac{e\norm{H}t}{N_0}\right)^{N_0}.
\]
\end{prop}

\begin{proof}
Define $\ket{\Phi(t)}$ and $\ket{\tilde\Phi(t)}$
as
\[
\ket{\Phi(t)}=e^{-itH}\ket{\Phi}=\sum_{k=0}^{\infty}\f{(-it)^{k}}{k!}H^{k}\ket{\Phi}\qquad\ket{\tilde\Phi(t)}=e^{-it\tilde{H}}\ket{\Phi}=\sum_{k=0}^{\infty}\f{(-it)^{k}}{k!}\tilde{H}^{k}\ket{\Phi}.
\]
Note that by assumption, $\tilde{H}^{k}\ket{\Phi}=H^{k}\ket{\Phi}$ for all $k< N_0$, and thus the first $N_0$ terms in the two above sums are equal. Looking at the difference between these two states, we have
\begin{align*}
\norm{\ket{\Phi(t)}-\ket{\tilde\Phi(t)}} & =\Norm{\sum_{k=0}^{\infty}\f{(-it)^{k}}{k!}\left(H^{k}-\tilde{H}^{k}\right)\ket{\Phi}}\\
 & =\Norm{\sum_{k=0}^{N_0-1}\f{(-it)^{k}}{k!}\left(H^{k}-\tilde{H}^{k}\right)\ket{\Phi}-\sum_{k=N_0}^{\infty}\f{(-it)^{k}}{k!}\left(H^{k}-\tilde{H}^{k}\right)\ket{\Phi}}\\
 & \leq\sum_{k=N_0}^{\infty}\f{t^{k}}{k!}\left(\norm{H}^{k}+\norm{\tilde{H}}^{k}\right) \\
 & \leq 2\sum_{k=N_0}^{\infty}\f{t^{k}}{k!} \norm{H}^{k}
\end{align*}
where the last step uses the fact that $\norm{\tilde{H}}\leq\norm{P}\norm{H}\norm{P} = \norm{H}$.  Thus for any $c \ge 1$, we have
\begin{align*}
\norm{\ket{\Phi(t)}-\ket{\tilde\Phi(t)}}
 & \leq\f{2}{c^{N_0}}\sum_{k=N_0}^{\infty}\f{(ct)^{k}}{k!}\norm{H}^{k}\\
 & \leq\f{2}{c^{N_0}}\exp(ct\norm{H}).
\end{align*}
We obtain the best bound by choosing $c=N_0/\norm{Ht}$, which gives
\[
  \norm{\ket{\Phi(t)}-\ket{\tilde\Phi(t)}}
  \le 2\left(\frac{e\norm{H}t}{N_0}\right)^{N_0}
\]
as claimed.  (If $c < 1$ then the bound is trivial.)
\end{proof}

\begin{prop}
\label{pro:hybrid}Let $U_1,\ldots, U_n$ and $V_1,\ldots, V_n$  be unitary operators.  Then for any $\ket{\psi}$,
\begin{equation}
  \Norm{\left(\prod_{i=n}^1 U_i - \prod_{i=n}^1 V_i \right) \ket{\psi}}   \le \sum_{j=1}^n \Bigg\|{(U_j - V_j)\prod_{i=j-1}^1 U_i \ket{\psi}}\Bigg\|.
\end{equation}
\end{prop}

\begin{proof}
The proof is by induction on $n$.  The case $n=1$ is obvious.  For the induction step, we have \begin{align}   \Norm{\left(\prod_{i=n}^1 U_i - \prod_{i=n}^1 V_i \right) \ket{\psi}}   &= \Norm{\left(\prod_{i=n}^1 U_i - V_n \prod_{i=n-1}^1 U_i             + V_n \prod_{i=n-1}^1 U_i - \prod_{i=n}^1 V_i \right) \ket{\psi}} \\   &\le \Norm{(U_n - V_n) \prod_{i=n-1}^1 U_i \ket{\psi}}       +\Norm{\left(\prod_{i=n-1}^1 U_i - \prod_{i=n-1}^1 V_i \right) \ket{\psi}} \\   &\le \sum_{j=1}^n \Norm{(U_j - V_j)\prod_{i=j-1}^1 U_i \ket{\psi}} \end{align} where the last step uses the induction hypothesis. 
\end{proof}

We are now ready to prove \lem{trunc}:

\begin{proof}[Proof of \lem{trunc}]
For $M \in \natural$ write
\begin{align*}
  \norm{(e^{-iHt}-e^{-i\tilde{H}t}) \ket{\Phi}}
  &= \Norm{\left(\left(e^{-iH\frac{t}{M}}\right)^{M} -
     \left(e^{-i\tilde{H}\frac{t}{M}}\right)^{M}\right) \ket{\Phi}} \\
  & \leq \sum_{j=1}^{M} \Norm{\left(e^{-iH\frac{t}{M}} - 
     e^{-i\tilde{H}\frac{t}{M}}\right) e^{-iW\left(j-1\right)\frac{t}{M}}
     \ket{\Phi}} \\
  & \leq \sum_{j=1}^{M} 
    \Norm{\left(e^{-iH\frac{t}{M}} - e^{-i\tilde{H}\frac{t}{M}}\right)
    \left(\ket{\gamma(\tfrac{(j-1)t}{M})} +
          \ket{\epsilon(\tfrac{(j-1)t}{M})}\right)} \\
  & \leq 2M\delta+\sum_{j=1}^{M} 
    \Norm{ \left(e^{-iH\frac{t}{M}} - e^{-i\tilde{H}\frac{t}{M}}\right)
    \frac{\ket{\gamma(\tfrac{(j-1)t}{M})}} 
         {\Norm{\ket{\gamma(\tfrac{(j-1)t}{M})}}}}
    \Norm{\ket{\gamma(\tfrac{(j-1)t}{M})}}\\ 
  & \leq 2M\delta 
    + 2M\left(\frac{e\norm{H}t}{M N_0}\right)^{N_0}(1+\delta)
\end{align*}
where in the second line we have used \pro{hybrid} and in the last step we have used \pro{trunc_prop} and the fact that $\Vert |\gamma(t)\rangle\Vert \leq 1+\delta$.
Now, for some $\eta>1$, choose
\[
  M= \left\lceil \frac{\eta e\norm{H}t}{N_0} \right\rceil
\]
for $0<t\leq T$ to get
\begin{align*}
  \norm{(e^{-iHt}-e^{-i\tilde{H}t}) \ket{\Phi}}
  &\leq 2M\left(\delta + \eta^{-N_0}(1+\delta)\right) \\
  &\leq 2\left(\frac{\eta e\norm{H}t}{N_0} + 1 \right) 
        \left(\delta + \eta^{-N_0}(1+\delta)\right).
\end{align*}
The choice $\eta=2$ gives the stated conclusion.
\end{proof}

Note that it would be slightly better to take a smaller value of $\eta$.  However, this does not significantly improve the final result; the above bound is simpler and sufficient for our purposes.

\end{document}